\documentclass[12pt]{article}

\newcommand{\blind}{1}

\usepackage{enumitem}
\usepackage[toc,page]{appendix}

\RequirePackage{amsthm,amsmath,amsfonts,amssymb,statmath}
\RequirePackage[authoryear]{natbib}
\RequirePackage[colorlinks,citecolor=blue,urlcolor=blue]{hyperref}
\RequirePackage{graphicx}
\usepackage{sidecap}
\usepackage{booktabs}
\usepackage{longtable}
\usepackage{multirow}
\usepackage{float}
\usepackage{mathtools}
\usepackage{color}
\usepackage{statmath}
\usepackage{xfrac}
\usepackage{yhmath}
\usepackage{bigints}
\usepackage{subfigure}
\usepackage{bbm}
\usepackage{array}
\usepackage{booktabs}
\usepackage{graphicx}
\usepackage{subfigure}
\usepackage{siunitx, tabularx}
\usepackage{adjustbox}
\usepackage{xr}
\usepackage{arydshln,leftidx}
\usepackage{verbatim}
\usepackage{upgreek}
\usepackage{algorithm,algpseudocode}
\usepackage{amssymb}
\usepackage{epstopdf}
\usepackage{bm}
\usepackage{makecell}
\usepackage{bigints}
\usepackage{threeparttable}

\newcommand{\RD}{\operatorname{Diff}}
\newcommand{\Pivot}{\text{\normalfont Pivot}}

\newcommand{\op}{\mathrm{op}}

\newcommand{\nbracket}[1]{\left( #1 \right)}
\newcommand{\cbracket}[1]{\left\{ #1 \right\}}
\newcommand{\rbracket}[1]{\left[ #1 \right]}

\newtheorem{proposition}{Proposition}[section]
\newtheorem{theorem}{Theorem}[section]
\newtheorem{definition}{Definition}[section]
\newtheorem{lemma}[theorem]{Lemma}
\newtheorem{corollary}{Corollary}[section]
\newtheorem{remark}{Remark}

\newtheorem{assumption}{Assumption}

\newcommand{\Var}{\operatorname{Var}}

\newcommand{\one}{\mathbbm{1}}

\def\argmax{\mathop{\rm argmax}\limits}

\newcommand{\real}{\mathbb{R}}
\newcommand{\Pp}{{\mathbb{P}}}
\newcommand{\Ee}{{\mathbb{E}}}
\newcommand{\cP}{{\mathcal{P}}}
\newcommand{\cD}{{\mathcal{D}}}   % data

\newcommand{\cS}{{\mathcal{S}}}

\newcommand{\cK}{{\mathcal{K}}}

\newcommand{\Ek}{\mathcal{E}_{k}}      % collection of k-subsets
\newcommand{\Ehat}{\widehat{E}}
\newcommand{\Eone}{\widehat{E}^*}  % empirical argmax subset
\newcommand{\KL}{\operatorname{KL}}

\newcommand{\Estar}{\widehat{W}^*}

\newcommand{\snoise}{\sigma_{s}}

\definecolor{myorange}{RGB}{230,120,20}

\begin{document}
%%---------------------------------------------------------------------------

\date{}

\def\spacingset#1{\renewcommand{\baselinestretch}%
{#1}\small\normalsize} \spacingset{1}

\if1\blind
{
  \title{\bf Flexible Inference for Winners with Conditional Validity}
\author{
 Soham Bakshi\,$^{1}$\thanks{Corresponding author: \texttt{bakso@umich.edu}.}
 \quad Lingjun Gao\,$^{2}$
 \quad Zijun Gao\,$^{3}$
 \quad Snigdha Panigrahi\,$^{1}$
 \\[6pt]
 {\small $^{1}$Department of Statistics, University of Michigan, MI, USA}\\
 {\small $^{2}$Department of Statistics, The Pennsylvania State University, PA, USA}\\
 {\small $^{3}$Department of Data Science and Operations, University of Southern California, CA, USA}
}
  \maketitle
} \fi

\if0\blind
{
  \bigskip
  \bigskip
  \bigskip
  \begin{center}
    {\LARGE\bf Inference after Top-$k$ Selection}
\end{center}
  \medskip
}\fi

\begin{abstract}
Researchers often select top-performing options---such as treatments, models, or model features---based on a data-driven criterion and then report effect estimates for the selected winners. 
Naive post-selection estimates, however, are known to suffer from the winner’s curse, producing systematically overoptimistic results.
We introduce a flexible conditional inference method that corrects for this overoptimism through an adaptive exponential randomization scheme.
Our method achieves selection quality that closely matches that of standard top-$k$ selection, while also yielding shorter confidence intervals than existing approaches.
Furthermore, our approach applies broadly to nonparametric settings with asymptotically linear selection statistics, covering wide-ranging applications such as inference for the efficacy of the most promising treatments in clinical trials, the abilities of top-ranked models on leaderboards, and the importance of the most predictive features in a model.
\end{abstract}

\noindent{\it keywords:} conditional inference, exponential mechanism, nonparametric inference, post-selection inference, winners, winner's curse

\spacingset{1.5}

%%---------------------------------------------------------------------------
\section{Introduction}
\label{sec:intro}
%%---------------------------------------------------------------------------

In many empirical workflows, researchers first select winners or top-ranked options according to a data-driven criterion, and then focus subsequent analysis only on those that appear most promising.
Examples include selecting treatments with the largest estimated effects, models with the best benchmark performance, or model features with the highest importance scores. 
In such settings, the object of interest is often not a single best candidate, but rather a set of top-$k$ winners pursued for further investigation.
A natural question then arises: \emph{How large are the true effects of the selected winners?}

Because the same data are used to select winners and estimate their effects, naive post-selection estimates tend to be overoptimistic. 
For example, classical intervals that fail to account for selection typically have below-nominal coverage.
This is popularly known as the \textit{winner's
curse}, and as a consequence, selected winners often perform worse in subsequent experiments or real-world deployments than their initial estimates suggest. 
This phenomenon is well documented across applications, including economics \citep{petrou2026inference}, machine learning \citep{lee2018winner}, and biomedical research \citep{zhong2008bias, ferguson2013empirical, forde2023review, guo2021inference}, and has motivated a growing body of
post-selection inference methods for correcting selection bias.

A large part of this literature develops simultaneous inference with marginal coverage guarantees; see, for example, recent work by \citet{andrews2022inference,andrews2024inference, zrnic2024locally,
zrnic2025flexible, petrou2026inference}. 
Such methods guarantee inferential validity only on average
over all possible winners that may be produced by the selection process.

By contrast, our work focuses on conditional inference: interval estimates with
conditional coverage guarantees for the effects of the winners
\emph{actually selected} from the observed data. 
Despite its conceptual appeal, existing approaches to conditional inference are often either restricted to specific parametric models or limited by low power, and sometimes both.
Examples include the polyhedral method \citep{Lee_2016}, which bases inference on a univariate truncated normal distribution, and data splitting and its variants, such as data thinning and data fission \citep{rasines2023splitting, dharamshi2025generalized, leiner2025data}, which reserve a holdout portion of the data for inference.
In particular, \citet{reid2017post} adapt the polyhedral method to top-$k$ winners, but the approach relies on exact normality and is tailored to specific selection criteria.
Moreover, the resulting intervals can be wide, sometimes infinitely long and unstable, when there are close competitors.

A refinement of the polyhedral approach, the hybrid method of \citet{andrews2024inference}, later generalized by \citet{mccloskey2024hybrid},  addresses the overly wide intervals produced by the polyhedral method by combining it with a simultaneous interval.
Taking a different approach, \citet{hoff2025selective} address this problem by constructing empirical Bayes estimates.
However, in seeking to restore power, both types of approaches either do not target conditional coverage or do not provide an exact control of it.

Next, we introduce the problem setup, discuss the two types of coverage guarantees, and then present a data example that highlights the contributions of our inferential method.

\subsection{Problem setup}
\label{sec:setup}

Consider selecting the $k$ most promising candidates from $p\geq k$ options,
indexed by $[p]=\{1,\ldots,p\}$, using selection statistics
$T=(T_1,\ldots,T_p)^\top$, where each $T_j$ is a noisy estimate of the true
effect $\mu_j\in\mathbb{R}$. 

A selection rule maps $T$ to a set $\widehat E(T)\subseteq[p]$ of $k$ winners.
We write this random selected set simply as $\widehat E$, and denote its realized value on the observed data $T=t$ by $E_o$.
A canonical example is the top-$k$ rule, which sets $\widehat E$ to the indices
of the $k$ largest entries of $T$. 
In our approach, $\widehat E$ is generated by a randomized selection rule based on a data-adaptive exponential mechanism, which we define in Section \ref{sec:exp:mechanism}.

Post selection, we seek inference for the effects of the selected winners
\begin{equation}
\label{eqn: target:inference}
\mu_{\widehat{E}}= \bigl(\mu_j: j \in\widehat E\bigr)^\top.
\end{equation}
Specifically, after observing $\{\widehat{E}=E_o\}$, we wish to infer about the parameter vector $\mu_{E_o}$.

\subsection{Coverage guarantees}
\label{sec:guarantees}

\paragraph{Conditional guarantees.} For any realized set of winners $E_o\subseteq [p]$, the $100\times(1-\alpha)\%$ intervals $\{C_j^{E_o}(T): j\in E_o\}$ produced by our method satisfy:
\begin{equation}
\label{eq:cond_cov}
 \text{\textit{Per-winner conditional coverage:}}  \ \     \mathbb P
    \left[
        \mu_j\in C_j^{E_o}(T)
        \,\middle|\,
        \widehat E=E_o
    \right]
    \geq 1-\alpha \ \text{ for all } j\in E_o.
\end{equation} 
Equivalently, on the event $\{\widehat E=E_o\}$, our method provides an interval estimate $C_j^{E_o}(T)$ for the effect of each winner $j\in E_o$ actually selected from the observed data, addressing the question: how large is $\mu_j$ for $j\in E_o$?

The guarantee \eqref{eq:cond_cov} in turn implies conditional coverage averaged over the
selected set:
\begin{equation}
\label{eq:cond_cov_avg}
\text{\textit{Averaged conditional coverage:}}\quad
\mathbb{E}\!\left[\frac{1}{k}\sum_{j\in\widehat E}
\mathbbm{1}\{\mu_j\in C_{j}^{\widehat{E}}(T)\}\,\middle|\,\widehat E=E_o\right]
\geq 1-\alpha,
\end{equation}
which ensures inferential validity across the full selected set rather than for each winner individually, and reduces to \eqref{eq:cond_cov} when $k=1$. 
Finally, applying a multiplicity correction to our intervals, e.g., by forming per-winner intervals at level $\alpha/k$, yields simultaneous coverage for the effects of the selected winners $\mu_{E_o}$; we do not pursue this direction here.

\paragraph{Marginal guarantees.} 
By contrast, marginal (or unconditional) inferential guarantees average over the randomness of the selected set $\widehat{E}$. 
For example, the intervals $\{M_{j}(T): j\in \widehat{E}\}$ satisfy a marginal coverage guarantee if
\begin{equation}
    \label{eq:marg_cov_avg}
 \text{\textit{Averaged marginal coverage:}} \ \ \mathbb{E}
\left[
\dfrac{1}{k}
\;\sum_{j\in\widehat E}
\mathbbm{1}\{\mu_j\in M_{j}(T)\}
\right] \geq 1-\alpha.
\end{equation}
One minus this averaged marginal coverage is the well-known false coverage rate (FCR), introduced by \cite{benjamini2005false}. 
Simultaneous methods achieve \eqref{eq:marg_cov_avg} by constructing a confidence
region $\mathcal M(T)$ for $\mu_{\widehat E}$ such that
$\mathbb P[\mu_{\widehat E}\in\mathcal M(T)]\geq 1-\alpha$.

\paragraph{Choosing the appropriate guarantee.} 
The distinction between conditional and marginal guarantees is important, and the appropriate choice depends on the underlying inferential question. 
For example, assessing the efficacy of drug A, because it is selected as a winner in a clinical trial, requires a conditional guarantee, whereas assessing whether the efficacy of the winning treatment has changed, regardless of its identity, may only require a marginal guarantee.
Crucially, interval estimates with marginal guarantees cannot answer the former type of inferential questions, whereas intervals with conditional guarantees can address both types of questions, as clarified in the next remark.

\begin{remark}
Conditional guarantees are strictly stronger than their marginal counterparts. As an example, by the law of total probability, \eqref{eq:cond_cov_avg} implies \eqref{eq:marg_cov_avg}, but not conversely. 
In particular, neither marginal coverage nor simultaneous marginal coverage control provides any type of conditional inferential guarantee for the selected winners in $E_o$.
Our first data example, presented next, illustrates this distinction.
For a related discussion from a hypothesis-testing perspective, see also \cite{neufeld2026inference}.
\label{rem:distinction}
\end{remark}

\subsection{Contributions and other related work}

Our work introduces a data-adaptive randomization scheme based on the exponential mechanism that enables conditional inference for winners without requiring an analytical characterization of the selection event.
The proposed method delivers valid conditional inference for the selected winners without sacrificing selection quality, and produces short, numerically stable intervals that remain valid beyond normal data.
We demonstrate this breadth in three nonparametric examples: inference on the efficacy of the most promising treatments in A/B/n trials with binary outcomes, the abilities of winners ranked using Bradley--Terry--Davidson comparisons, and the importance of the most predictive features, potentially based on a black-box fit; see Section~\ref{sec:examples}.

Figure~\ref{fig:example1} illustrates how our method compares with existing approaches in a first example with Gaussian data; full details are deferred to Appendix~\ref{app:sim-details}.

\begin{figure}[t]
  \centering
  \includegraphics[width=\textwidth]{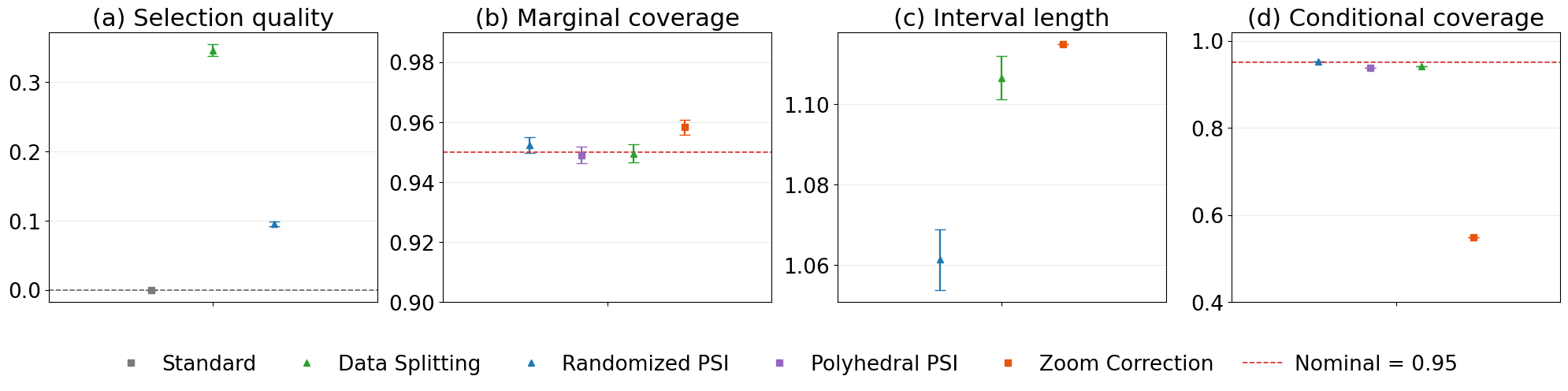}
  \vspace{-1cm}
  \caption{Comparison of inferential methods: (a) standardized regret, with zero marked by a horizontal line; (b) marginal coverage rate; (c) average interval length; and (d) conditional coverage rate, with the nominal level $1-\alpha$ marked by horizontal lines in panels~(b) and~(d). The intervals produced by ``Polyhedral PSI''  are \textcolor{red}{496.2\%} longer than those produced by the proposed ``Randomized PSI'' approach, and are therefore omitted from panel (c) to avoid obscuring comparisons among the other methods.}
  \label{fig:example1}
\end{figure}

\begin{enumerate}[leftmargin=*]
\item \emph{Selection quality.} 
Panel~(a) reports regret, a measure of selection quality defined in Section \ref{sec:exp:mechanism}. 
Regret is measured relative to the standard top-$k$ rule, which uses the full data without additional randomization and has regret zero; values closer to zero therefore indicate better selection quality.
As shown in the plot, our method (``Randomized PSI'') attains regret close to zero, selecting nearly as well as the top-$k$ rule, whereas ``Data Splitting'', which uses only half the sample, has a much larger regret, corresponding to greater deviation from the selection rule applied to the original full dataset.
Note that the other two methods considered use the standard top-$k$ rule for selection.

\item \emph{Coverage.}
Panels~(b) and~(d) report marginal and conditional coverage rates, respectively. 
Our method attains nominal coverage in both senses, as do the conditional baselines based on ``Data Splitting'' and ``Polyhedral PSI'' \citep{reid2017post}. 
In contrast, the ``Zoom Correction'' method of \cite{zrnic2025flexible}, included as a marginal baseline, fails to cover the selected winner conditional on its selection, consistent with the distinction between the two types of guarantees discussed earlier.

\item \emph{Interval length.} Panel~(c) shows that the intervals based on the proposed ``Randomized PSI'' are substantially shorter than those produced by ``Polyhedral PSI''. 
Although our intervals are closer in length to those produced by ``Data Splitting'', the latter achieves this only at the cost of substantially lower selection quality (see Panel~(a)).

Furthermore, it can be seen that even when only marginal coverage may be sufficient, our method, which targets the stronger conditional guarantee, still yields shorter intervals than the ``Zoom Correction'' approach to simultaneous inference.
\end{enumerate}

The exponential mechanism was introduced by \citet{bakshi2025classification} for inference in decision trees and developed for clustering problems by \citet{wu2025hierarchical}.
However, these papers applied the mechanism using a prespecified level of randomization and established validity only under specific models.
The present work relaxes both requirements to provide a broadly applicable approach:

\noindent{\textbf{How much should one randomize?}} \quad  We provide a tuning-free calibration of the randomization level, based on a notion of regret, that enables powerful inference while keeping the randomized selection close to the standard top-$k$ rule.
The regret budget $q$ is directly interpretable and requires no tuning: the practitioner specifies the acceptable loss in selection quality, and the corresponding temperature is determined automatically.
Prior work on randomized conditional inference \citep{randomizedresponse, kivaranovic2019length, panigrahi2021selection, panigrahi2023approximate, panigrahi2024exact} has highlighted the power gains obtained by increasing randomization. 
Yet these gains can come at the expense of selection quality, just as allocating more data to inference under data splitting can improve inferential power but worsen the quality of selection.
Carefully choosing the amount of randomization is therefore essential to improving inference without compromising the quality of selection.
%This consideration is especially important in real-world deployments, where poor selection quality can directly affect the choice of questions or targets pursued for subsequent inference.

\noindent{\textbf{How flexible is our post-selection inference method?}} \quad Central to the broad applicability of our approach, we establish valid inference in nonparametric settings with asymptotically linear selection statistics.
Unlike earlier randomized approaches for asymptotic inference
\citep{randomizedresponse, rasines2023splitting, perry2026post}, and approaches
developed under more relaxed assumptions \citep{panigrahi2023carving,
bakshi2024selective}, which impose either direct or indirect restrictions on the probability of the selection event, our theory places no such restrictions.

Finally, we note that \citet{jinglei2024winners} employ the exponential mechanism to draw inference about the population minimum of a vector from noisy observations.
Our goal is different: we seek inference on the effects of units that appear small or large in the observed data, which might not correspond to the population minimum or maximum.

\section{Randomized selection of winners}
\label{sec:selection}

Throughout, we consider the selection of the top-$k$ winners from $\Ek = \{E \subseteq \rbracket{p} : |E| = k\}$, the collection of all $k$-subsets of $\rbracket{p}$, with $|\Ek| = \binom{p}{k}$, based on the selection statistics $T = (T_1,\ldots,T_p)^\top \in \mathbb{R}^p$.

To identify these winners, each candidate set $E \in \Ek$ is assigned a score $s_E(T)$, where $s_E: \real^p \to \real$ is a criterion tailored to the problem and evaluated at the observed data.
We assume this score admits the additive decomposition
$$
s_E(T) =\sum_{j\in E}s_j(T).
$$
A simple yet common example is $s_E(T) = \sum_{j \in E} T_j$, for $E \in \Ek$, which sums the selection statistics indexed by the set $E$. 

The standard top-$k$ rule without external randomization selects a  subset that maximizes the empirical score 
\begin{equation}
    \Eone \in  \Estar(T) = \argmax_{E \in \Ek} \; s_E(T) 
    %= \{E \in \Ek : s_E(T)= \max_{E\in \Ek} s_E(T)\}.
    \label{eqn: standard selection} 
\end{equation} 
In the presence of ties, the top-$k$ winners subset is picked from the set of all maximizers $\Estar(T)$, either randomly or based on a tie-breaking rule. 
Instead of selecting the winners to maximize the score as in \eqref{eqn: standard selection}, we introduce a randomization scheme for selection using a data-adaptive exponential mechanism, described next.

%%---------------------------------------------------------------------------
\subsection{Randomization scheme}
\label{sec:exp:mechanism}
%%------------------------------------------------

Introducing some notations, let $\bar{s}(t)= \dfrac{1}{|\Ek|}\sum_{E \in \mathcal{E}_k} s_E(T)$ denote the mean score over all subsets in $\mathcal{E}_k$. 
For a fixed $\varepsilon>0$, let 
$$
\snoise(t ) =\snoise(t; k, \varepsilon)= \sqrt{
    \frac{1}{|\mathcal{E}_k|}
    \sum_{E \in \mathcal{E}_k}
    \bigl(s_E(t) - \bar{s}(t)\bigr)^2 + \varepsilon^2}
$$ 
represent the dispersion of the scores; $\varepsilon$, chosen to be a small constant, ensures that this quantity is well defined, even when the scores of all candidate subsets are nearly equal.

For fixed $\tau > 0$ and $\varepsilon > 0$, the random set of selected winners $\Ehat$ is sampled according to the discrete probability distribution
\begin{equation}
  \label{eq:expmech}
  \begin{aligned}
  \Pp\rbracket{\Ehat = E \;\Big|\; T=t}
    \;=\;
  \dfrac{\exp\!\left\{\dfrac{s_E(t)}{\tau \cdot \snoise (t)}\right\}}
       {\displaystyle\sum_{E' \in \Ek}
        \exp\!\left\{\dfrac{s_{E'}(t)}{\tau \cdot \snoise (t)}\right\}} = \dfrac{\exp\!\left\{\dfrac{s_E(t) - \bar{s}(t)}{\tau \cdot \snoise (t)}\right\}}
       {\displaystyle\sum_{E' \in \Ek}
        \exp\!\left\{\dfrac{s_{E'}(t)-\bar{s}(t)}{\tau \cdot \snoise (t)}\right\}},  
  \end{aligned}
\end{equation}
where $\snoise(t)$ is as defined above.
Here, the second equality follows from subtracting the constant mean score from both the numerator and the denominator of the probabilities. 
Note that a practical concern may arise in implementing this randomization scheme, especially when $|\Ek| = \binom{p}{k} $ is large. 
In Section~\ref{sec:efficient}, we present a computationally efficient implementation for sampling from the distribution in \eqref{eq:expmech}.

For notational simplicity, from this point on we let $\tau(t)= \tau \cdot \snoise (t)$.
Put differently, our randomization scheme replaces the hard maximum in \eqref{eqn: standard selection} with a softmax over $\left\{s_E(t)/\tau(t): E\in \Ek\right\}$, where $\tau$, the temperature parameter of the data-adaptive exponential mechanism, controls the overall level of randomization.

\begin{remark}
Scaling the scores by their dispersion makes $\tau$ scale-free, so that it can be applied consistently across problems whose scores may have different magnitudes or scales.
This type of data-adaptive scaling is important because, without it, the effective temperature would depend on the absolute scale of the scores; with it, the selection probabilities defined in \eqref{eq:expmech} depend only on relative differences among subset scores.
\end{remark}

\subsection{Guarantees on the quality of randomized selection}

Let $s^*(T) = \displaystyle\max_{E \in \Ek} s_E(T)$ denote the optimal score, and let $$\Delta_+(T)=s^*(T) -
  \displaystyle\max_{E\notin \widehat{W}^*(T)}
  s_E(T)$$ 
denote the gap between the optimal score and the
next-best score, with the convention that $\Delta_+(T)=+\infty$ if $\widehat{W}^*(T)=\Ek$.

Recall the standard top-$k$ rule defined in \eqref{eqn: standard selection} selects one of the score-maximizers $\Eone \in \widehat{W}^*(T)$, attaining the optimal score $s_{\Eone}(T) = s^*(T)$. 
Given the observed data $\{T=t\}$, Proposition \ref{prop:mismatch} provides an explicit probabilistic bound on the event that our exponential mechanism fails to match the score achieved by the standard top-$k$ rule, i.e. $s_{\Ehat} (T)<s^{*}(T)$.

\begin{proposition}
\label{prop:mismatch}
Let $\Ehat$ be the set of top-$k$ winners selected by the randomization scheme in \eqref{eq:expmech} with temperature $\tau(T) = \tau \cdot \snoise(T)$. 
Conditional on ${T=t}$, we have
\begin{equation*}
\label{eq:mismatch-simple-general}
\Pp\rbracket{\Ehat \notin \Estar(T)\mid T=t}
\le (|\Ek|-1) \exp\left\{-
\frac{\Delta_+(t)}{\tau(t)}
\right\}, \ \text{ for any } \ t\in \real^p.
\end{equation*}
\end{proposition}

\begin{remark}
The bound in Proposition~\ref{prop:mismatch} implies that, as $\tau\to 0$, i.e., as the temperature parameter $\tau$ decreases, our mechanism selects a winner subset in $\widehat{W}^*(T)$ with probability converging to $1$.
This type of limiting behavior was also noted in prior work introducing this exponential randomization; see Proposition 3.1 in \cite{bakshi2025classification}.
\end{remark}

Proposition~\ref{prop:regret} additionally provides a guarantee for another measure of selection quality: the expected conditional standardized regret, defined as
$$\Ee \left[\frac{1}{\snoise(T)}(s^*(T)-s_{\Ehat}(T)) \Big\lvert T=t\right].$$ 
This quantity measures, in expectation, the deviation of the score under the randomized rule from that under the standard top-$k$ rule, standardized by dispersion.
\begin{proposition}
\label{prop:regret}
Consider the randomization scheme in \eqref{eq:expmech} with temperature $\tau(T) = \tau \cdot \snoise(T)$, and let $\Ehat$ denote its output, the set of top-$k$ winners. 
Then, for any $t\in \real^p$, we have
\begin{equation*}
\label{eq:regret-bound}
\Ee\rbracket{\frac{1}{\snoise(T)} (s^*(T)-s_{\Ehat}(T))\Big\lvert T=t}
\le
\tau\log\frac{|\Ek|}{|\Estar(t)|}
\le
\tau\log|\Ek|.
\end{equation*} 
\end{proposition}

As an immediate consequence of Proposition~\ref{prop:regret}, Corollary~\ref{cor: temp par} provides a simple, tuning-free choice of the temperature parameter $\tau$, ensuring that the standardized regret of randomized winner selection is no greater than a user-specified value $q$.

\begin{corollary}[Choice of temperature parameter]
For any $q>0$, setting the temperature parameter $\tau \leq (\log |\Ek|)^{-1}q$ ensures that 
$$
\Ee\rbracket{\frac{1}{\snoise(T)} (s^*(T)-s_{\Ehat}(T))\Big\lvert T=t} \leq q, \quad \forall t\in \real^p.
$$
 \label{cor: temp par}
 \end{corollary}

%%---------------------------------------------------------------------------
\section{Inference method}
\label{sec:method}

To perform inference, we require assumptions on the selection statistics. 
We begin with the exact Gaussian model:
\begin{align}
\label{eq:gaussian_model}
T \sim \mathcal{N}_p(\mu ,\Sigma),
\end{align}
where $\mu = (\mu_1,\ldots, \mu_p)^\top \in \real^p$ are the unknown mean effect parameters.
For now, we assume that the covariance matrix $\Sigma$ is known. 
Asymptotic inferential guarantees for settings in which~\eqref{eq:gaussian_model} holds only approximately, including cases where $\Sigma$ is consistently estimable, are deferred to Section~\ref{sec:asymptotic}.

After observing the selected set $\left\{\Ehat = E_o\right\}$, recall that our goal is to provide conditionally valid inference for the effects of the selected winners, $\mu_{E_o}$. 
For the rest of this section, we fix a winner $j_o \in E_o$, and without loss of generality, we focus on inference for the scalar-valued parameter $\mu_{j_o}$. 

\subsection{Outline of the inferential workflow}

Following a standard inferential recipe, we construct a pivot, that is a function of $\mu_{j_o}$ and the selection statistics. 
Inverting this pivot, given in Theorem \ref{thm:pivot}, then directly yields interval estimates, as well as p-values for tests whose null hypotheses about $\mu_{j_o}$, as stated in Corollary \ref{cor:conditional-validity}.

To obtain conditional coverage guarantees, as discussed in Section \ref{sec:guarantees}, a seemingly natural approach is to derive the pivot from the conditional density of $T_{j_o} \mid \left\{\widehat{E}=E_o\right\}$.
However, obtaining such a pivot requires both eliminating nuisance parameters and characterizing the underlying conditional density. 
In contrast to most existing conditional approaches, both with and without randomization, our approach yields this density in closed form without ever requiring an analytical description of the selection event. 
We now outline our inferential workflow.

\noindent\textbf{Step 1: Conditioning to remove nuisance parameters.}
We first identify appropriate statistics $T_{j_o}^\perp$ such that, after further conditioning on them, the conditional density of 
\begin{equation}
\label{eqn: conditional:distribution}
    T_{j_o} \Big\lvert \left\{\widehat{E}=E_o, \; T_{j_o}^\perp\right\}.
\end{equation}
is free of nuisance parameters. 
Lemma~\ref{lem:conditionaldensity} states this result.

\noindent\textbf{Step 2: Obtaining a closed-form conditional density.} 
The key observation is that, by design of our randomization scheme, we obtain an exact expression for the conditional density in~\eqref{eqn: conditional:distribution}, as stated in Proposition~\ref{prop:cond_density}. Importantly, unlike earlier conditional approaches, the conditional selection probability $\mathbb{P}[\widehat E=E_o \mid T=t]$, which serves as the correction factor for selection, is available in closed form.
As a result, our method does not require an analytical characterization of the selection event, and is not tailored to the specific score-based criterion used in the selection rule.

\noindent\textbf{Step 3: Constructing test, interval and point estimate.}
We derive a pivot directly by applying the probability integral transform to the conditional density obtained in Step 2. This pivot yields valid tests and interval estimates. Moreover, as a byproduct, maximizing the conditional density yields a bias-corrected point estimate of $\mu_{j_o}$.

\subsection{Step 1: Conditioning to remove nuisance parameters}
\label{sec:key_stats}

To fix notation, let $R_{j_o} \in \real^{(p-1) \times p}$ denote the row-deletion matrix obtained by removing row $j_o$ from $I_p$, the $p\times p$ identity matrix, so that $R_{j_o} x = x_{-j_o}$ for any $x \in \real^p$, and 
let $e_{j_o}$ denote the vector with a $1$ in the $j_o$-th coordinate and zeros elsewhere.
Let $\phi(s; \mu, \sigma^2)$ denote the Gaussian density of a normal random variable with mean $\mu$ and variance $\sigma^2$, evaluated at $s$.

Define the population matrices
\begin{align*}
  \Sigma_{-j_o,\,j_o}
  &= R_{j_o} \Sigma\, e_{j_o}
   \in \real^{p-1},
  &
  \Sigma_{-j_o,-j_o}
  &= R_{j_o} \Sigma R_{j_o}^\top
   \in \real^{(p-1) \times (p-1)},
\end{align*}
and the statistics
\begin{equation}
  \label{eq:Tperp}
  T_{j_o}^\perp= R_{j_o} T- \Sigma_{-j_o,j_o} \sigma_{j_o}^{-2} e_{j_o}^\top T \in \real^{p-1}.
\end{equation}

In Lemma \ref{lem:reconstruction} and Lemma \ref{lem:decomp} we note that: (i) the full vector of selection statistics $T$ can be constructed from these statistics $(T_{j_o}, T_{j_o}^\perp)$, and (ii) $T_{j_o}$ and $T_{j_o}^\perp$ are independent normal variables. 
Using both (i) and (ii), Lemma~\ref{lem:conditionaldensity} shows that the conditional density in~\eqref{eqn: conditional:distribution} is free of nuisance parameters and depends only on the parameter of interest $\mu_{j_o}$.

\begin{lemma}
\label{lem:reconstruction}
It follows that $T=r(T_{j_o},T_{j_o}^{\perp})$, where
$r:\real \times \real^{p-1}\to\mathbb{R}^p$ is defined by
$r(z_1,z_2)=\Sigma e_{j_o}\sigma_{j_o}^{-2}z_1+ R_{j_o}^\top z_2$
for $z=(z_1,z_2^\top)^\top\in\mathbb{R}^p$.
\end{lemma}

The proof follows directly by evaluating $r$ at $(T_{j_o},T_{j_o}^{\perp})$ and is therefore omitted.

\begin{lemma}
\label{lem:decomp}
Under the model in \eqref{eq:gaussian_model}, $(T_{j_o}, T_{j_o}^\perp)$ is jointly normally distributed as
\begin{align}
  \label{eq:joint_gauss}
  \begin{pmatrix} T_{j_o} \\ T_{j_o}^\perp \end{pmatrix}
  \sim
  \mathcal{N}_p\!\left(
    \begin{pmatrix} \mu_{j_o} \\ \nu_{j_o} \end{pmatrix},
    \begin{pmatrix} \sigma_{j_o}^2 & 0 \\ 0 & \Gamma_{j_o}
    \end{pmatrix}
  \right),
\end{align}
where $\nu_{j_o}= R_{j_o} \mu
     - \dfrac{1}{\sigma_{j_o}^{2}}\Sigma_{-j_o,j_o} \mu_{j_o}
   \in \real^{p-1}$, and $\Gamma_{j_o}= \Sigma_{-j_o,-j_o}
     - \dfrac{1}{\sigma_{j_o}^{2}}\Sigma_{-j_o,j_o} \Sigma_{-j_o,j_o}^\top
   \in \real^{(p-1) \times (p-1)}$.
\end{lemma}

\begin{lemma}
Assume the model in \eqref{eq:gaussian_model}. 
The conditional density of $T_{j_o} \Big\lvert \cbracket{\Ehat = E_o, T_{j_o}^\perp = t^\perp}$, evaluated at $t\in \real$, is proportional to
$\phi(t;\, \mu_{j_o}, \sigma_{j_o}^2) \times \Pp[\Ehat=E_o \mid T = r(t, t^\perp)]$.
\label{lem:conditionaldensity}
\end{lemma}

In the expression for the conditional density in Lemma~\ref{lem:conditionaldensity}, the conditional selection probability $\Pp[\Ehat=E_o \mid T = r(t, t^\perp)]$ is taken with respect to the distribution of the external randomization and is therefore free of model parameters. 
It follows immediately that this density depends only on $\mu_{j_o}$.

The following section establishes that this conditional selection probability is exactly the probability of selecting the winning subset $E_o$ under the proposed exponential mechanism.

\subsection{Step 2: Obtaining a closed-form conditional density}

\begin{definition}[Selection weight]
\label{def:Lambda}
For the map $r$ defined in Lemma~\ref{lem:reconstruction} and the observed selection $\Ehat = E_o$, define $\Lambda^{E_o} : \real \times \real^{p-1} \to (0,1)$ as
\begin{align}
  \label{eq:Lambda}
  \Lambda^{E_o}(u_1, u_2)
  =
  \dfrac{
    \exp\!\left\{
      \dfrac{s_{E_o}(r(u_1,u_2))}
            {\tau \cdot \snoise(r(u_1,u_2))}
    \right\}
  }{
    \displaystyle\sum_{E' \in \Ek}
    \exp\!\left\{
      \dfrac{s_{E'}(r(u_1,u_2))}
            {\tau \cdot \snoise(r(u_1,u_2))}
    \right\}
  }.
\end{align}
\end{definition}

\begin{lemma}
\label{lem:lambda_prob}
For all $u = (u_1, u_2^\top)^\top \in \real^p$, we have that
\begin{align*}
  \Lambda^{E_o}(u_1, u_2)
  = \Pp\rbracket{\Ehat = E_o \mid T = r(u_1, u_2)}.
\end{align*}
\end{lemma}

Note that the identity in Lemma~\ref{lem:lambda_prob} holds exactly and follows directly from the definition of the exponential randomization scheme in \eqref{eq:expmech}.
No distributional assumptions on $T$ are required, since the identity is conditional on $T$.

We are now ready to derive an exact expression for the conditional density of interest, which is formalized in Proposition \ref{prop:cond_density}. 

\begin{proposition}[Conditional Density]
\label{prop:cond_density}
Under~\eqref{eq:gaussian_model}, the conditional density function of $T_{j_o}$ given $\cbracket{\Ehat = E_o , T_{j_o}^\perp = t^\perp}$, evaluated at $t \in \real$, is
\[ \left(\displaystyle\int_{-\infty}^{\infty}
    \phi\nbracket{u;\, \mu_{j_o},\, \sigma_{j_o}^2}  \Lambda^{E_o}\nbracket{u, t^\perp}
    \, du \right)^{-1} \times  \phi\nbracket{t;\, \mu_{j_o},\, \sigma_{j_o}^2}
\Lambda^{E_o}\nbracket{t, t^\perp}.
\]
\end{proposition}

\begin{proof}[Proof of Proposition~\ref{prop:cond_density}]
Applying Bayes’ rule, the conditional density in the claim is proportional to
\begin{align*}
  p(T_{j_o}=t \mid \Ehat=E_o, T_{j_o}^\perp=t^\perp) 
  &\propto p(T_{j_o}=t \mid T_{j_o}^\perp=t^\perp)\Pp[\Ehat=E_o \mid T_{j_o}=t, T_{j_o}^\perp=t^\perp]  \\
  &=  \phi(t;\, \mu_{j_o}, \sigma_{j_o}^2)  \Lambda^{E_o}(t, t^\perp),
\end{align*}
where $p(T_{j_o}=t \mid T_{j_o}^\perp=t^\perp)$, the density of $T_{j_o} \lvert T_{j_o}^{\perp}$, is simply the normal density of $T_{j_o}$ by Lemma~\ref{lem:decomp}, and the selection probability equals $\Lambda^{E_o}(t, t^\perp)$ by Lemma \ref{lem:lambda_prob}.
Dividing by the corresponding normalizing constant yields the claimed density.
\end{proof}

\subsection{Step 3: Constructing tests, interval and point estimates}
\label{sec:pivot}
%%---------------------------------------------------------------------------

\begin{definition}[Pivot]
\label{def:pivot}
Consider observing $\{\widehat{E}=E_o\}$. Let
\begin{align}
  \label{eq:pivot}
  \Pivot^{E_o}_{\mu_{j_o}}(T_{j_o};\, T_{j_o}^\perp,\, \sigma_{j_o}^2)
  =
  \dfrac{
    \int_{-\infty}^{T_{j_o}}
    \phi\nbracket{u;\, \mu_{j_o},\, \sigma_{j_o}^2}
    \Lambda^{E_o}\nbracket{u,\, T_{j_o}^\perp}
    \, du
  }{
    \int_{-\infty}^{\infty}
    \phi\nbracket{u;\, \mu_{j_o},\, \sigma_{j_o}^2}
    \Lambda^{E_o}\nbracket{u,\, T_{j_o}^\perp}
    \, du
  }.
\end{align}
\end{definition}

\begin{theorem}[Conditionally-valid pivot]
\label{thm:pivot}
Under~\eqref{eq:gaussian_model}, for any $j_o \in E_o$, we have
\begin{equation*}
  \label{eq:pivot_unif}
  \Pivot^{E_o}_{\mu_{j_o}}(T_{j_o};\, T_{j_o}^\perp,\, \sigma_{j_o}^2)
  \mid \cbracket{\Ehat = E_o}
  \;\sim\;
  \mathrm{Unif}(0,1).
\end{equation*}
\end{theorem}

\begin{proof}[Proof of Theorem~\ref{thm:pivot}]
Observe that the pivot $\Pivot^{E_o}_{\mu_{j_o}}(T_{j_o}; T_{j_o}^\perp,\sigma_{j_o}^2)$ is the conditional CDF corresponding to the density in Proposition~\ref{prop:cond_density}, evaluated at $T_{j_o}$. 
Therefore, by the probability integral transform,
$\Pivot^{E_o}_{\mu_{j_o}}(T_{j_o};T_{j_o}^\perp,\sigma_{j_o}^2)\mid {\widehat{E}=E_o,T_{j_o}^\perp}$
is distributed as a $\mathrm{Unif}(0,1)$ random variable. 
The result follows by further marginalizing over $T_{j_o}^\perp$.
\end{proof}

As an immediate consequence of Theorem~\ref{thm:pivot}, the pivot can be used to construct tests and interval estimates that are valid conditional on the selection event $\{\Ehat = E_o\}$. 
We first define the relevant $p$-value and confidence interval, and then state their conditional guarantees in Corollary \ref{cor:conditional-validity}.

\begin{definition}[Conditionally-valid $p$-value and confidence interval]
\label{def:selective-inference}
Fix $\alpha \in (0,1)$.
\begin{enumerate}
  \item \emph{(Two-sided $p$-value)}. For $\mu_0 \in \real$, define
  \[
    p^{E_o}_{\mu_0}(T_{j_o}; T_{j_o}^\perp, \sigma_{j_o}^2)
    =
    2\min\cbracket{
      \Pivot^{E_o}_{\mu_0}(T_{j_o}; T_{j_o}^\perp, \sigma_{j_o}^2),\;
      1 - \Pivot^{E_o}_{\mu_0}(T_{j_o}; T_{j_o}^\perp, \sigma_{j_o}^2)
    }.
  \]
  To test $H_0 : \mu_{j_o} = \mu_0$ against $H_1 : \mu_{j_o} \neq \mu_0$,
  reject $H_0$ whenever
  $p^{E_o}_{\mu_0}\rbracket{T_{j_o}; T_{j_o}^\perp, \sigma_{j_o}^2} \leq \alpha$. 
  Tests for one-sided alternatives can be defined analogously.
  \item \emph{(Confidence interval)}. A $100\times (1-\alpha)\%$ 
  confidence interval for $\mu_{j_o}$ is obtained by inverting the pivot:
  \[
    C^{E_o}_{j_o}
    =
    \cbracket{
      \mu \in \real :
      \Pivot^{E_o}_{\mu}(T_{j_o}; T_{j_o}^\perp, \sigma_{j_o}^2)
      \in
      \rbracket{\tfrac{\alpha}{2},\; 1 - \tfrac{\alpha}{2}}
    }.
  \]
\end{enumerate}
\end{definition}

\begin{corollary}[Conditional validity]
\label{cor:conditional-validity}
Fix $\alpha \in (0,1)$. Under the assumptions of Theorem~\ref{thm:pivot}, the
inferential procedures in Definition~\ref{def:selective-inference} are valid
conditional on $\{\Ehat = E_o\}$:
\begin{enumerate}[label=(\roman*),topsep=2pt,itemsep=2pt,leftmargin=*]
  \item under $H_0 : \mu_{j_o} = \mu_0$, \;
  $\mathbb{P}\!\left[
p_{\mu_{j_o}}(T_{j_o};T_{j_o}^\perp,\sigma_{j_o}^2) \leq \alpha
\,\middle|\, \widehat{E}=E_o
\right]= \alpha$,
  \item $ \Pp\!\left[
      \mu_{j_o} \in C^{E_o}_{j_o}
      \mid \Ehat = E_o\right]
    = 1-\alpha$.
\end{enumerate}
\end{corollary}

\noindent{\textbf{Point estimate}}. \quad From the univariate conditional density in Proposition~\ref{prop:cond_density}, we obtain the log-likelihood for $\mu_{j_o}$ given the observed values $(T_{j_o}, T_{j_o}^\perp)=(t,t^\perp)$, which equals
\begin{align}
  \label{eq:sel_loglik}
  \ell(\mu_{j_o}; t)
  = \log \phi\nbracket{t;\, \mu_{j_o},\, \sigma_{j_o}^2}
   + \log \Lambda^{E_o}\nbracket{t, t^\perp}
   - \log \mathcal{Z}(\mu_{j_o}, t^\perp),
\end{align} where  $\mathcal{Z}(\mu_{j_o}, t^\perp)$ denotes the  normalizing constant 
\begin{align}
\label{eq:norm_const}
  \mathcal{Z}^{E_o}(\mu_{j_o},t^\perp)
  =
  \int_{-\infty}^{\infty}
  \phi(u;\mu_{j_o},\sigma_{j_o}^2)
  \Lambda^{E_o}(u,t^\perp)
  \,du.
\end{align}
Although this normalizing constant lacks a closed form, it is a one-dimensional integral and can therefore be evaluated accurately by numerical quadrature.
Maximizing this log-likelihood yields the univariate maximum likelihood estimate (MLE), which we use as our point estimate for $\mu_{j_o}$.
Proposition~\ref{prop:smle} gives the estimating equation for this MLE and establishes that it has a unique solution.

\begin{proposition}[Univariate MLE]
\label{prop:smle}
For $\mu_{j_o}\in \real$, define the map
\begin{align*}
  e(\mu_{j_o})
  = \frac{
       \int_{\real} u\, \phi\nbracket{u;\, \mu_{j_o},\, \sigma_{j_o}^2}\,
       \Lambda^{E_o}\nbracket{u, t^\perp}\, du
     }{
       \mathcal{Z}(\mu_{j_o}, t^\perp)
     }.
\end{align*}
Then the selective MLE, 
$\widehat{\mu}_{j_o}^{E_o}
   = \argmax_{\mu_{j_o} \in \real} \ell(\mu_{j_o}; t),$ is the unique solution of $e(\mu_{j_o}) = t$.
\end{proposition}

\section{Efficient implementation}
\label{sec:efficient}
We now outline an efficient implementation of the randomized selection rule and the subsequent inference; full details are deferred to Appendix~\ref{app:effcient}.

The main computational bottleneck is the normalizing constant
\[
  Z_{p,k}(t)
  =
  \sum_{E\in\mathcal E_k}
  \exp\left\{\frac{s_E(t)}{\tau(t)}\right\},
\]
in the exponential mechanism, where $\tau(t)$ is the data-adaptive temperature defined in \eqref{eq:expmech}.
Evaluating $Z_{p,k}(t)$ directly requires summing over all $|\Ek|=\binom{p}{k}$ candidate subsets.
The additive score $s_E(t)=\sum_{j\in E}t_j$, however, allows us write it as an elementary symmetric polynomial,
\[
  Z_{p,k}(t)
  =
  \sum_{E\in\mathcal E_k}
  \prod_{j\in E}\exp\eta_j(t),
  \qquad \text{where }\ \eta_j(t)=\frac{t_j}{\tau(t)},\quad \forall j\in [p],
\]
which enables exact evaluation in $O(pk)$ time by dynamic programming, see \cite{chen1997statistical}; Algorithm~\ref{alg:topk_log_partition_dp} implements it on the log scale for numerical stability.
Related work has leveraged various forms of score structure to implement joint exponential mechanisms over exponentially large output spaces. 
In particular, \citet{gillenwater2021differentially} use a forward--backward dynamic program for a chain-factorized quantile score, while \citet{gillenwater2022topk} group candidate outputs by utility value and use multiplicities to sample efficiently for private top-$k$ selection. 
By contrast, our recursion exploits the additive structure of the score criterion under a fixed-cardinality constraint, reducing the computation to an elementary symmetric polynomial that can be evaluated in $O(pk)$ time.

This routine serves two purposes.
First, it enables sampling of the selected set: Algorithm~\ref{alg:exact_randomized_topk_sampler} evaluates Algorithm~\ref{alg:topk_log_partition_dp} at the observed statistic and returns an exact draw $E_o$ from the exponential mechanism by backward sampling, in $O(pk)$ time (Proposition~\ref{prop:seq_sampling_correct}), following weighted sampling without replacement.
Second, after observing $\widehat E=E_o$, it supplies the selection weights~\eqref{eq:Lambda} that enter the pivot in Definition~\ref{def:pivot}.

For a target $j_o\in E_o$ and conditioning value $t_{j_o}^\perp$, Algorithm~\ref{alg:selective_pivot_computation} approximates the pivot in Definition~\ref{def:pivot} by one-dimensional weighted quadrature in the target statistic $T_{j_o}$, over nodes $u_1,\ldots,u_G$, with selection weights
\[
  \Lambda_g^{E_o}
  =
  \Lambda^{E_o}(u_g,t_{j_o}^\perp)
  =
  \frac{\prod_{j\in E_o}\exp\{\eta_j(t_g)\}}{Z_{p,k}(t_g)},
  \qquad t_g = r(u_g,t_{j_o}^\perp),\quad g=1,\ldots,G.
\]
Algorithm~\ref{alg:selective_weights_grid} precomputes these weights in $O(Gpk)$ time through repeated calls to Algorithm~\ref{alg:topk_log_partition_dp}.
Since they do not depend on $\mu_{j_o}$, the same weights are reused to invert the pivot for confidence intervals and to maximize the selective likelihood.

\section{Asymptotic Theory}
\label{sec:asymptotic}

This section extends our method to asymptotically-linear selection statistics, enabling flexible conditional inference for a broad class of nonparametric problems.
For the asymptotic theory, we make the dependence on sample size explicit, writing $T_n=(T_{n,1},\ldots,T_{n,p})^\top$ for the selection statistics and $\mu_n=(\mu_{n,1},\ldots,\mu_{n,p})^\top$ for the population parameter vector, for which $T_n$ is consistent in the absence of selection. 
Let $\Ehat_n$ denote the randomized selected set obtained by applying the data-adaptive exponential mechanism in \eqref{eq:expmech} to $T_n$.

The main result, Theorem \ref{thm:asymp-main}, shows that the pivot used in Section~\ref{sec:method} to derive conditionally valid tests and confidence intervals remains asymptotically valid conditional on $\cbracket{\Ehat_n=E_o}$.
Before stating the result, we impose two conditions: an asymptotic linear representation (ALR) for the selection statistics and a smoothness condition on the score function, given in Assumptions \ref{ass:linearizable} and \ref{ass:score-deriv}, respectively.

\begin{assumption}[ALR for selection statistics]
\label{ass:linearizable}
We assume that the statistic $T_n\in\real^p$ admits the asymptotic linear representation
\begin{align}
\label{eq:Talr}
  \Sigma^{-1/2}(T_n-\mu_n)
  =
  \frac{1}{\sqrt n}\sum_{i=1}^n a_{i,n}Y_{i,n}+R_n,
\end{align}
where (i) $a_{i,n}\in\real^{p\times d}$ and the remainder $R_n=o_p(1)$ satisfies 
\begin{align*}
  \frac{1}{n}\sum_{i=1}^n a_{i,n}a_{i,n}^{\top}=I_p,
  \qquad
  \sup_{n}\max_{1\le i\le n}\|a_{i,n}\|_{\op}<\infty,
  \qquad
  \sup_n\,\Ee\!\left[\exp\{4C_R\|R_n\|\}\right]<\infty,
\end{align*}
with $C_R:=c_f^{(1)}\max\!\big\{\sigma_{j_o},\,
\|\Gamma_{j_o}^{1/2}\|_{\op}\}$, where
$c_f^{(1)}$ is a constant depending on the smoothness of the score in the selection rule and defined explicitly in Proposition~\ref{prop:f_bounded_deriv}, 
and (ii) $\{Y_{i,n}:1\le i\le n\}$ is a triangular array of independent, standardized and uniformly sub-Gaussian random vectors in $\real^d$,  i.e., satisfying, for some $K_0<\infty$,
\begin{align*}
  \Ee[Y_{i,n}]=0_d,
  \quad
  \Var(Y_{i,n})=I_d,
  \quad
  \sup_{n}\max_{1\le i\le n}
  \Ee[\exp\{s^\top Y_{i,n}\}]
  \le
  \exp\{K_0^2\|s\|^2/2\},
  \qquad
  \forall s\in\real^d.
\end{align*}
\end{assumption}

\begin{remark}
The identity $n^{-1}\sum_{i=1}^n a_{i,n}a_{i,n}^\top=I_p$ can be relaxed to hold only in the limit; i.e., it suffices that $V_n = n^{-1}\sum_{i=1}^n a_{i,n}a_{i,n}^\top \to I_p$.
In this case, one can re-normalize the linear representation so that the identity holds exactly, while absorbing the resulting discrepancy into the remainder term in the ALR.
\end{remark}

\begin{remark}
    We verify the conditions in Assumption~\ref{ass:linearizable} for all examples considered later; see Section~\ref{sec:examples}.
\end{remark}
%\begin{remark}
%\label{rem:alr-sufficient}
%Assumption~\ref{ass:linearizable} holds whenever
%$\|R_n\|\le n^{-1/2}\Xi_n$ with $\Xi_n$ uniformly sub-Gaussian; then
%$\Ee\|R_n\|^4=O(n^{-2})$ and the exponential moment is bounded uniformly
%in $n$. This is satisfied in standard examples, including the
%Bradley--Terry--Davidson model of Section~\ref{sec:examples}, where the
%remainder inherits sub-Gaussian tails from the uniform boundedness of
%the per-match scores.
%\end{remark}

\begin{assumption}[Bounded score derivatives]
\label{ass:score-deriv}
For $l=1,2,3$, the subset score functions $s_E:\real^p\to\real$ in the selection rule satisfy: $c_s^{(l)}
  =\max_{E\in\mathcal E_k}\sup_{t\in\real^p}
   \|\nabla^{(l)}s_E(t)\|_{\op}<\infty$.
\end{assumption}

\begin{remark}
Assumption~\ref{ass:score-deriv} requires only that derivatives up to third order be globally bounded; notably the score $s_E$ itself need not be bounded. This condition holds for scores of the form $s_E(T)=\sum_{j\in E} T_j$ and, more generally, for additive scores $s_E(T)=\sum_{j\in E}s_j(T)$ where each $s_j$ has bounded derivatives.
\end{remark}

\begin{theorem}[Asymptotic conditional validity]
\label{thm:asymp-main}
Under Assumptions~\ref{ass:linearizable} and~\ref{ass:score-deriv}, for any $j_o\in E_o$, we have
\begin{align*}
  \Pivot^{E_o}_{\mu_{n,j_{o}}}\!\big(T_{n, j_{o}};\,T^\perp_{n, j_{o}},\,
    \sigma_{j_o}^2\big)
  \;\big|\;\left\{\widehat E_n=E_o\right\}
  \;\xrightarrow{d}\;
  \mathrm{Unif}(0,1),
  \qquad \text{as } \ n\to\infty.
\end{align*}
\end{theorem} 

As a consequence of Theorem~\ref{thm:asymp-main}, the tests and interval estimates derived from our pivot are asymptotically valid.
The proof relies on a Lindeberg-type argument applied to expectations reweighted by the likelihood ratio between the conditional (or, post-selection) and pre-selection distributions of the selection statistics. 
Details and supporting results for developing the proof are deferred to Appendix~\ref{app:asymp}.

\section{Applications}
\label{sec:examples}

We illustrate our proposed inferential method through three applications: A/B/n testing with binary outcomes (Section~\ref{sec:dosage}), ranking via Bradley--Terry--Davidson models (Section~\ref{sec:BT}), and nonparametric feature importance (Section~\ref{sec:VI}). All three are non-Gaussian, so the exact pivot of Section~\ref{sec:pivot} no longer applies and the validity rests on the asymptotic guarantee of Theorem~\ref{thm:asymp-main}. For each, we specify the data-generating model, the target parameters, and the asymptotically linear selection statistic, and verify Assumption~\ref{ass:linearizable}.

% We illustrate our proposed inferential method through three examples: (i) ranking via Bradley--Terry--Davidson models (Section~\ref{sec:BT}), (ii) nonparametric feature importance (Section~\ref{sec:VI}), and A/B/n testing with binary outcomes (Section~\ref{sec:dosage}). 

\subsection{A/B/n testing with binary outcomes}
\label{sec:dosage}

Consider a treatment with $p$ levels and $n$ independent subjects in each level\footnote{Extensions to unequal group sizes with group proportions converging to fixed limits are straightforward.}. 
For a treatment level $j$, we denote the observations in level $j$ as $O_{ij}\sim \mathrm{Bernoulli}(\pi_j)$, where $O_{ij}=1$ denotes a favorable outcome, such as recovery.
Two parameters of interest for each level $j$ are the success probability $\pi_j$ and the corresponding log-odds $\log(\pi_j/(1-\pi_j))$.

For success probabilities, let $\mu_n=\sqrt {n} \pi$ where $\pi=c(\pi_1, \ldots, \pi_p)^\top$, $\Sigma=\operatorname{Diag}(\pi_j(1-\pi_j))_{j=1}^p$, and $T_{n}=\sqrt{n}\widehat\pi_n$ where $\widehat\pi_{n,j}=n^{-1}\sum_{i=1}^n O_{ij}$. 
For $i\in[n]$, denote $Y_{i,n}=\Sigma^{-1/2}(O_{i1}-\pi_1,\ldots,O_{ip}-\pi_p)^\top\in\mathbb R^d$, where $d=p$. 
In Proposition~\ref{prop:ALR.ABn.prob}, we verify that Assumption~\ref{ass:linearizable} holds in this setting.
For log-odds, for some $\varepsilon' > 0$, let
$T_{n,\varepsilon'}=\sqrt n\log\{\widehat\pi_{n,\varepsilon'}/(1-\widehat\pi_{n,\varepsilon'})\}$ where $\widehat\pi_{n,\varepsilon'} = \min \{1-\varepsilon', \max\{\varepsilon', \widehat\pi_{n}\}\}$,
$\mu_n=\sqrt n\log\{\pi/(1-\pi)\}$, and
$\Sigma=\operatorname{Diag}\left(1/\{\pi_j(1-\pi_j)\}\right)_{j=1}^p$. For $i\in[n]$, denote
$Y_{i,n}=\Sigma^{-1/2}\left(\frac{O_{i1}-\pi_1}{\pi_1(1-\pi_1)},\ldots,\frac{O_{ip}-\pi_p}{\pi_p(1-\pi_p)}\right)^\top\in\mathbb R^d$, where again $d=p$.
In Proposition~\ref{prop:ALR.ABn.odds}, we verify that Assumption~\ref{thm:asymp-main} holds in this setting, provided that $\pi_j \in (\varepsilon, 1-\varepsilon)$ for some $\varepsilon > 0$.

\subsection{Ranking via Bradley--Terry--Davidson models}
\label{sec:BT}

The Bradley--Terry (BT) model \citep{bradley1952rank} is commonly used for analyzing paired-comparison data to rank abilities, and we work with its extension, the Bradley--Terry--Davidson (BTD) model \citep{davidson1977extending}, which further accommodates draws.
Explicitly, suppose there are $p$ players competing in a round-robin tournament.
Let 
$\theta = (\theta_1,\ldots,\theta_p) \in \mathbb{R}^p$, 
$\sum_{j=1}^p \theta_j = 0$ be the ability parameters, and 
$\eta \in \mathbb{R}$ be the draw-propensity parameter.
Index the $n$ matches by
$i \in [n]\equiv (i_1,i_2) \in [p] \times [p]$, where match $i$ takes place between players
$i_1 < i_2$.
The BTD model assumes the outcome corresponding to a win for player
$i_1$, a win for player $i_2$, or a draw,
$(O_i^{(1)}, O_i^{(2)}, O_i^{(3)})$, follows
$\mathrm{Multinomial}(1;\, p_i^{(1)}, p_i^{(2)}, p_i^{(3)})$, with
$p_i^{(1)} \propto \exp\{\theta_{i_1}\}$, $p_i^{(2)} \propto \exp\{\theta_{i_2}\}$, and
$p_i^{(3)} \propto \exp\{\eta + (\theta_{i_1}+\theta_{i_2})/2\}$.
Let $\widehat{\theta}_n$ and $\widehat{\eta}_n$ be the maximum likelihood estimators subject to 
$\sum_j \theta_j = 0$, 
\begin{align}
  \label{eq:BTD_loglik}
  \argmax_{\substack{(\theta,\eta)}}
  \sum_{i=1}^n
  \Bigl\{
    O_i^{(1)}\theta_{i_1}
    + O_i^{(2)}\theta_{i_2}
    + O_i^{(3)}\!\left(\eta+\dfrac{\theta_{i_1}+\theta_{i_2}}{2}\right)
    - \log S_i(\theta,\eta)
  \Bigr\}, 
\end{align}
where $S_i(\theta,\eta) = \exp\{\theta_{i_1}\}+\exp\{\theta_{i_2}\}
+2\,\exp\{\eta+(\theta_{i_1}+\theta_{i_2})/2\}$.
% \ell_n(\mu,\eta), \quad\ell_n(\mu,\eta)
%   \;=\;
We define the selection statistic by $T_n=\sqrt{n}\widehat{\theta}_n$ and let $\mu_n=\sqrt{n}\theta$. 
For $1\leq j<l\leq p$, let $n_{jl}$ denote the number of matches played between players $j$ and $l$, and thus $n=\sum_{1\leq j<l\leq p}n_{jl}$.
Let $\Sigma=(I_{\theta\theta}-I_{\theta\eta}I_{\eta\eta}^{-1}I_{\eta\theta})^{-1}$, the inverse of the Schur complement of $I_{\eta\eta}$ in the limiting Fisher information matrix.
In Proposition~\ref{prop:BTD_normal}, we verify that Assumption~\ref{thm:asymp-main} holds in this setting under mild regularity conditions.

\subsection{Non-parametric feature importance}
\label{sec:VI}

Let $(X_i, O_i) \sim \Pp$, $i \in [n]$ be i.i.d. observations, where \(X_i = (X_{i1}, \ldots, X_{ip}) \in \mathbb{R}^p\) denotes the $p$ features, \(O_i \in \mathbb{R}\) denotes the outcome,  \(\Pp\) denotes an unknown distribution.
Following \citet{williamson2023general}, the importance of the $j$-th feature is measured by
\begin{align*}
    \psi_j
    = \frac{\mathbb{E}\left[\left(O_{i} - \mathbb{E}[O_i \mid X_{i,-j}]\right)^2\right]}{\Var(O_i)} - \frac{\mathbb{E}\left[\left(O_i - \mathbb{E}[O_i \mid X_i]\right)^2\right]}{\Var(O_i)},
\end{align*}
where $X_{i,-j}$ denotes $X_i$ excluding $X_{ij}$.
A larger value of $\psi_j$ indicates that a larger proportion of the variation in $O_i$ cannot be explained without $X_{ij}$, indicating that $X_{ij}$ is more important.

For $j\in[p]$, define
\begin{align*}
\psi_{ij}
= \frac{
\left({m}(X_i)-{m}_{-j}(X_{i})\right)^2
+
2(O_i-{m}(X_i))({m}(X_i)-{m}_{-j}(X_{i}))
}{\Var(O_i)
} - \psi_j\frac{(O_i-\Ee[O_i])^2}{\Var(O_i)},
\end{align*}
where $m(x)=\mathbb{E}[O \mid X=x]$,
$m_{-j}(x_{-j})=\mathbb{E}[O \mid X_{-j}=x_{-j}]$.
Let $\Sigma
=\operatorname{Diag}(\Var(\psi_{ij}))_{j=1}^p$, and define
$Y_{i,n} = \Sigma^{-1/2} (\psi_{i1} - \psi_1, \ldots, \psi_{ip} - \psi_p)^\top$.
We estimate $\psi_j$ by
\begin{align*}
\widehat{\psi}_{n,j}
=
\frac{1}{n}\sum_{i=1}^n \widehat{\psi}_{ij},\quad
\widehat{\psi}_{ij} =
\frac{
\left(\widehat{m}(X_i)-\widehat{m}_{-j}(X_{i})\right)^2
+
2(O_i-\widehat{m}(X_i))(\widehat{m}(X_i)-\widehat{m}_{-j}(X_{i}))
}{
(1/n)\sum_{i=1}^n (O_i-\bar{O}_n)^2
},
\end{align*}
where $\widehat{m}(x)$ and $\widehat{m}_{-j}(x_{-j})$ denote estimates of $m(x)$ and $m_{-j}(x_{-j})$ obtained using flexible nonparametric methods on an independent hold-out dataset. 
Here the term
% $\psi_{j}(O_i-\bar O_n)^2/(n^{-1}\sum_{i=1}^n (O_i-\bar O_n)^2)$ in $\widehat \psi_{ij}$ cancels when averaging over $i\in[n]$.
If no hold-out dataset is available, we use cross-fitting as described in Algorithm~\ref{alg:feature.importance}.

We define the selection statistic by
$T_n = \sqrt{n}\widehat{\psi}_n$, and let
$\mu_n = \sqrt{n} \psi$.
We use $T_{n,\varepsilon'}$ to denote the counterpart of $T_n$ where $\Var(O_i)$ is estimated by $\max\{\varepsilon', (1/n)\sum_{i=1}^n (O_i-\bar{O}_n)^2\}$.
%We set $a_{i,m}=I_p$.
In Proposition~\ref{prop:VIM_normal}, we verify that Assumption~\ref{thm:asymp-main} holds in this setting when $\Var(O_i)$ is bounded away from zero and the nuisance estimates are obtained from an independent hold-out sample. 
The result can be extended to the cross-fitted estimator following \cite{chernozhukov2018double}.
%The following proposition shows the asymptotic linear representation of $T_n$ when the nuisance estimates are obtained from an independent hold-out sample, and the result can be extended to the cross-fitted estimator following \cite{chernozhukov2018double}.

% \begin{figure}[t]
%   \centering
%   \includegraphics[width=\textwidth]{Figures/sim_gaussian.png}
%   \vspace{-0.9cm}
%   \caption{Gaussian design ($p=20$, $k=3$). Left: selection quality
%     (standardized regret; Standard is zero by construction, smaller is better).
%     Center: marginal coverage, with the nominal level $0.95$ marked. Right:
%     average interval length. Points are Monte Carlo means with $\pm1$ standard
%     error, shown for the Low- and High-signal regimes.}
%   \label{fig:sim-gaussian}
% \end{figure}

\section{Simulation experiments}
\label{sec:simulations}

We evaluate our method (``Randomized PSI'') against three baselines across four designs: an exact Gaussian-model design and three other designs corresponding to the applications in Section~\ref{sec:examples}—binomial dosage, Bradley--Terry--Davidson comparisons, and nonparametric feature importance using a spline-GAM learner.
Within each design we consider two regimes, defined by the standardized gap $\Delta_{\mathrm{std}}$ between the $k$-th and $(k{+}1)$-th largest population scores: \emph{Weak separation} ($\Delta_{\mathrm{std}}=0.3$) and \emph{Strong separation} ($\Delta_{\mathrm{std}}=2.0$), corresponding to less and more reliable selection, respectively.

\paragraph{Baselines and evaluation summaries.}
Across all designs, we select the top $k=3$ winners under the additive score $s_E(T)=\sum_{j\in E}T_j$, and all interval estimates target the nominal level $1-\alpha=0.95$. For our method, the temperature is set without tuning by taking the regret budget $q=1$ in Corollary~\ref{cor: temp par}, giving $\tau=(\log|\Ek|)^{-1}$.

Note that ``Standard'' denotes the top-$k$ rule in~\eqref{eqn: standard selection}. The two nonrandomized inferential baselines, the conditionally valid  ``Polyhedral PSI'' method of \citet{reid2017post} and the marginally valid ``Zoom Correction'' method of \citet{zrnic2025flexible}, perform inference after selection by this rule. ``Data Splitting'' follows the common rule of thumb of using one half of the data for selecting the winners and the other half for inference.

We report three Monte Carlo summaries: (i) selection quality, measured by regret relative to the full-data top-$k$ rule, $s^*(T)-s_{\Ehat}(T)$,
standardized by the score dispersion $\snoise(T)$; as defined in \eqref{eq:regret-bound}, (ii) marginal coverage of the selected winners, as defined in \eqref{eq:marg_cov_avg}, and (iii) average interval length of intervals for selected winners. 

Note that ``Standard'' has zero regret by construction, and smaller values indicate closer agreement with the selection based on the full sample.
We report only marginal coverage in this section because, as Figure~\ref{fig:example1} already illustrates, methods that guarantee only marginal validity may fail to attain the nominal conditional coverage. 
Results for the Gaussian design are summarized in Figure~\ref{fig:example2}, while results for the Binomial dosage model, the Bradley--Terry--Davidson paired-comparison model, and the nonparametric feature-importance model are reported in Figures~\ref{fig:example3}, \ref{fig:example4}, and~\ref{fig:example5}, respectively.

\paragraph{Results and takeaways.}
The key findings from Figures~\ref{fig:example2}, \ref{fig:example3}, \ref{fig:example4} and \ref{fig:example5}, which are consistent across all four designs, are as follows:
\begin{enumerate}[leftmargin=*]
\item \textbf{Selection quality.} ``Randomized PSI'' closely matches the standard top-$k$ rule in terms of the empirical score attained by the selected subset, with low regret across all designs and signal regimes. ``Data Splitting'' performs substantially worse and exhibits much higher regret; allocating more data to selection can reduce this regret, but this would leave less data for inference, and therefore produce intervals wider than those shown in the third panels and further reduce inferential power relative to ``Randomized PSI''.
\item \textbf{Marginal validity of inference.} All baselines and ``Randomized PSI'' attain the nominal coverage level of $0.95$ across all designs, as expected. Zoom Correction, however, is somewhat conservative and exceeds the nominal coverage level in most experiments.
\item \textbf{Inferential power.} ``Randomized PSI'' produces the shortest valid intervals across all designs. Its intervals are consistently and substantially shorter than those from the widely used polyhedral approach, offering a clear improvement over this conditional baseline. Note that the intervals produced by ``Polyhedral PSI'' are so wide that including them in the plots obscures comparisons with the other methods; instead, we report in the figure captions the improvement in power achieved by our approach relative to this conditional baseline.
Our interval estimates are also comparable in length to, and often shorter than, those from ``Data Splitting''. As noted earlier, however, ``Data Splitting'' achieves this only at the cost of substantially lower selection quality.
Finally, ``Randomized PSI'' also produces intervals no longer than those from ``Zoom Correction'', despite providing the stronger conditional guarantee. As emphasized earlier, this guarantee yields honest interval estimates for the selected winners actually observed, which marginally valid intervals do not provide.
In particular, our interval estimates become narrower as the separation gap increases, showcasing the ability of conditional intervals to adapt to the strength of the signal in the data.
\end{enumerate}

\begin{figure}[H]
  \centering
  \includegraphics[width=\textwidth]{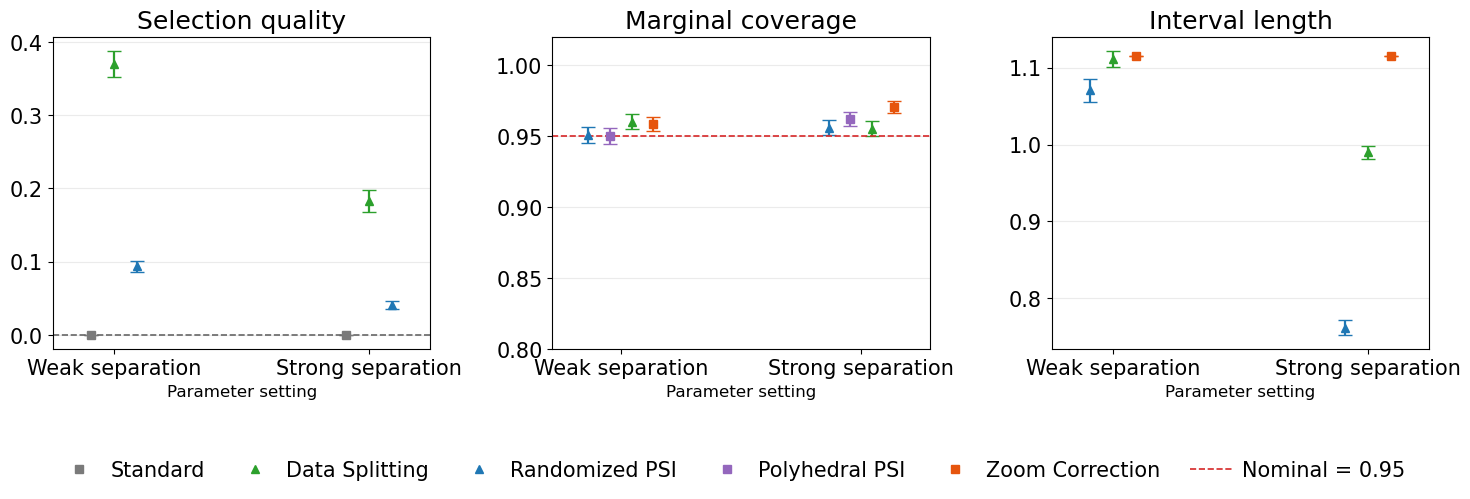}
  \vspace{-1cm}
  \caption{Comparisons on Gaussian data. The left, middle, and right panels show selection equality, marginal coverage, and interval width, respectively. In the low- and high-signal settings, the ``Polyhedral PSI'' intervals are \textcolor{red}{452.5\%} and \textcolor{red}{98.9\%} longer, respectively, than the proposed ``Randomized PSI'' intervals.}
  \label{fig:example2}
\end{figure}

\begin{figure}[H]
  \centering
  \includegraphics[width=\textwidth]{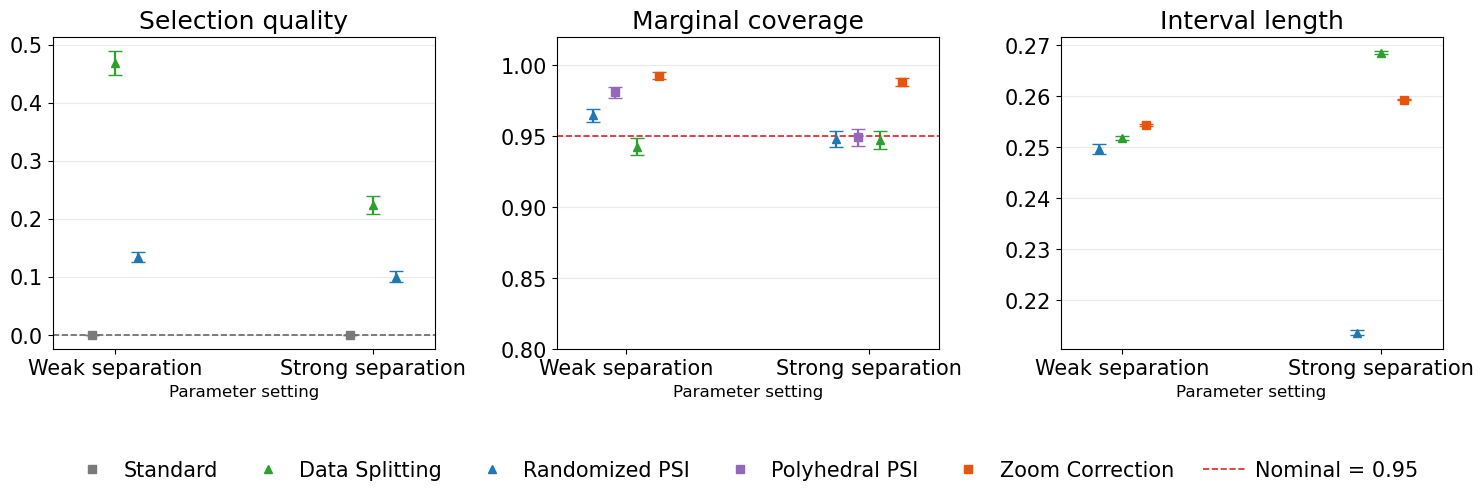}
  \vspace{-1cm}
  \caption{Comparisons on Binomial data. The left, middle, and right panels show selection equality, marginal coverage, and interval width, respectively. In the low- and high-signal settings, the ``Polyhedral PSI'' intervals are \textcolor{red}{70.3\%} and \textcolor{red}{18.5\%} longer, respectively, than the proposed ``Randomized PSI'' intervals.}
  \label{fig:example3}
\end{figure}

\begin{figure}[H]
  \centering
  \includegraphics[width=\textwidth]{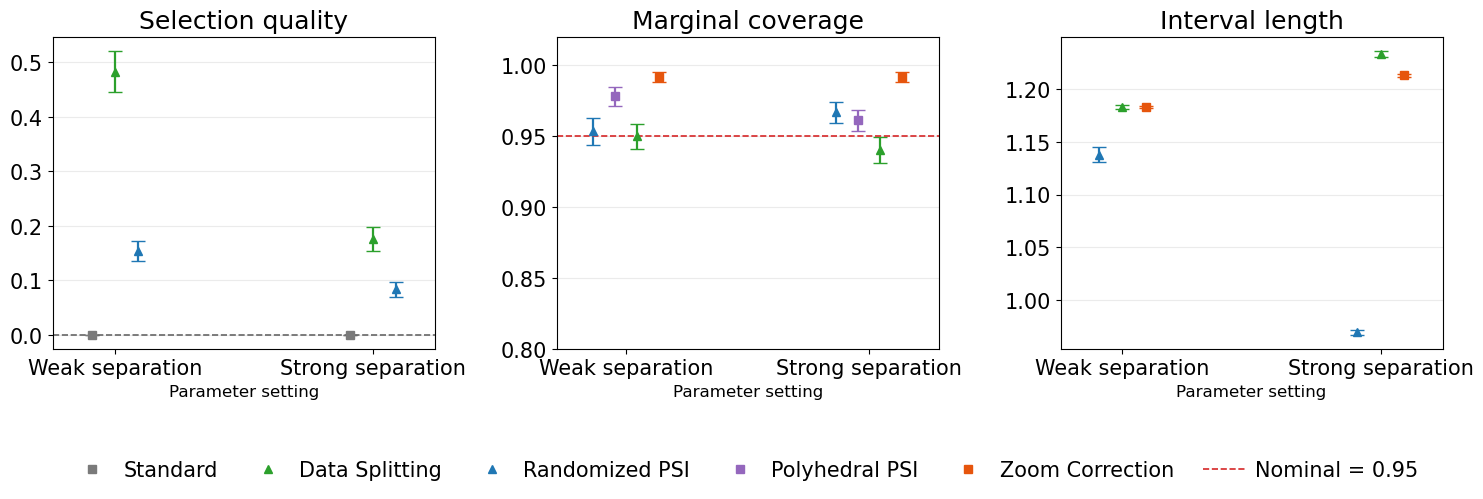}
  \vspace{-1cm}
  \caption{Comparisons on pairwise-comparisons data via Bradley--Terry--Davidson model. The left, middle, and right panels show selection equality, marginal coverage, and interval width, respectively. In the low- and high-signal settings, the ``Polyhedral PSI'' intervals are \textcolor{red}{101.5\%} and \textcolor{red}{29.6\%} longer, respectively, than the proposed ``Randomized PSI'' intervals.}
  \label{fig:example4}
\end{figure}

\begin{figure}[H]
  \centering
  \includegraphics[width=\textwidth]{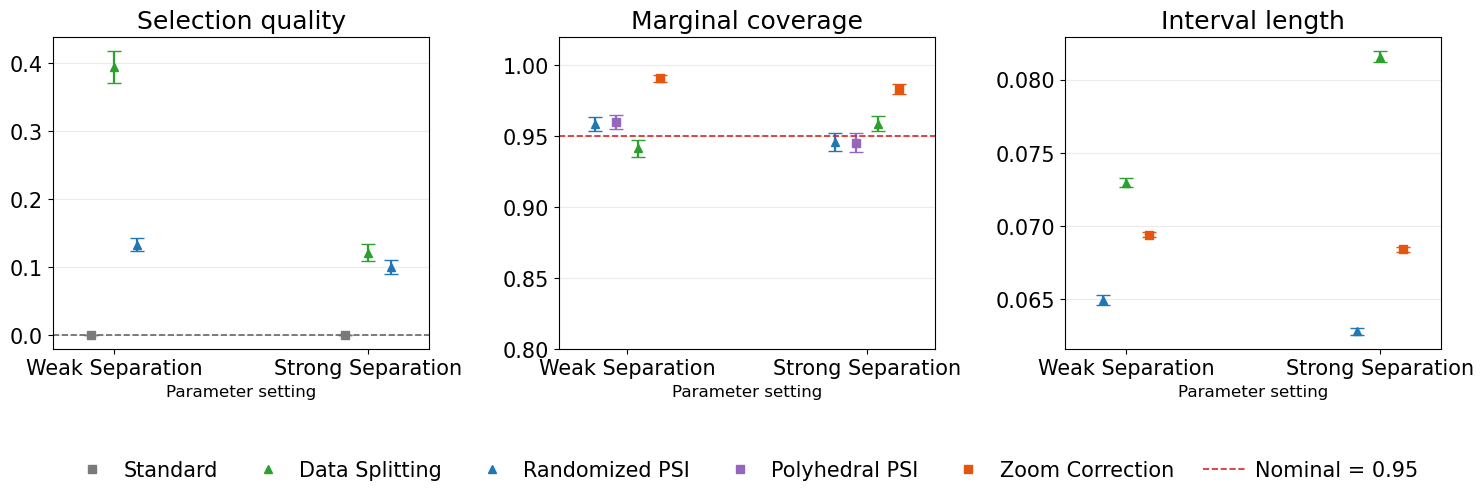}
  \vspace{-1cm}
  \caption{Comparisons on non-parametric feature importance: The left, middle, and right panels show selection equality, marginal coverage, and interval width, respectively. In the low- and high-signal settings, the ``Polyhedral PSI'' intervals are \textcolor{red}{254.0\%} and \textcolor{red}{133.5\%} longer, respectively, than the proposed ``Randomized PSI'' intervals.}
  \label{fig:example5}
\end{figure}

\paragraph{Code availability.}
The code used to reproduce the numerical experiments and empirical results in this paper is publicly available at
\url{https://github.com/BakshiSoham/top-k-winners-with-confidence}.

\section{Conclusion}

This work provides a principled approach to drawing valid inference about winners selected from the observed data.
Among existing methods, simultaneous approaches do not guarantee valid inference for the winners realized in a given dataset, whereas conditional approaches are often tailored to specific selection rules, lack broad applicability across data types, and produce intervals that are too wide to be practically useful. 
Our work fills these gaps with a data-adaptive randomization method based on the exponential mechanism, enabling conditionally valid tests, intervals, and point estimates for selected winners.

The proposed method makes randomized inference more broadly applicable in practice by providing a principled approach to selecting the randomization level and enabling assumption-lean, nonparametric inference.
In particular, our tuning-free choice of randomization, guided by a newly introduced notion of conditional regret, ensures that gains in inferential power do not come at the expense of selection quality. Furthermore, we provide statistical guarantees showing that the same randomization scheme enables assumption-lean inference without restrictions on the form of the selection event, making the proposed method, to our knowledge, one of the most flexible approaches in the post-selection inference literature.
For these reasons, we believe the randomization approach in our paper provides a promising framework for future research not only on the winner’s curse but also on a broad range of post-selection inference problems. Promising directions include improving power further by avoiding extra conditioning and addressing follow-up questions that compare the selected winners, such as between the winner and the first runner-up, as further queries on the same data.

\bibliographystyle{apalike}
\bibliography{references}

\begin{appendices}

%%---------------------------------------------------------------------------

%---------------------------------------------------------------------------
\section{Efficient Implementation Details}
\label{app:effcient}

The dynamic program for the normalizing constant $Z_{p,k}(t)$ computes partial normalizing constants
\[
  Z_{i,\ell}(t)
  =
  \sum_{\substack{
      E\subseteq \{1,\ldots,i\}\\
      |E|=\ell
  }}
  \prod_{j\in E}\exp\eta_j(t),
  \qquad
  i=0,\ldots,p,\quad \ell=0,\ldots,k
\] with boundary conditions $Z_{i,0}(t)=1, \ Z_{0,l}(t) = 0$ for all $i, l \geq 1$, following the recursion 
\begin{align}\label{eq:recursive.rule}
  Z_{i,\ell}(t)
  =
  Z_{i-1,\ell}(t)
  +
  \exp\eta_i (t) Z_{i-1,\ell-1}(t), \quad \forall i\leq p, \  l\leq k.
\end{align}
Intuitively, the first term in Eq.~\eqref{eq:recursive.rule} corresponds to subsets that exclude coordinate $i$, while
the second term corresponds to subsets that include coordinate $i$ and
choose the remaining $\ell-1$ coordinates from $\{1,\ldots,i-1\}$. For numerical stability, the recursion is implemented on the log scale in Algorithm~\ref{alg:topk_log_partition_dp},
\[ L_{i,\ell}(t)=\log Z_{i,\ell}(t), \quad
  L_{i,\ell}(t)
  =
  \operatorname{logaddexp}
  \left\{
    L_{i-1,\ell}(t),
    \eta_i(t)+L_{i-1,\ell-1}(t)
  \right\}
\] where $\mathrm{logaddexp}(a,b) = \log(\exp a + \exp b)$ computed in a numerically stable manner, and invalid states are assigned $-\infty$.

\begin{proposition}[Correctness of Algorithm~\ref{alg:exact_randomized_topk_sampler}]
\label{prop:seq_sampling_correct}
Algorithm~\ref{alg:exact_randomized_topk_sampler} returns a subset $\widehat{E} \in \mathcal{E}_k$ satisfying
\[
  \Pp[\widehat{E} = E \mid T = t]
  \;=\;
  \frac{\exp\{\tau^{-1}s_E(t)\}}{\sum_{E' \in \mathcal{E}_k} \exp\{\tau^{-1}s_{E'}(t)\}},
  \qquad\forall E \in \mathcal{E}_k.
\]
\end{proposition}

\begin{algorithm}[H]
\caption{\textsc{TopKLogPartitionDP}: log-partition dynamic program}
\label{alg:topk_log_partition_dp}
\begin{algorithmic}[1]
\Require Selection statistic value $t\in\mathbb R^p$; subset size $k$; temperature function $\tau(\cdot)$
% \Ensure Log-weights $\eta=(\eta_1,\ldots,\eta_p)$; log-DP table $\nbracket{L_{i,l}: l \in [k] \leq i\in[p]}$ with log-partition value $L_{p,k}$ 
\State Compute $\eta_j \gets t_j/\tau(t)$ for $j=1,\ldots,p$

\State Initialize $L_{i,\ell}\gets -\infty$ for $i=0,\ldots,p$ and $\ell=0,\ldots,k$
\State Set $L_{i,0}\gets 0$ for $i=0,\ldots,p$

\For{$i=1,\ldots,p$}
    \For{$\ell=1,\ldots,\min\{i,k\}$}
        \State
        $L_{i,\ell}
        \gets
        \operatorname{logaddexp}
        \left(
          L_{i-1,\ell},
          \eta_i + L_{i-1,\ell-1}
        \right)$
    \EndFor
\EndFor

\State \Return Log-weights $\eta=(\eta_1,\ldots,\eta_p)$; log-partition DP table $L=\nbracket{L_{i,l}: l \in [k], i\in[p]}$ 
\end{algorithmic}
\end{algorithm}

\begin{algorithm}[H]
\caption{Exact sampling from the randomized top-$k$ rule}
\label{alg:exact_randomized_topk_sampler}
\begin{algorithmic}[1]
\Require Observed statistic $t_{\mathrm{obs}}\in\mathbb R^p$; subset size $k$; temperature function $\tau(\cdot)$

\State $(\eta,L)\gets \textsc{TopKLogPartitionDP}(t_{\mathrm{obs}},k,\tau(\cdot))$
\State $\widehat E\gets\emptyset$, \quad $m\gets k$

\For{$i=p,p-1,\ldots,1$ \textbf{while} $m>0$}
    \State $q_i \gets
    \exp\{\eta_i+L_{i-1,m-1}-L_{i,m}\}$
    \State Draw $B_i\sim\operatorname{Bernoulli}(q_i)$
    \If{$B_i=1$}
        \State $\widehat E\gets \widehat E\cup\{i\}$, \quad $m\gets m-1$
    \EndIf
\EndFor

\State \Return $\widehat E$
\end{algorithmic}
\end{algorithm}

\begin{algorithm}[H]
\caption{\textsc{SelectiveWeightsOnGrid}: selective weights on a quadrature grid}
\label{alg:selective_weights_grid}
\begin{algorithmic}[1]
\Require Observed selected set $E_o$; nuisance value $t^\perp$; grid points $u_1,\ldots,u_G$; reconstruction map $r$; subset size $k$; temperature function $\tau(\cdot)$
\Ensure Selective weights 
$\Lambda^{E_o}_1,\ldots,\Lambda^{E_o}_G$
\For{$g=1,\ldots,G$}
    \State $t_g \gets r(u_g,t^\perp)$

    \State
    $(\eta^{(g)},L^{(g)})
    \gets \textsc{TopKLogPartitionDP}(t_g,k,\tau(\cdot))$ 

    \State
    $\log \Lambda^{E_o}_g
    \gets
    \displaystyle
    \sum_{j\in E_o}\eta^{(g)}_j
    -
    L^{(g)}_{p,k}$

    \State $\Lambda^{E_o}_g\gets \exp\{\log\Lambda^{E_o}_g\}$
\EndFor

\State \Return $\Lambda^{E_o}_1,\ldots,\Lambda^{E_o}_G$
\end{algorithmic}
\end{algorithm}

\begin{algorithm}[H]
\caption{\textsc{SelectivePivot}: numerical computation of the selective pivot}
\label{alg:selective_pivot_computation}
\begin{algorithmic}[1]
\Require Observed target statistic $T_{j_o}$; null value $\mu_{j_o}$; variance $\sigma_{j_o}^2$; observed selected set $E_o$; nuisance value $t^\perp$; grid points $u_1,\ldots,u_G$; quadrature weights $a_1,\ldots,a_G$; reconstruction map $r$; subset size $k$; temperature function $\tau(\cdot)$
\Ensure Approximation to $\Pivot^{E_o}_{\mu_{j_o}}(T_{j_o})$

\State
$(\Lambda^{E_o}_1,\ldots,\Lambda^{E_o}_G)
\gets
\textsc{SelectiveWeightsOnGrid}
(E_o,t^\perp,\{u_g\}_{g=1}^G,r,k,\tau(\cdot))$

\State
$N
\gets
\displaystyle
\sum_{g:u_g\le T_{j_o}}
a_g\,
\phi(u_g;\mu_{j_o},\sigma_{j_o}^2)
\Lambda^{E_o}_g$

\State
$D
\gets
\displaystyle
\sum_{g=1}^{G}
a_g\,
\phi(u_g;\mu_{j_o},\sigma_{j_o}^2)
\Lambda^{E_o}_g$

\State \Return $N/D$
\end{algorithmic}
\end{algorithm}

\subsection{Proofs for Section~\ref{sec:exp:mechanism}}

\begin{proof}[Proof of Proposition~\ref{prop:mismatch}]
For each $E \in \Ek$, let $\Delta_E(t) = s^*(t)-s_E(t)\geq 0$. 
Dividing both the numerator and denominator in~\eqref{eq:expmech} by $\exp\{s^*(t)/\tau(t)\}$, we obtain, for each $E\in\Ek$,
\[
\Pp\rbracket{\Ehat =E\mid T=t} \propto \exp\!\left\{
  \frac{s_E(t)-s^*(t)}{\tau(t)}
  \right\}
  =
  \exp\!\left\{
  -\frac{\Delta_E(t)}{\tau(t)}
  \right\}.
\]
Combining this with the fact that $\Delta_E(t)=0$ for all $E\in\Estar(t)$ yields
\[
\Pp\rbracket{\Ehat \notin \Estar(T)\mid T=t}
=
\frac{
\displaystyle
\sum_{E\notin\Estar(t)}
\exp\!\left\{
-\frac{\Delta_E(t)}{\tau(t)}
\right\}
}{
\displaystyle
|\Estar(t)|
+
\sum_{E\notin\Estar(t)}
\exp\!\left\{
-\frac{\Delta_E(t)}{\tau(t)}
\right\}
}.
\] 

Note that, for every $E\notin\Estar(T)$, we have $\Delta_E(t)\geq \Delta_+(t)$.
Furthermore, since the map $x\mapsto \dfrac{x}{(|\Estar(t)|+x)}$ is increasing for $x\geq 0$, we have
\begin{align*}
 & \sum_{E\notin\Estar(t)}
  \exp\!\left\{
  -\frac{\Delta_E(t)}{\tau(t)}
  \right\}
  \le
  \bigl(|\Ek|-|\Estar(t)|\bigr)
  \exp\!\left\{
  -\frac{\Delta_+(t)}{\tau(t)}
  \right\}.
\end{align*}  
Using the bound in the previous display, we conclude that
 \begin{align*} 
\Pp\rbracket{\Ehat \notin \Estar(T)\mid T=t}
\leq
\dfrac{
\bigl(|\Ek|-|\Estar(t)|\bigr)
  \exp\!\left\{
  -\dfrac{\Delta_+(t)}{\tau(t)}
  \right\}
}{
\displaystyle
|\Estar(t)|
+
\bigl(|\Ek|-|\Estar(t)|\bigr)
  \exp\!\left\{
  -\frac{\Delta_+(t)}{\tau(t)}
  \right\}
}.
\end{align*}  

The final bound on the mismatch probability in the claim follows by noting that $|\Estar(t)|\ge 1$, and hence the denominator is bounded below by $1$.
\end{proof}

\begin{proof}[Proof of Proposition~\ref{prop:regret}]
We begin by noting that
\[
  \Ee\!\left[s^*(T)-s_{\Ehat}(T)\mid T=t\right]
  =
  s^*(t)
  -
  \sum_{E\in\Ek}
  \Pp\{\Ehat=E\mid T=t\}\,s_E(t).
\]
Therefore, it suffices to obtain a lower bound on the weighted score
$\sum_{E\in\Ek}\Pp\{\Ehat=E\mid T=t\}s_E(t)$. Given $\{T=t\}$, applying Lemma \ref{lem:gibbs-variational} with $a_E=s_E(t)$, $\lambda=\tau(t)$, and $q$ equal to the discrete uniform distribution on $\Estar(t)$, i.e.,
$q_E=|\Estar(t)|^{-1}\mathbf{1}\{E\in\Estar(t)\}$,
gives
\[
  s^*(t)+\tau(t)\log |\Estar(t)| \leq \sum_{E\in\Ek}
  \pi_E s_E(t)
  +
  \tau(t) H(\pi),
\]
where 
$$\pi_E=\Pp\left[\Ehat=E\mid T=t\right] = \dfrac{\exp\!\left\{\dfrac{s_E(t)}{\tau(t)}\right\}}
       {\displaystyle\sum_{E' \in \Ek}
        \exp\!\left\{\dfrac{s_{E'}(t)}{\tau(t)}\right\}}$$
is the probability of selecting the subset $E$ under the exponential randomization scheme defined in \eqref{eq:expmech}. Therefore, it follows that
\[
  s^*(t)
  -
  \sum_{E\in\Ek}
  \Pp\left[\Ehat=E\mid T=t\right]s_E(t)
  \le
  \tau(t)\{H(\pi)-\log |\Estar(t)|\}.
\]
Finally, noting that $H(\pi)\le \log |\Ek|$, since $\pi$ is supported on $\Ek$, which has $|\Ek|$ elements, leads to the regret bound:
\[
  \Ee\!\left[s^*(T)-s_{\Ehat}(T)\mid T=t\right]
  =
  s^*(t)
  -
  \sum_{E\in\Ek}
  \Pp\{\Ehat=E\mid T=t\}s_E(t)
  \le
  \tau(t)\log\frac{|\Ek|}{|\Estar(t)|}.
\]
Equivalently, we have
\[
  \Ee\!\left[\frac{1}{\snoise(T)}(s^*(T)-s_{\Ehat}(T))\Big\lvert T=t\right]\le
  \tau\log\frac{|\Ek|}{|\Estar(t)|}.
\]
\end{proof}

The following finite-dimensional Gibbs variational principle follows from
the variational characterization of the Kullback-Leibler divergence 
established by \cite{Csiszar1975}.

\begin{lemma}[Adopted from \cite{Csiszar1975}]
\label{lem:gibbs-variational}
Let $\{a_E:E\in\Ek\}$ be a finite collection of real numbers, and let
$\lambda>0$. 
Consider the following map defined on the probability simplex $\Delta(\Ek)$
\[
  q\in\Delta(\Ek) \mapsto \sum_{E\in\Ek}q_Ea_E+\lambda H(q),
\]
where $H(q)=-\sum_{E\in\Ek}q_E\log q_E$.
The maximizer of this map equals $\pi \in \Delta(\Ek)$, given by
\[
  \pi_E
  =
  \frac{\exp\{a_E/\lambda\}}
       {\sum_{E'\in\Ek}\exp\{a_{E'}/\lambda\}},
  \qquad E\in\Ek .
\]

\end{lemma}

\begin{proof}
Consider $\pi \in \Delta(\Ek)$ such that
\[
  \pi_E
  =
  \frac{\exp\{a_E/\lambda\}}
       {\sum_{E'\in\Ek}\exp\{a_{E'}/\lambda\}},
  \qquad E\in\Ek.
\]
Then, for every $E\in\Ek$, we have
\begin{equation}
  \lambda\log\pi_E
  = a_E - \lambda\log\sum_{E'\in\Ek}\exp\{a_{E'}/\lambda\}.
  \label{eqn: identity:pi}
\end{equation}

For any $q \in \Delta(\Ek)$, define the KL divergence between the probability distributions $q$ and $\pi$ by
$\KL(q\|\pi)=\sum_{E\in\Ek}q_E\log(q_E/\pi_E)$. 
Then, substituting the expression for $\lambda\log\pi_E$ from the identity in \eqref{eqn: identity:pi} into the expression for $\lambda \KL(q|\pi)$, we obtain:
\begin{equation}
  \sum_{E\in\Ek}q_Ea_E+\lambda H(q)
  =
  \lambda\log\sum_{E'\in\Ek}\exp\{a_{E'}/\lambda\}
  -
  \lambda\KL(q\|\pi).
  \label{eqn: KL:q}
\end{equation}
Setting $q=\pi$ in the preceding display give
\begin{equation}
  \sum_{E\in\Ek}\pi_Ea_E+\lambda H(\pi)
  =
  \lambda\log\sum_{E'\in\Ek}\exp\{a_{E'}/\lambda\}.
  \label{eqn: KL:pi}
\end{equation} Finally, subtracting the identity in \eqref{eqn: KL:q} from the identity in \eqref{eqn: KL:pi} yields
\[
  \left\{
  \sum_{E\in\Ek}\pi_Ea_E+\lambda H(\pi)
  \right\}
  -
  \left\{
  \sum_{E\in\Ek}q_Ea_E+\lambda H(q)
  \right\}
  =
  \lambda\KL(q\|\pi)\ge 0.
\]
This proves that the maximizer of the map stated in the lemma is $\pi$.
\end{proof}

\subsection{Proofs for Section~\ref{sec:method}}

\begin{proof}[Proof of Lemma~\ref{lem:decomp}]
Define the linear map
\begin{align*}
  M_{j_o}
  &=
  \begin{pmatrix}
    e_{j_o}^\top \\[3pt]
    R_{j_o}\nbracket{I_p - \Sigma e_{j_o}\sigma_{j_o}^{-2}e_{j_o}^\top}
  \end{pmatrix}
  \in \real^{p\times p},
\end{align*}
so that, by construction, $M_{j_o}T = (T_{j_o}, T_{j_o}^\perp)^\top$.
Since $T \sim \mathcal {N}_p\nbracket{\mu, \tfrac{1}{n}\Sigma}$ and $M_{j_o}$ is a
fixed and linear, $M_{j_o}T$ is distributed as a normal variable with mean $M_{j_o}\mu$ and
covariance $\dfrac{1}{n}M_{j_o}\Sigma M_{j_o}^\top$; a direct calculation of these two quantities yields the parameters in the claimed normal distribution.
\end{proof}

\begin{proof}[Proof of Corollary \ref{cor:conditional-validity}]
By Theorem~\ref{thm:pivot}, conditional on $\{\Ehat = E_o\}$ the pivot evaluated at
the true parameter, $U = \Pivot^{E_o}_{\mu_{j_o}}(T_{j_o}; T_{j_o}^\perp, \sigma_{j_o}^2)$, is distributed as a $\mathrm{Uniform}(0,1)$ random variable. 

For Type-I error, under $H_0$ the p-value is $p_{\mu_{j_o}}(T_{j_o};T_{j_o}^\perp,\sigma_{j_o}^2) = 2\min\{U, 1-U\}$; since $2\min\{U,1-U\} \sim
\mathrm{Uniform}(0,1)$, we have $\Pp[\,p_{\mu_{j_o}}(T_{j_o};T_{j_o}^\perp,\sigma_{j_o}^2) \le \alpha \mid \Ehat = E_o\,] =\alpha$.
For coverage,
$\{\mu_{j_o} \in C^{E_o}_{j_o}\} = \{U \in (\alpha/2,\, 1-\alpha/2)\}$, which
has conditional probability $1-\alpha$. 
\end{proof}

\begin{proof}[Proof of Proposition~\ref{prop:smle}]
Since $\log\Lambda^{E_o}(t, t^\perp)$ does not depend on
$\mu_{j_o}$, differentiating~\eqref{eq:sel_loglik} gives the score
\begin{align*}
  \frac{\partial \ell}{\partial \mu_{j_o}}
  = \frac{t - \mu_{j_o}}{\sigma_{j_o}^2}
  - \frac{\partial}{\partial \mu_{j_o}}
    \log \mathcal{Z}(\mu_{j_o}, t^\perp).
\end{align*}
Differentiating under the integral sign in~\eqref{eq:norm_const}, we have that
\begin{align*}
  \frac{\partial}{\partial \mu_{j_o}}
  \log \mathcal{Z}(\mu_{j_o}, t^\perp)
  =
  \frac{1}{\sigma_{j_o}^2}
  \nbracket{e(\mu_{j_o})- \mu_{j_o}.
  },
\end{align*}
From here, it immediately follows that $e(\widehat{\mu}_{j_o}^{E_o}) = t$.

It remains to establish uniqueness of the MLE.
To this end, note that 
\begin{align*}
  e(\mu_{j_o})=
  \frac{
    \int_{-\infty}^{\infty}
    u\, \phi\nbracket{u;\, \mu_{j_o},\, \sigma_{j_o}^2}
    \Lambda^{E_o}\nbracket{u, t^\perp} du
  }{
    \mathcal{Z}(\mu_{j_o}, t^\perp)
  }  =
  \Ee\rbracket{T_{j_o} \mid \Ehat = E_o,\; T_{j_o}^\perp = t^\perp}.
\end{align*}
Therefore, we have that
\begin{align*}
  \frac{d}{d\mu} e(\mu_{j_o})
  =
  \frac{
    \Var\rbracket{T_{j_o} \mid \Ehat = E_o,\; T_{j_o}^\perp = t^\perp}
  }{
    \sigma_{j_o}^2
  }
  > 0,
\end{align*}
which establishes strict monotonicity of the map $e$. 
As $\mu_{j_o} \to \pm\infty$, the normal density function $\phi(u;\mu_{j_o},\sigma_{j_o}^2)$ places its mass at $\pm\infty$, and hence $e(\mu_{j_o})\to\pm\infty$. 
By continuity, the intermediate value theorem gives existence of a solution, while strict monotonicity gives uniqueness. Hence $\ell$ is strictly concave with a unique maximizer.
\end{proof}

%##########################################################################
\section{Proofs and supporting results for Section~\ref{sec:asymptotic}}
\label{app:asymp}
%##########################################################################

%This appendix proves Theorem~\ref{thm:asymp-main}. We first introduce
%the standardized representation of $T_n$ and transfer the asymptotic
%linear representation to standardized coordinates, then define the
%selection weights and the Gaussian comparator, reduce the target weak
%convergence to two relative-difference quantities and show they vanish, and finally collect the supporting derivative and Lindeberg
%bounds.

To prove Theorem~\ref{thm:asymp-main}, we first introduce notations and establish several preliminaries which are used throughout the proofs.

\paragraph{Notations.} Let $\|\cdot\|$ denote the Euclidean norm. For a matrix $A$ regarded as a linear
map, $\|A\|_{\op}=\sup_{\|u\|\le 1}\|Au\|$ is its largest singular value. For the
$k$th derivative $\nabla^{(k)}g(z)$, regarded as a symmetric $k$-linear form $\nabla^{(k)}g(z)[u_1,\dots,u_k]$,
\[
  \big\|\nabla^{(k)}g(z)\big\|_{\op}
   =\sup\Big\{\,\big|\nabla^{(k)}g(z)[u_1,\dots,u_k]\big| : \|u_j\|\le 1,\ j=1,\dots,k\,\Big\};
\]
for $k=1$ this reduces to $\|\nabla g(z)\|$, and for $k=2$ to the spectral norm of
the Hessian. These two conventions agree on symmetric matrices. To consolidate notations, let
\begin{align}
\label{eq:scaled}
  \widetilde T_n
  =\begin{pmatrix}T_{n,j_o}\\T_{n,j_o}^\perp\end{pmatrix}
  \in\real^p,
  \qquad
  \widetilde\mu_n
  =\begin{pmatrix}\mu_{n,j_o}\\\nu_{n,j_o}\end{pmatrix},
  \qquad
  \widetilde\Sigma
  =\begin{pmatrix}\sigma_{j_o}^2 & 0\\0 & \Gamma_{j_o}\end{pmatrix},
\end{align}
where the subscript $n$ emphasizes the dependence of the key inferential statistics and corresponding parameters on the sample size $n$.

\subsection{Preliminaries}
We define the \emph{standardized selection statistics} 
\begin{align}
  \label{eq:zeta}
  \zeta_n
  =\widetilde\Sigma^{-1/2}\big(\widetilde T_n-\widetilde\mu_n\big)
   \in\real^p,
\end{align}
so that $\widetilde T_n=\widetilde\Sigma^{1/2}\zeta_n+\widetilde\mu_n$. In particular, by the block-diagonal structure of $\widetilde\Sigma$, we have
\begin{align}
\label{eq:blocks-zeta}
T_{n,j_o}=\sigma_{j_o}\zeta_{n,1}+\mu_{n,j_o}\in\real,
  \qquad
  T_{n,j_o}^\perp
  =(\Gamma_{j_o})^{1/2}\zeta_{n,2}+\nu_{n,j_o}\in\real^{p-1}.
\end{align}

\begin{proposition}[Transfer of asymptotic linear representation]
\label{prop:clt-decomp}
Under Assumption~\ref{ass:linearizable}, we have
\begin{align}
\label{eq:zeta-alr}
  \zeta_n
  =\frac{1}{\sqrt n}\sum_{i=1}^n\widetilde a_{i,n}Y_{i,n}+\widetilde R_n
  \;\xrightarrow{d}\;\mathcal N_p(0,I_p),
\end{align}
where $\widetilde a_{i,n}=Q_{j_o}a_{i,n}$ and $\widetilde R_n=Q_{j_o}R_n$ for a fixed orthogonal matrix $Q_{j_o}$ depending only on $\Sigma$ and $j_o$, i.e., 
$\|\widetilde a_{i,n}\|_{\op}=\|a_{i,n}\|_{\op}$ and
$\|\widetilde R_n\|=\|R_n\|$.
\end{proposition}

\begin{proof}[Proof of Proposition~\ref{prop:clt-decomp}]
Let $M_{j_o}$ denote the linear map sending $T_n$ to $(T_{n,j_o},T_{n,j_o}^\perp)^\top$, as in proof of Lemma~\ref{lem:decomp}. Then $\widetilde T_n=M_{j_o}T_n$, $\widetilde\mu_n=M_{j_o}\mu_n$, and $\widetilde\Sigma=M_{j_o}\Sigma M_{j_o}^\top$. 
The proof follows by setting $Q_{j_o}=\widetilde\Sigma^{-\frac{1}{2}}M_{j_o}\Sigma^{\frac{1}{2}}$, so that $\widetilde a_{i,n}=Q_{j_o}a_{i,n}$ and $\widetilde R_n=Q_{j_o}R_n$, and noting that $Q_{j_o}$ is orthogonal. Furthermore, weak convergence follows by a direct application of the Lindeberg--Feller central limit theorem.
\end{proof}

\begin{definition}
\label{def:gamma}
Based on the linear representation in Proposition \ref{prop:clt-decomp}, let $\gamma_n=\frac{1}{\sqrt n}\sum_{i=1}^n\widetilde a_{i,n}Y_{i,n}$.
\end{definition}

\begin{definition}
Define the Gaussian counterpart of $\zeta_n$ by 
\begin{align}
\label{eq:Zn}
Z_n=\frac{1}{\sqrt n}\sum_{i=1}^n\widetilde a_{i,n}Z_{i,n},
\end{align} 
where $\{Z_{i,n}\}$ are independent $\mathcal N_d(0,I_d)$ random vectors. 
Since $\frac1n\sum_{i=1}^n\widetilde a_{i,n}\widetilde a_{i,n}^\top=I_p$, it follows that $Z_n\sim\mathcal N_p(0,I_p)$.
\end{definition}

We redefine the selection weights and pivot from Definitions~\ref{def:Lambda} and \ref{def:pivot}, respectively, in terms of the standardized statistic $\zeta_n$.

\begin{definition}[Selection weight in terms of standardized statistics]
\label{def:weights-std}
Based on the selection weight $\Lambda^{E_o}:\real^p\to(0,1)$ from Definition~\ref{def:Lambda}, define the log-selection weight $f:\real^p\to\real$ by
\begin{align}
  \label{eq:f}
  f(u)
  =\log\Lambda^{E_o}(u)
  =\frac{s_{E_o}(r(u))}{\tau\cdot\snoise(r(u))}
   -\log\sum_{E\in\mathcal E_k}
     \exp\!\left\{\frac{s_E(r(u))}{\tau\cdot\snoise(r(u))}\right\},
\end{align}
where $r(u)=r(u_1,u_2)$ is the reconstruction map from Lemma~\ref{lem:reconstruction}. We then define the standardized version of the selection weight $F:\real^p\to(0,1)$ by
\begin{align}
  \label{eq:F}
  F(v)
  =\exp\!\big\{f\big(\widetilde\Sigma^{1/2}v+\widetilde\mu_n\big)\big\}
  =\Lambda^{E_o}\!\big(\widetilde\Sigma^{1/2}v+\widetilde\mu_n\big).
\end{align} 
\end{definition}

\begin{definition}[Pivot in terms of standardized statistics]
\label{def:pivot-std}
For $v=(v_1,v_2^\top)^\top\in\real^p$, redefine the pivot from Definition \ref{def:pivot} as
\begin{align*}
  \cP^{E_o}_{\mu_{n,j_o}}(v;\widetilde\Sigma)
  =
  \frac{
    \int_{-\infty}^{\sigma_{j_o}v_1+\mu_{n,j_o}}
    \phi\big(u;\mu_{n,j_o},\sigma_{j_o}^2\big)\,
    \Lambda^{E_o}\!\big(u,(\Gamma_{j_o})^{1/2}v_2+\nu_{n,j_o}\big)\,du
  }{
    \int_{-\infty}^{\infty}
    \phi\big(u;\mu_{n,j_o},\sigma_{j_o}^2\big)\,
    \Lambda^{E_o}\!\big(u,(\Gamma_{j_o})^{1/2}v_2+\nu_{n,j_o}\big)\,du
  }.
\end{align*}
For the asymptotic analysis, we prove validity of the pivot for any $E_o$, any $j_o\in E_o$, and covariance matrix $\widetilde\Sigma$. 
Henceforth, to simplify notation, we let
\begin{equation}
    \label{eq:pivot-std}
    \cP(v)= \cP^{E_o}_{\mu_{n,j_o}}(v;\widetilde\Sigma).
\end{equation} Note that $\cP(\zeta_n)=\Pivot^{E_o}_{\mu_{n,j_o}}\!\big(T_{n,j_o};T_{n,j_o}^\perp,\sigma_{j_o}^2\big)$.
\end{definition}

\begin{remark}\label{rem:gaussian:comp}
Since $Z_n\sim\mathcal N_p(0,I_p)$, Theorem~\ref{thm:pivot} gives, in the notation of \eqref{eq:pivot-std},
\begin{align*}
\cP(Z_n)
=\Pivot^{E_o}_{\mu_{n,j_o}}\!\big(\sigma_{j_o}Z_{n,1}+\mu_{n,j_o};\,\Gamma_{j_o}^{1/2}Z_{n,2}+\nu_{n,j_o},\,\sigma_{j_o}^2\big)
=\Pivot^{E_o}_{\mu_{n,j_o}}\!\big(T_{n,j_o};\,T_{n,j_o}^\perp,\,\sigma_{j_o}^2\big),
\end{align*}
and consequently $\cP(Z_n)\mid\{\widehat E_n=E_o\}\sim\mathrm{Unif}(0,1)$.
\end{remark}

\subsection{Proof of Theorem \ref{thm:asymp-main}}
 The proof relies on two main results, Propositions~\ref{prop:relative-diff} and~\ref{prop:diff-vanish} stated below, together with the supporting results in Appendices~\ref{app:aux-results} and~\ref{app:deriv-bounds}.

\begin{proof}[Proof of Theorem~\ref{thm:asymp-main}]
Since $\cP(Z_n)\Big\lvert \left\{\widehat E_n=E_o\right\}$ is exactly distributed as $\mathrm{Unif}(0,1)$, it suffices to show that, for any $h\in\mathcal C^3_b(\real)=\{g\in\mathcal C^3(\real): \sup_x|g^{(l)}(x)|<\infty,\ l=0,1,2,3\}$: 
\begin{align}
  \label{eq:target-wc}
  \lim_{n\to\infty}
  \left|
    \Ee\big[h\circ\cP(\zeta_n)\mid\{\widehat E_n=E_o\}\big]
    -\Ee\big[h\circ\cP(Z_n)\mid\{\widehat E_n=E_o\}\big]
  \right|=0,
\end{align}
Following arguments similar to those in \cite{randomizedresponse, bakshi2024selective}, this implies that $\cP(\zeta_n)$ converges weakly to $\mathrm{Unif}(0,1)$.

By Proposition~\ref{prop:relative-diff}, it further suffices to show that $\RD_n^{(1)}\to 0$ and $\RD_n^{(2)}\to0$, as defined therein, to establish \eqref{eq:target-wc}.

Our proof is finally complete by verifying in Proposition~\ref{prop:diff-vanish} that these two conditions hold under Assumptions~\ref{ass:linearizable} and~\ref{ass:score-deriv}.
\end{proof}

\begin{proposition}[Sufficient condition for weak convergence]
\label{prop:relative-diff}
Let $F$ and $\cP$ be defined according to Definition \ref{def:weights-std} and \ref{def:pivot-std}. 
For $h\in\mathcal C^3_b(\real)$, define
\begin{align*}
\begin{gathered}
  \RD_n^{(1)}
  =\frac{\big|\Ee[F(\gamma_n)]-\Ee[F(Z_n)]\big|}{\Ee[F(Z_n)]},\ \RD_n^{(2)}
  =\frac{\big|\Ee[h\circ\cP(\gamma_n)\,F(\gamma_n)]
           -\Ee[h\circ\cP(Z_n)\,F(Z_n)]\big|}{\Ee[F(Z_n)]}.
\end{gathered}
\end{align*}
If $\displaystyle\lim_{n  \to \infty} \RD_n^{(1)}=0$  and $\displaystyle\lim_{n  \to \infty} \RD_n^{(2)}=0$, then~\eqref{eq:target-wc} holds.
\end{proposition}

\begin{proof}[Proof of Proposition~\ref{prop:relative-diff}]
By Proposition~\ref{prop:condexpectation}, the expression whose limit appears on the left-hand side of \eqref{eq:target-wc} is equal to
\begin{align*}
  \left|
    \frac{\Ee[h\circ\cP(\zeta_n)\,F(\zeta_n)]}{\Ee[F(\zeta_n)]}
    -\frac{\Ee[h\circ\cP(Z_n)\,F(Z_n)]}{\Ee[F(Z_n)]}
  \right|.
\end{align*}

By the triangle inequality, we have
\begin{align*}
  \left|
    \frac{\Ee[h\circ\cP(\zeta_n)\,F(\zeta_n)]}{\Ee[F(\zeta_n)]}
    -\frac{\Ee[h\circ\cP(Z_n)\,F(Z_n)]}{\Ee[F(Z_n)]}
  \right|
  &\le
  \left|
    \frac{\Ee[h\circ\cP(\zeta_n)\,F(\zeta_n)]}{\Ee[F(\zeta_n)]}
    -\frac{\Ee[h\circ\cP(\gamma_n)\,F(\gamma_n)]}{\Ee[F(\gamma_n)]}
  \right|\\
  &\;+
  \left|
    \frac{\Ee[h\circ\cP(\gamma_n)\,F(\gamma_n)]}{\Ee[F(\gamma_n)]}
    -\frac{\Ee[h\circ\cP(Z_n)\,F(Z_n)]}{\Ee[F(Z_n)]}
  \right|,
\end{align*}
where the first term on the right-hand side converges to $0$ by Prop.~\ref{prop:alr-remainder}.
It remains to show that the second term on the right-hand side converges to $0$.

Adding and subtracting
$\Ee[h\circ\cP(\gamma_n)\,F(\gamma_n)]/\Ee[F(Z_n)]$ in the second term
and applying the triangle inequality yields the bound
\begin{align*}
  \left|
    \frac{\Ee[h\circ\cP(\gamma_n)\,F(\gamma_n)]}{\Ee[F(\gamma_n)]}
    -\frac{\Ee[h\circ\cP(Z_n)\,F(Z_n)]}{\Ee[F(Z_n)]}
  \right|
  &\le
  \underbrace{
    \left|
      \frac{\Ee[h\circ\cP(\gamma_n)\,F(\gamma_n)]}{\Ee[F(\gamma_n)]}
      -\frac{\Ee[h\circ\cP(\gamma_n)\,F(\gamma_n)]}{\Ee[F(Z_n)]}
    \right|
  }_{(\mathrm{term}_1)}\\
  &+ \underbrace{
    \left|
      \frac{\Ee[h\circ\cP(\gamma_n)\,F(\gamma_n)]}{\Ee[F(Z_n)]}
      -\frac{\Ee[h\circ\cP(Z_n)\,F(Z_n)]}{\Ee[F(Z_n)]}
    \right|
  }_{(\mathrm{term}_2)}.
\end{align*}
For $h\in\mathcal C^3_b(\real)$, write
$\|h\|_{\mathcal C^3_b}=\max_{0\le l\le 3}\sup_{x\in\real}|h^{(l)}(x)|$. 
To conclude the proof, note that
\begin{align*}
  (\mathrm{term}_1)
  &=\frac{\Ee[h\circ\cP(\gamma_n)\,F(\gamma_n)]}{\Ee[F(\gamma_n)]}
   \times
   \frac{\big|\Ee[F(\gamma_n)]-\Ee[F(Z_n)]\big|}{\Ee[F(Z_n)]}
   \le\|h\|_{\mathcal C^3_b}\times\RD_n^{(1)},\\
  (\mathrm{term}_2)
  &=\RD_n^{(2)}.
\end{align*}
Here, in bounding $(\mathrm{term}_1)$, we use Proposition~\ref{prop:condexpectation} to note that
$\dfrac{\Ee[h\circ\cP(\gamma_n)\,F(\gamma_n)]}{\Ee[F(\gamma_n)]}
=\Ee[h\circ\cP(\gamma_n)\mid\{\widehat E_n=E_o\}]$, which is trivially bounded by $\|h\|_{\mathcal C^3_b}$. 

This yields
$$
\left|
    \frac{\Ee[h\circ\cP(\zeta_n)\,F(\zeta_n)]}{\Ee[F(\zeta_n)]}
    -\frac{\Ee[h\circ\cP(Z_n)\,F(Z_n)]}{\Ee[F(Z_n)]}
  \right| \le \|h\|_{\mathcal C^3_b}\times \RD_n^{(1)}+\RD_n^{(2)},
$$
which, in limit, equals $0$ if $\RD_n^{(1)}\to0$ and $\RD_n^{(2)}\to 0$.
\end{proof}

\begin{proposition}[]
\label{prop:diff-vanish}
Let $\RD_n^{(1)}$ and $\RD_n^{(2)}$ be defined according to Proposition \ref{prop:relative-diff}.
Under Assumptions~\ref{ass:linearizable} and~\ref{ass:score-deriv}, we have
   $$
   \lim_{n  \to \infty} \RD_n^{(1)}=0, \quad \lim_{n  \to \infty}  \RD_n^{(2)}=0. 
   $$ 
\end{proposition}

\begin{proof}[Proof of Proposition~\ref{prop:diff-vanish}]
We obtain upper bounds on the numerators of $\RD_n^{(1)}$ and $\RD_n^{(2)}$ in \eqref{bounds:num}, and a lower bound on their common denominator in \eqref{bound:den}. 
Combining these bounds yields
\begin{align*}
  \RD_n^{(1)}\le\frac{C_{\mathrm{Lin}}}{c_0\sqrt n}\to0,
  \qquad
  \RD_n^{(2)}\le\frac{C_{\mathrm{Lin}}}{c_0\sqrt n}\to0.
\end{align*}

\noindent\emph{Bounds on the numerators.} By Corollary~\ref{cor:lindeberg-final}, we have 
\begin{equation}
\begin{aligned}
  \big|\Ee[F(\gamma_n)]-\Ee[F(Z_n)]\big|
  &\le\frac{C_{\mathrm{Lin}}}{\sqrt n}\exp\{f(\widetilde\mu_n)\},\\
  \big|\Ee[h\circ\cP(\gamma_n)\,F(\gamma_n)]
       -\Ee[h\circ\cP(Z_n)\,F(Z_n)]\big|
  &\le\frac{C_{\mathrm{Lin}}}{\sqrt n}\exp\{f(\widetilde\mu_n)\}.
\end{aligned}
\label{bounds:num}
\end{equation}

\medskip\noindent\emph{Bound on the denominator.} 
Note $Z_n\sim\mathcal N_p(0,I_p)$ exactly and 
$F(Z_n)=\exp\{f(\widetilde\Sigma^{1/2}Z_n+\widetilde\mu_n)\}$ by~\eqref{eq:F}. Applying the mean-value theorem to $f$ along the segment joining
$\widetilde\mu_n$ and $\widetilde\Sigma^{1/2}Z_n+\widetilde\mu_n$, and using
$\|\nabla f\|_{\op}\le c_f^{(1)}$ from~\eqref{eq:f-bound} together with
$\|\widetilde\Sigma^{1/2}Z_n\|\le\lambda_{\max}(\widetilde\Sigma)^{1/2}\|Z_n\|$,
gives
\begin{align*}
  f\big(\widetilde\Sigma^{1/2}Z_n+\widetilde\mu_n\big)
  \ge f(\widetilde\mu_n)
     -c_f^{(1)}\lambda_{\max}(\widetilde\Sigma)^{1/2}\|Z_n\|.
\end{align*}
Exponentiating, taking expectations, and applying Jensen's inequality to the
convex map $x\mapsto e^{-c_f^{(1)}\lambda_{\max}(\widetilde\Sigma)^{1/2}x}$
then yields
\begin{equation}
\begin{aligned}
  \Ee[F(Z_n)]
  &\ge\exp\{f(\widetilde\mu_n)\}\,
       \exp\!\big\{-c_f^{(1)}\lambda_{\max}(\widetilde\Sigma)^{1/2}\,
       \Ee\|Z_n\|\big\}\\
   &\ge\exp\{f(\widetilde\mu_n)\}\,
       \exp\!\big\{-c_f^{(1)}\lambda_{\max}(\widetilde\Sigma)^{1/2}\sqrt p\big\},
\end{aligned}
\label{bound:den}
\end{equation}
using $\Ee\|Z_n\|\le(\Ee\|Z_n\|^2)^{1/2}=\sqrt p$. 
Hence
$\Ee[F(Z_n)]\ge c_0\exp\{f(\widetilde\mu_n)\}$ for a strictly positive
constant $c_0$ independent of $n$. 
\end{proof}

\subsection{Auxiliary results}
\label{app:aux-results}

\begin{proposition}
\label{prop:alr-remainder}
Let $F$ and $\cP$ be defined according to Definition \ref{def:weights-std} and \ref{def:pivot-std}. 
Under Assumptions~\ref{ass:linearizable} and~\ref{ass:score-deriv}, for
every $h\in\mathcal C^3_b(\real)$, we have
\begin{align}
  \label{eq:alr-swap}
  \lim_{n\to\infty}
  \left|
    \frac{\Ee[h\circ\cP(\zeta_n)\,F(\zeta_n)]}{\Ee[F(\zeta_n)]}
    -\frac{\Ee[h\circ\cP(\gamma_n)\,F(\gamma_n)]}{\Ee[F(\gamma_n)]}
  \right|=0.
\end{align}
\end{proposition}

\begin{proof}[Proof of Proposition~\ref{prop:alr-remainder}]
Let $g_1(z)=F(z)$ and $g_2(z)=h\circ\cP(z)\cdot F(z)$.
%as in
%Proposition~\ref{prop:deriv-bounds}, and that
%$C_R=c_f^{(1)}\lambda_{\max}(\widetilde\Sigma)^{1/2}$ is the
%first-order derivative rate appearing in
%Proposition~\ref{prop:deriv-bounds}\,(i); the exponential-moment
%condition of Assumption~\ref{ass:linearizable} is imposed at rate
%$4C_R$. 
Adding and subtracting $\Ee[g_2(\gamma_n)]/\Ee[F(\zeta_n)]$
and applying the triangle inequality yields \begin{align*}
  \left|
    \frac{\Ee[g_2(\zeta_n)]}{\Ee[F(\zeta_n)]}
    -\frac{\Ee[g_2(\gamma_n)]}{\Ee[F(\gamma_n)]}
  \right|
  \le
  \underbrace{
    \frac{\Ee\big[|g_2(\zeta_n)-g_2(\gamma_n)|\big]}{\Ee[F(\zeta_n)]}
  }_{(\mathrm I)}
  +
  \underbrace{
    \frac{\Ee[g_2(\gamma_n)]}{\Ee[F(\gamma_n)]}\cdot
    \frac{\big|\Ee[F(\gamma_n)]-\Ee[F(\zeta_n)]\big|}{\Ee[F(\zeta_n)]}
  }_{(\mathrm{II})}.
\end{align*}

\medskip\noindent\textit{Increment bound.}
By the mean-value theorem applied along the segment joining $\gamma_n$ and
$\zeta_n=\gamma_n+\widetilde R_n$, together with the bound on the first-order derivative
$$\|\nabla g_j(z)\|_{\op}\le\widetilde C_j^{(1)}
\exp\{f(\widetilde\mu_n)+C_R\|z\|\}$$
from Proposition~\ref{prop:deriv-bounds} (i), we obtain
\begin{align*}
  |g_j(\zeta_n)-g_j(\gamma_n)|
  \le\widetilde C_j^{(1)}\exp\{f(\widetilde\mu_n)\}\,
     \exp\!\big\{C_R(\|\gamma_n\|+\|\widetilde R_n\|)\big\}\,
     \|\widetilde R_n\|,
  \qquad j=1,2.
\end{align*}
After taking expectations on both sides, we apply H\"older's inequality to the bound on the right-hand side with exponents $(\tfrac12,\tfrac14,\tfrac14)$ to obtain
\begin{align*}
  \Ee\big[|g_j(\zeta_n)-g_j(\gamma_n)|\big]
  \le\widetilde C_j^{(1)}\exp\{f(\widetilde\mu_n)\}\,
     \Ee\big[e^{2C_R\|\gamma_n\|}\big]^{1/2}\,
     \Ee\big[e^{4C_R\|\widetilde R_n\|}\big]^{1/4}\,
     \Ee\big[\|\widetilde R_n\|^4\big]^{1/4}.
\end{align*}

Next, we bound the three moments on the right-hand side of the previous display.
\begin{enumerate}
    \item Since $\gamma_n=\frac{1}{\sqrt n}\sum_i\widetilde a_{i,n}Y_{i,n}$ is uniformly sub-Gaussian under Assumption~\ref{ass:linearizable}, we have
$\sup_n\Ee[e^{2C_R\|\gamma_n\|}]<\infty$.
    \item Since
$\|\widetilde R_n\|=\|R_n\|$ by Proposition~\ref{prop:clt-decomp}, and the moment bound from Assumption~\ref{ass:linearizable}\,(i), we have 
\begin{align*}
  \Ee\big[e^{4C_R\|\widetilde R_n\|}\big]
  =\Ee\big[e^{4C_R\|R_n\|}\big]
  \le\sup_n\Ee\big[e^{4C_R\|R_n\|}\big]<\infty.
\end{align*}
   \item Since
$\sup_n\Ee[e^{4C_R\|R_n\|}]<\infty$ and $R_n=o_p(1)$ under
Assumption~\ref{ass:linearizable}, the family $\{\|R_n\|^4\}_n$ is
uniformly integrable; combined with $\|R_n\|\xrightarrow{p}0$ this gives
$\Ee[\|\widetilde R_n\|^4]=\Ee[\|R_n\|^4]\to0$.
\end{enumerate}

Let $\eta_n:=\Ee[|\widetilde R_n|^4]^{1/4}$. 
Absorbing the bounds in (i) and (ii) into a constant $C<\infty$, independent of $n$, we thus obtain
\begin{align}
  \label{eq:mvt-final}
  \Ee\big[|g_j(\zeta_n)-g_j(\gamma_n)|\big]
  \le\widetilde C_j^{(1)}\,C\,\exp\{f(\widetilde\mu_n)\}\,\eta_n,
  \qquad j=1,2.
\end{align}

\medskip\noindent\textit{Lower bound on denominator.}
By \eqref{eq:F}, $F(\zeta_n)=\exp\{f(\widetilde\Sigma^{1/2}\zeta_n+\widetilde\mu_n)\}$.
Applying the mean-value theorem to $f$ along the segment from $\widetilde\mu_n$ to $\widetilde\Sigma^{1/2}\zeta_n+\widetilde\mu_n$, together with the gradient bound $\|\nabla f\|_{\op}\le c_f^{(1)}$ from Proposition~\ref{prop:f_bounded_deriv},
\[
f\big(\widetilde\Sigma^{1/2}\zeta_n+\widetilde\mu_n\big)
\ge f(\widetilde\mu_n)-c_f^{(1)}\big\|\widetilde\Sigma^{1/2}\zeta_n\big\|
\ge f(\widetilde\mu_n)-C_R\|\zeta_n\|,
\]
where $C_R=c_f^{(1)}\lambda_{\max}(\widetilde\Sigma)^{1/2}$ as in Proposition~\ref{prop:deriv-bounds}. Exponentiating,
\[
F(\zeta_n)\ge\exp\{f(\widetilde\mu_n)\}\exp\{-C_R\|\zeta_n\|\}.
\]Then, taking expectations and applying Jensen's inequality to the convex map $x\mapsto e^{-C_R x}$, we obtain
$$
\Ee[F(\zeta_n)]
  \ge\exp\{f(\widetilde\mu_n)\}\exp\{-C_R\,\Ee\|\zeta_n\|\}.
$$
Since $\zeta_n=\gamma_n+\widetilde R_n\xrightarrow{d}\mathcal N_p(0,I_p)$ and the exponential-moment bounds in Assumption~\ref{ass:linearizable} (sub-Gaussianity of $\gamma_n$ from (ii); $\sup_n\Ee[e^{4C_R\|R_n\|}]<\infty$ from (i)) make $\{\|\zeta_n\|\}$ uniformly integrable, we have $\Ee\|\zeta_n\|\to\Ee\|Z\|\le\sqrt p$ for $Z\sim\mathcal N_p(0,I_p)$, using $\Ee\|Z\|\le(\Ee\|Z\|^2)^{1/2}=\sqrt p$. Hence $\Ee\|\zeta_n\|\le\sqrt p+o(1)$, and for all $n$ large enough that the $o(1)$ term is at most $1$,
\begin{align}
  \label{eq:Fn-lower-final}
  \Ee[F(\zeta_n)]\ge c_0\exp\{f(\widetilde\mu_n)\},
  \qquad c_0:=\exp\{-C_R(\sqrt p+1)\}>0.
\end{align}

\medskip\noindent\textit{Bounding $(\mathrm I)$ and $(\mathrm{II})$.}
For $(\mathrm I)$, combining~\eqref{eq:mvt-final} for $j=2$
with~\eqref{eq:Fn-lower-final} yields:
\begin{align*}
  (\mathrm I)
  \le\frac{\widetilde C_2^{(1)}C\exp\{f(\widetilde\mu_n)\}\,\eta_n}
          {c_0\exp\{f(\widetilde\mu_n)\}}
  =\frac{\widetilde C_2^{(1)}C}{c_0}\,\eta_n\to0.
\end{align*}
For $(\mathrm{II})$, since $|h\circ\cP|\le\|h\|_{\mathcal C^3_b}$ and
$g_2=(h\circ\cP) \cdot F$, we have
$$\Ee[g_2(\gamma_n)]/\Ee[F(\gamma_n)]\le\|h\|_{\mathcal C^3_b}.$$
Now, combining~\eqref{eq:mvt-final} for $j=1$ with~\eqref{eq:Fn-lower-final} yields
\begin{align*}
  (\mathrm{II})
  \le\|h\|_{\mathcal C^3_b}\cdot
     \frac{\widetilde C_1^{(1)}C\exp\{f(\widetilde\mu_n)\}\,\eta_n}
          {c_0\exp\{f(\widetilde\mu_n)\}}
  =\frac{\|h\|_{\mathcal C^3_b}\,\widetilde C_1^{(1)}C}{c_0}\,\eta_n\to0.
\end{align*}
\end{proof}

\begin{proposition}
\label{prop:condexpectation}
For every $h\in\mathcal C^3_b(\real)$, it holds that
\begin{align}
  \label{eq:cond-zeta}
  \Ee\big[h\circ\cP(\zeta_n)\mid\{\Ehat_n=E_o\}\big]
  &=\frac{\Ee\big[h\circ\cP(\zeta_n)\,F(\zeta_n)\big]}
         {\Ee\big[F(\zeta_n)\big]},\\
  \label{eq:cond-Z}
  \Ee\big[h\circ\cP(Z_n)\mid\{\Ehat_n=E_o\}\big]
  &=\frac{\Ee\big[h\circ\cP(Z_n)\,F(Z_n)\big]}
         {\Ee\big[F(Z_n)\big]},
\end{align}
where, in~\eqref{eq:cond-Z}, $\{\Ehat_n=E_o\}$ denotes the selection event arising from an independent application of the exponential mechanism, defined in \eqref{eq:expmech}, to the Gaussian statistic $\widetilde\Sigma^{1/2}Z_n+\widetilde\mu_n$, and $F$ and $\cP$ are as defined in Definition \ref{def:weights-std} and \ref{def:pivot-std}. 
\end{proposition}

\begin{proof}[Proof of Proposition~\ref{prop:condexpectation}]
We prove~\eqref{eq:cond-zeta}; the identity in \eqref{eq:cond-Z} follows by
an identical argument.

By the definition of conditional expectation,
\begin{align}
  \label{eq:cond-ratio}
  \Ee\big[h\circ\cP(\zeta_n)\mid\{\Ehat_n=E_o\}\big]
  =\dfrac{\Ee\big[h\circ\cP(\zeta_n)\,\one\{\Ehat_n=E_o\}\big]}{\Pp(\Ehat_n=E_o)}.
\end{align}
By Lemma~\ref{lem:lambda_prob} and the definition of $F$ in~\eqref{eq:F}, we obtain
\begin{align}
  \label{eq:select-prob}
  \Pp\big(\Ehat_n=E_o\mid\zeta_n\big)
  =\Pp\big(\Ehat_n=E_o\mid T_n=\widetilde\Sigma^{1/2}\zeta_n
     +\widetilde\mu_n\big)
  =\Lambda^{E_o}\big(\widetilde\Sigma^{1/2}\zeta_n+\widetilde\mu_n\big)
  =F(\zeta_n).
\end{align}

Applying the tower property together with~\eqref{eq:select-prob} yields
\begin{align*}
  \Ee\big[h\circ\cP(\zeta_n)\,\one\{\Ehat_n=E_o\}\big]
  &=\Ee\Big[h\circ\cP(\zeta_n)\,
      \Ee\big[\one\{\Ehat_n=E_o\}\mid\zeta_n\big]\Big]
   =\Ee\big[h\circ\cP(\zeta_n)\,F(\zeta_n)\big],\\
  \Pp(\Ehat_n=E_o)
  &=\Ee\big[\Ee[\one\{\Ehat_n=E_o\}\mid\zeta_n]\big]
   =\Ee\big[F(\zeta_n)\big].
\end{align*}
Substituting these into~\eqref{eq:cond-ratio} gives~\eqref{eq:cond-zeta}.
\end{proof}

Lemma~\ref{lem:lindeberg-swap}, stated below, is a triangular-array, matrix-weighted version of the Lindeberg swap principle \citep{chatterjee2006generalization}, specialized to test functions with bounded third derivatives. 
The proof tracks the dependence on $\max_i\|\widetilde a_{i,n}\|_{\op}$ explicitly.
This yields the $O(n^{-1/2})$ rate in Corollary~\ref{cor:lindeberg-final}, which is used to bound the numerators in the proof of Theorem~\ref{thm:asymp-main}.

\begin{lemma}[Multivariate Lindeberg bound]
\label{lem:lindeberg-swap}
Let $\widetilde a_{i,n}\in\real^{p\times d}$ satisfy
$\frac1n\sum_{i=1}^n\widetilde a_{i,n}\widetilde a_{i,n}^\top=I_p$, and let
$\{Y_{i,n}\}_{i\in[n]}$, $\{Z_{i,n}\}_{i\in[n]}$ be triangular arrays of
independent random vectors in $\real^d$ with $\Ee[Y_{i,n}]=0_d$,
$\mathrm{Cov}(Y_{i,n})=I_d$, and $Z_{i,n}\sim\mathcal N_d(0,I_d)$.
Assume moreover that the two arrays are mutually independent.
%Define
%\begin{align}
%  \label{eq:lindeberg-sums}
%  \gamma_n=\frac{1}{\sqrt n}\sum_{i=1}^n\widetilde a_{i,n}Y_{i,n},
%  \qquad
%  Z_n=\frac{1}{\sqrt n}\sum_{i=1}^n\widetilde a_{i,n}Z_{i,n},
%\end{align}
Then, for every $g\in\mathcal C^3(\real^p)$, we have
\begin{align*}
  \big|\Ee[g(\gamma_n)]-\Ee[g(Z_n)]\big|
  \le\frac{(m_3+m_3')\,p}{6\sqrt n}\,
     \max_{i\in[n]}\|\widetilde a_{i,n}\|_{\op}\,
     \max_{i\in[n]}\sup_{\alpha\in(0,1)}
     \Ee\big[\|\nabla^{(3)}g(\xi^\alpha_{i,n})\|_{\op}\big],
\end{align*}
where, for each $i\in[n]$ and $\alpha\in[0,1]$, 
\begin{align}
  \label{eq:interpolation}
  \xi^{\alpha}_{i,n}=W_i+\frac{\alpha}{\sqrt n}\widetilde a_{i,n}Y_{i,n},
\end{align}
with $W_i=\frac{1}{\sqrt n}\Big(\sum_{l<i}\widetilde a_{l,n}Z_{l,n}
     +\sum_{l>i}\widetilde a_{l,n}Y_{l,n}\Big)$, and
$m_3=\sup_n\max_{i\in[n]}\Ee\|Y_{i,n}\|^3<\infty$,
and 
$m_3'=\Ee\|Z_{1,n}\|^3<\infty$.
\end{lemma}

\begin{proof}[Proof of Lemma~\ref{lem:lindeberg-swap}]
\textit{Telescoping.}
For each $i\in[n]$, define
\begin{align*}
\begin{gathered}
  \Phi_{i,n}
  =\Ee\Big[g\Big(W_i+\tfrac{1}{\sqrt n}\widetilde a_{i,n}Y_{i,n}\Big)\Big],
  \qquad
  \bar\Phi_{i,n}
  =\Ee\Big[g\Big(W_i+\tfrac{1}{\sqrt n}\widetilde a_{i,n}Z_{i,n}\Big)\Big],\\
  S_i=\frac{1}{\sqrt n}\big(\sum_{l\le i}\widetilde a_{l,n}Z_{l,n}
+\sum_{l>i}\widetilde a_{l,n}Y_{l,n}\big), \text{ for  } i=0,\dots,n,
\end{gathered}
\end{align*}
so that
$S_0=\gamma_n$ and $S_n=Z_n$, and 
$\Ee[g(S_{i-1})]-\Ee[g(S_i)]=\Phi_{i,n}-\bar\Phi_{i,n}$.
Furthermore, summing over theses differences gives
\begin{align}
  \label{eq:telescope}
  \Ee[g(\gamma_n)]-\Ee[g(Z_n)]
  =\sum_{i=1}^n\big(\Phi_{i,n}-\bar\Phi_{i,n}\big).
\end{align}

\medskip\noindent\textit{Taylor expansion with integral remainder.}
Fix $i\in[n]$ and $v\in\real^d$, and set $h_v=\tfrac{1}{\sqrt n}\widetilde a_{i,n}v\in\real^p$. 
By Taylor's theorem with integral remainder applied to $g$ about $W_i$, and writing $\nabla^{(k)}g(z)[u_1,\dots,u_k]$ for the symmetric $k$-linear form given by the $k$th derivative of $g$ at $z$ (so that $\nabla^{(1)}g(z)[u]=\nabla g(z)^\top u$),
\begin{align}
  \label{eq:taylor}
  g(W_i+h_v)
  =g(W_i)+\nabla^{(1)}g(W_i)[h_v]
   +\tfrac12\nabla^{(2)}g(W_i)[h_v,h_v]+R_i(v),
\end{align}
with
\begin{align}
  \label{eq:remainder}
  R_i(v)
  =\frac12\int_0^1(1-t)^2\,
   \nabla^{(3)}g\Big(W_i+\tfrac{t}{\sqrt n}\widetilde a_{i,n}v\Big)
   \big[h_v,h_v,h_v\big]\,dt.
\end{align}
Note that setting $v=Y_{i,n}$ and $\alpha=t$ gives
$W_i+\tfrac{t}{\sqrt n}\widetilde a_{i,n}Y_{i,n}=\xi^{\alpha}_{i,n}$,
the interpolation point defined in~\eqref{eq:interpolation}.

\medskip\noindent\textit{Evaluating $\Phi_{i,n}-\bar\Phi_{i,n}$.}
Taking expectations in~\eqref{eq:taylor} conditionally on $W_i$ (independent
of both $Y_{i,n}$ and $Z_{i,n}$) and subtracting,
\begin{align*}
  \Phi_{i,n}-\bar\Phi_{i,n}
  &=\underbrace{\Ee\big[\nabla^{(1)}g(W_i)[h_{Y_{i,n}}]\big]
               -\Ee\big[\nabla^{(1)}g(W_i)[h_{Z_{i,n}}]\big]}_{=\,0
               \ (\text{matching first moments})}\\
  &\quad
   +\tfrac12\underbrace{\Big(\Ee\big[\nabla^{(2)}g(W_i)[h_{Y_{i,n}},h_{Y_{i,n}}]\big]
               -\Ee\big[\nabla^{(2)}g(W_i)[h_{Z_{i,n}},h_{Z_{i,n}}]\big]\Big)}_{=\,0
               \ (\text{matching second moments})}\\
  &\quad
   +\Ee[R_i(Y_{i,n})]-\Ee[R_i(Z_{i,n})].
\end{align*}
Since $W_i$ is independent of both $Y_{i,n}$ and $Z_{i,n}$, taking
expectations in~\eqref{eq:taylor} conditionally on $W_i$:
the order-$0$ term $\Ee[g(W_i)]$ appears in both $\Phi_{i,n}$ and
$\bar\Phi_{i,n}$ and cancels; the order-$1$ term vanishes because
$\Ee[\nabla^{(1)}g(W_i)[h_{Y_{i,n}}]]
=\Ee[\nabla^{(1)}g(W_i)\,\widetilde a_{i,n}]\,\Ee[Y_{i,n}]=0$ and
identically for $Z_{i,n}$; and the order-$2$ term satisfies
\begin{align*}
  \Ee\big[\nabla^{(2)}g(W_i)[h_{Y_{i,n}},h_{Y_{i,n}}]\big]
  =\frac1n\mathrm{tr}\big(\nabla^{(2)}g(W_i)\,
    \widetilde a_{i,n}\mathrm{Cov}(Y_{i,n})\widetilde a_{i,n}^\top\big)
  =\frac1n\mathrm{tr}\big(\nabla^{(2)}g(W_i)\,
    \widetilde a_{i,n}\widetilde a_{i,n}^\top\big),
\end{align*}
which is identical for $Z_{i,n}$ since
$\mathrm{Cov}(Y_{i,n})=I_d=\mathrm{Cov}(Z_{i,n})$. Hence
$\Phi_{i,n}-\bar\Phi_{i,n}=\Ee[R_i(Y_{i,n})]-\Ee[R_i(Z_{i,n})]$.

\medskip\noindent\textit{Bounding the remainders from the Taylor series expansion.}
Applying the multilinear bound
$|\nabla^{(3)}g(z)[h,h,h]|\le\|\nabla^{(3)}g(z)\|_{\op}\|h\|^3$ and
$\|\widetilde a_{i,n}v\|\le\|\widetilde a_{i,n}\|_{\op}\|v\|$
to~\eqref{eq:remainder}, together with $\int_0^1(1-t)^2dt=\tfrac13$, we obtain
\begin{align*}
  |R_i(v)|
  \le\frac{\|\widetilde a_{i,n}\|_{\op}^3\|v\|^3}{6\,n^{3/2}}
     \sup_{\alpha\in(0,1)}\|\nabla^{(3)}g(\xi^\alpha_{i,n})\|_{\op}.
\end{align*}
Taking expectations with $v=Y_{i,n}$ and using the fact that
$W_i\perp Y_{i,n}$, we have
\begin{align}
  \label{eq:Ri-bound}
  \big|\Ee[R_i(Y_{i,n})]\big|
  \le\frac{\|\widetilde a_{i,n}\|_{\op}^3\,m_3}{6\,n^{3/2}}
     \sup_{\alpha\in(0,1)}\Ee\big[\|\nabla^{(3)}g(\xi^\alpha_{i,n})\|_{\op}\big],
\end{align}
Applying a similar argument for $v=Z_{i,n}$, we finally obtain
\begin{align*}
  |\Phi_{i,n}-\bar\Phi_{i,n}|
  \le\frac{(m_3+m_3')\|\widetilde a_{i,n}\|_{\op}^3}{6\,n^{3/2}}
     \sup_{\alpha\in(0,1)}\Ee\big[\|\nabla^{(3)}g(\xi^\alpha_{i,n})\|_{\op}\big].
\end{align*}

\medskip\noindent\textit{Summation and normalization.}
Summing over the terms $|\Phi_{i,n}-\bar\Phi_{i,n}|$ in \eqref{eq:telescope}, we have
\begin{align}
  \label{eq:general-bound}
  \big|\Ee[g(\gamma_n)]-\Ee[g(Z_n)]\big|
  \le\frac{m_3+m_3'}{6\sqrt n}\,
     \Big(\frac1n\sum_{i=1}^n\|\widetilde a_{i,n}\|_{\op}^3\Big)
     \max_{i\in[n]}\sup_{\alpha\in(0,1)}
     \Ee\big[\|\nabla^{(3)}g(\xi^\alpha_{i,n})\|_{\op}\big].
\end{align}
Finally, using the fact that $\|\cdot\|_{\op}\le\|\cdot\|_F$, together with
$\frac1n\sum_i\widetilde a_{i,n}\widetilde a_{i,n}^\top=I_p$, we obtain
$\frac1n\sum_i\|\widetilde a_{i,n}\|_{\op}^2
\le\frac1n\sum_i\|\widetilde a_{i,n}\|_F^2=\mathrm{tr}(I_p)=p$.
Hence,
\begin{align*}
  \frac1n\sum_{i=1}^n\|\widetilde a_{i,n}\|_{\op}^3
  \le\max_{i\in[n]}\|\widetilde a_{i,n}\|_{\op}\cdot
     \frac1n\sum_{i=1}^n\|\widetilde a_{i,n}\|_{\op}^2
  \le p\,\max_{i\in[n]}\|\widetilde a_{i,n}\|_{\op}.
\end{align*}
Substituting into~\eqref{eq:general-bound} leads to our claim.
\end{proof}

\begin{corollary}[Lindeberg bound]
\label{cor:lindeberg-final}
For $h\in\mathcal C^3_b(\real)$, let $g_1(z)=F(z)$ and $g_2(z)=h\circ\cP(z)\cdot F(z)$.
Under Assumptions~\ref{ass:linearizable} and~\ref{ass:score-deriv}, there exists a constant $C_{\mathrm{Lin}}<\infty$ for which
\begin{align*}
  \big|\Ee[g_j(\gamma_n)]-\Ee[g_j(Z_n)]\big|
  \le\frac{C_{\mathrm{Lin}}}{\sqrt n}\exp\{f(\widetilde\mu_n)\},
  \qquad j=1,2,
\end{align*}
where $f$ is the function specified in Definition \ref{def:weights-std}.
\end{corollary}

\begin{proof}[Proof of Corollary~\ref{cor:lindeberg-final}]
We apply Lemma~\ref{lem:lindeberg-swap} with $g=g_j\in\mathcal C^3(\real^p)$.
It remains to show that
\begin{align}
  \label{eq:xi_bound}
  \sup_n\max_{i\in[n]}\sup_{\alpha\in(0,1)}
  \Ee\big[\|\nabla^{(3)}g_j(\xi^\alpha_{i,n})\|_{\op}\big]
  \le C'\exp\{f(\widetilde\mu_n)\}
\end{align}
for some $C'<\infty$ independent of $n$.
Then, together with $\max_i\|\widetilde a_{i,n}\|_{\op}\le B$ from
Assumption~\ref{ass:linearizable} (since, $\|\widetilde a_{i,n}\|_{\op}=\|a_{i,n}\|_{\op}$), the bound in Lemma~\ref{lem:lindeberg-swap} gives $C_{\mathrm{Lin}}=\frac{(m_3+m_3')\,p\,B\,C'}{6}$.

For the remainder of the proof, we establish \eqref{eq:xi_bound}.

\medskip\noindent\textit{Pointwise bound.}
By Proposition~\ref{prop:deriv-bounds}\,(i), for every $z\in\real^p$, we have
\begin{align*}
  \|\nabla^{(3)}g_j(z)\|_{\op}
  \le\widetilde C_j^{(3)}\exp\{f(\widetilde\mu_n)+C_R\|z\|\}.
\end{align*}
Evaluating this bound at $z=\xi^\alpha_{i,n}$ and taking expectations yields
\begin{align}
  \label{eq:xi_exp_bound}
  \Ee\big[\|\nabla^{(3)}g_j(\xi^\alpha_{i,n})\|_{\op}\big]
  \le\widetilde C_j^{(3)}\exp\{f(\widetilde\mu_n)\}\,
     \Ee\big[\exp\{C_R\|\xi^\alpha_{i,n}\|\}\big].
\end{align}

It remains to obtain a bound on $\Ee[\exp{C_R\|\xi^\alpha_{i,n}\|}]$ that is uniform in $i,\alpha,n$.
This follows from the fact that $\xi^\alpha_{i,n}$ is a sub-Gaussian vector, as shown below.

\medskip\noindent\textit{Sub-Gaussianity of $\xi^\alpha_{i,n}$.}
By~\eqref{eq:interpolation}, we have
\begin{align*}
  \xi^\alpha_{i,n}
  =\frac{1}{\sqrt n}\Big(
     \sum_{l<i}\widetilde a_{l,n}Z_{l,n}
     +\alpha\,\widetilde a_{i,n}Y_{i,n}
     +\sum_{l>i}\widetilde a_{l,n}Y_{l,n}\Big),
\end{align*}
which is a sum of independent mean-zero terms. 
Note that, for any unit  vector $u\in\real^p$, the Gaussian summands $\tfrac{1}{\sqrt n}u^\top\widetilde a_{l,n}Z_{l,n}$ are sub-Gaussian with parameter  $\|\widetilde a_{l,n}\|_{\op}^2/n$. 
Similarly, by the uniform sub-Gaussianity condition in Assumption~\ref{ass:linearizable}, the non-Gaussian summands $\tfrac{\alpha}{\sqrt n}u^\top\widetilde a_{l,n}Y_{l,n}$ are sub-Gaussian with parameter $K_0^2\|\widetilde a_{l,n}\|_{\op}^2/n$. 
Since $\alpha\le1$ and the summands are independent, it follows that $u^\top\xi^\alpha_{i,n}$ is sub-Gaussian with parameter at most
\begin{align*}
  \frac{K_0^2\vee1}{n}\sum_{l=1}^n\|\widetilde a_{l,n}\|_{\op}^2
  \le\frac{K_0^2\vee1}{n}\sum_{l=1}^n\|\widetilde a_{l,n}\|_F^2
  =(K_0^2\vee1)\,\mathrm{tr}\Big(\tfrac1n\sum_l
     \widetilde a_{l,n}\widetilde a_{l,n}^\top\Big)
  =(K_0^2\vee1)\,p
  =:\sigma_\xi^2.
\end{align*}
Note that $\sigma_\xi^2$ does not depend on $i,\alpha,n$.

\medskip\noindent\textit{Norm moment-generating bound.}
Because $z\mapsto\|z\|$ is $1$-Lipschitz, the centered norm
$\|\xi^\alpha_{i,n}\|-\Ee\|\xi^\alpha_{i,n}\|$ is also sub-Gaussian with the
same parameter $\sigma_\xi^2$, and $\Ee\|\xi^\alpha_{i,n}\|
\le(\Ee\|\xi^\alpha_{i,n}\|^2)^{1/2}\le\sigma_\xi$ by Jensen's inequality. 
Thus, we have
\begin{align*}
  \Ee\big[\exp\{C_R\|\xi^\alpha_{i,n}\|\}\big]
  &\le\exp\{C_R\,\Ee\|\xi^\alpha_{i,n}\|\}\,
     \Ee\big[\exp\{C_R(\|\xi^\alpha_{i,n}\|-\Ee\|\xi^\alpha_{i,n}\|)\}\big]\\
  &\le\exp\Big\{C_R\sigma_\xi+\frac{C_R^2\sigma_\xi^2}{2}\Big\}
  =:C_\xi<\infty,
\end{align*}
uniformly in $i,\alpha,n$, where $C_\xi$ depends only on $C_R$, $K_0$,
and $p$. 
This yields our claim.
\end{proof}

\subsection{Bounds on derivatives.}
\label{app:deriv-bounds}
The bounds invoked above are collected here. Throughout, unsubscripted $\nabla^{(l)}$ denotes differentiation in the standardized argument $v$; $\nabla_T$ denotes differentiation in $T=r(v)$, related by the fixed Jacobian $J_r$.

\begin{proposition}[Bounded Derivatives of $f$]
\label{prop:f_bounded_deriv}
Under Assumption~\ref{ass:score-deriv}, with $\tau > 0$ and
$\varepsilon > 0$ fixed, there exist constants
$c_f^{(l)} < \infty$ for $l = 1, 2, 3$ such that $c_f^{(l)} = O(\tau^{-l})$ and
\begin{align}
  \label{eq:f-bound}
  \sup_{v \in \real^p}
  \|\nabla^{(l)} f(v)\|_{\op}
  \leq c_f^{(l)} < \infty, \qquad l=1,2,3.
\end{align} 
\end{proposition}

\begin{proof}[Proof of Proposition~\ref{prop:f_bounded_deriv}]
Throughout, write $T=r(v)$ and
$\beta_E(v)=\tau^{-1}\widetilde s_E(r(v))$, where
$\widetilde s_E=(s_E-\bar s)/\snoise$. By the definition of the
exponential mechanism~\eqref{eq:expmech} and translation invariance of
the log-sum-exp by the common term $\bar s/(\tau\snoise)$,
\begin{align}
  \label{eq:f-beta-clean}
  f(v)=\beta_{E_o}(v)
       -\log\!\sum_{E\in\Ek}\exp\{\beta_E(v)\},
  \qquad
  \Lambda_E(v)=\frac{\exp\{\beta_E(v)\}}{\sum_{E'\in\Ek}\exp\{\beta_{E'}(v)\}}.
\end{align}

\medskip\noindent\emph{Uniform derivative bounds on $\beta_E$.}
We bound $\nabla_v^{(l)}\beta_E$ for $l=1,2,3$ uniformly in $v$ and $E$.
Since $v\mapsto\beta_E(v)$ is the composition
$\tau^{-1}\widetilde s_E\circ r$, the chain rule contributes a factor
$\kappa_r:=\|J_r\|_{\op}$ per order of differentiation, where $J_r$ is
the (fixed) Jacobian of $r$; it therefore suffices to bound the
derivatives of $\widetilde s_E$ with respect to $T$. Differentiating $\widetilde s_E=(s_E-\bar s)/\snoise$ by the quotient
rule gives, at first order,
\begin{align}
  \label{eq:ts-quotient}
  \nabla_T\widetilde s_E
  =\frac{\nabla_T(s_E-\bar s)}{\snoise}
   -\widetilde s_E\,\frac{\nabla_T\snoise}{\snoise},
\end{align}
and the cases $l=2,3$ follow by differentiating~\eqref{eq:ts-quotient}
further: each resulting term is a product of factors
$\nabla_T^{(j)}(s_E-\bar s)$ and $\nabla_T^{(j')}\snoise$ with
$j,j'\le l$, divided by a power of $\snoise$. For the dispersion,
differentiating
$\snoise=\big(|\Ek|^{-1}\sum_E(s_E-\bar s)^2+\varepsilon^2\big)^{1/2}$
and applying Cauchy--Schwarz with $\snoise\ge\varepsilon$ yields
\begin{align}
  \label{eq:snoise-bound}
  \|\nabla_T^{(j)}\snoise\|_{\op}\le\bar c^{(j)},
  \qquad j=1,2,3,
\end{align}
where, at first order, the bound on
$\nabla \snoise=\snoise^{-1} |\Ek|^{-1}\sum_E (s_E-\bar s)\nabla(s_E -\bar s)$
satisfies $\|\nabla\snoise\|_{\op}\le 2c_s^{(1)}=:\bar c^{(1)}$ with the
factor $\snoise^{-1}$ cancelling exactly against
$\big(|\Ek|^{-1}\sum_E(s_E-\bar s)^2\big)^{1/2}\le\snoise$, so that no
$\varepsilon^{-1}$ appears; for $j\ge2$ the constant $\bar c^{(j)}$
additionally involves $\varepsilon^{-(j-1)}$. Substituting these bounds,
together with $|\widetilde s_E|\le\sqrt{|\Ek|}$ and $\snoise\ge\varepsilon$,
into the quotient-rule expression at each order gives
\begin{align}
  \label{eq:beta-deriv-bound}
  \|\nabla_v^{(l)}\beta_E(v)\|_{\op}\le\frac{B_l}{\tau}<\infty,
  \qquad l=1,2,3,\;\;\forall\,v\in\real^p,\;E\in\Ek,
\end{align}
where $B_l$ depends only on $c_s^{(1)},\dots,c_s^{(l)}$,
$\varepsilon^{-l}$, $|\Ek|$, and $\kappa_r^{\,l}$.

\medskip\noindent\emph{Cumulant expansion of $f$.}
Differentiating~\eqref{eq:f-beta-clean} and using that the derivatives of
$\log\sum_E\exp\beta_E$ are the cumulants of $\nabla_v\beta$ under the
softmax law $\Lambda$,
\begin{align}
  \label{eq:f-cumulant-clean}
  \nabla_v^{(l)}f
  =\nabla_v^{(l)}\beta_{E_o}
   -\Ee_\Lambda\!\big[\nabla_v^{(l)}\beta\big]
   -\cK_l,
\end{align}
where $\cK_l$ collects the cumulants of orders $2,\dots,l$ of
$\nabla_v\beta$ under $\Lambda$:
\begin{align*}
  \cK_1=0,
  \qquad
  \cK_2=\mathrm{Cov}_\Lambda[\nabla_v\beta],
  \qquad
  \cK_3=3\,\mathrm{Sym}\,\mathrm{Cov}_\Lambda
        \big[\nabla_v^{(2)}\beta,\nabla_v\beta\big]
        +\kappa_{3,\Lambda}[\nabla_v\beta].
\end{align*}
Since $\Lambda$ is a probability measure on $\Ek$, Jensen's inequality
and~\eqref{eq:beta-deriv-bound} bound each term:
\begin{align*}
  \|\nabla_v^{(l)}\beta_{E_o}\|_{\op}
  \vee\|\Ee_\Lambda[\nabla_v^{(l)}\beta]\|_{\op}
  \le\frac{B_l}{\tau},
  \qquad
  \|\cK_l\|_{\op}
  \le C_l\max_{j\le l}\Big(\frac{B_j}{\tau}\Big)^{\!l/j}
  =O(\tau^{-l}),
\end{align*}
where $C_l$ is a combinatorial constant depending only on $l$ and
$|\Ek|$. Taking $c_f^{(l)}$ to be the sum of these bounds
establishes~\eqref{eq:f-bound}, with $c_f^{(l)}=O(\tau^{-l})$.
\end{proof}

\begin{proposition}[Derivative bounds for $F$ and $h\circ\cP\cdot F$]
\label{prop:deriv-bounds}
For $h\in\mathcal C^3_b(\real)$, let $g_1=F$ and $g_2=h\circ\cP\cdot F$,
and recall $C_R=c_f^{(1)}\lambda_{\max}(\widetilde\Sigma)^{1/2}$. Under
Assumption~\ref{ass:score-deriv}, for $j=1,2$ and $l=0,1,2,3$:
\begin{enumerate}
\item[(i)] there exist constants $\widetilde C^{(l)}_j<\infty$ such that,
uniformly in $z\in\real^p$,
\begin{align}
  \label{eq:deriv-F}
  \|\nabla^{(l)}g_j(z)\|_{\op}
  \le\widetilde C^{(l)}_j\,F(z)
  \le\widetilde C^{(l)}_j
     \exp\!\big\{f(\widetilde\mu_n)+C_R\|z\|\big\};
\end{align}
\item[(ii)] under Assumption~\ref{ass:linearizable}, there exist
constants $C^{(l)}_j<\infty$ such that, for all sufficiently large $n$,
\begin{align}
  \label{eq:deriv-gamma}
  \Ee\big[\|\nabla^{(l)}g_j(\gamma_n)\|_{\op}\big]
  \le C^{(l)}_j\exp\{f(\widetilde\mu_n)\}.
\end{align}
\end{enumerate}
\end{proposition}

\begin{proof}[Proof of Proposition~\ref{prop:deriv-bounds}]
\textit{Part (i), $j=1$ (derivatives of $F$).}
By~\eqref{eq:F}, $F(z)=\exp\{f(\widetilde\Sigma^{1/2}z+\widetilde\mu_n)\}$.
Set $w=\widetilde\Sigma^{1/2}z+\widetilde\mu_n$ and
$\bar v_j=\widetilde\Sigma^{1/2}v_j$. Since $z\mapsto w$ is linear,
Fa\`a di Bruno's formula gives
\begin{align*}
  \nabla^{(1)}F(z)[v_1]
  &=F(z)\,\nabla^{(1)}f(w)[\bar v_1],\\
  \nabla^{(2)}F(z)[v_1,v_2]
  &=F(z)\big(\nabla^{(2)}f(w)[\bar v_1,\bar v_2]
     +\nabla^{(1)}f(w)[\bar v_1]\,\nabla^{(1)}f(w)[\bar v_2]\big),\\
  \nabla^{(3)}F(z)[v_1,v_2,v_3]
  &=F(z)\big(\nabla^{(3)}f(w)[\bar v_1,\bar v_2,\bar v_3]
     +{\textstyle\sum_{\mathrm{sym}}}
       \nabla^{(2)}f(w)[\bar v_a,\bar v_b]\,\nabla^{(1)}f(w)[\bar v_c]\\
  &\qquad\quad
     +\nabla^{(1)}f(w)[\bar v_1]\,\nabla^{(1)}f(w)[\bar v_2]\,
       \nabla^{(1)}f(w)[\bar v_3]\big),
\end{align*}
the symmetric sum running over the three ways of pairing two of the
$\bar v$'s. Bounding each factor by
$|\nabla^{(k)}f(w)[u_1,\dots,u_k]|\le\|\nabla^{(k)}f(w)\|_{\op}
\prod_j\|u_j\|$, then $\|\nabla^{(k)}f(w)\|_{\op}\le c_f^{(k)}$
by~\eqref{eq:f-bound}, and
$\|\bar v_j\|\le\lambda_{\max}(\widetilde\Sigma)^{1/2}\|v_j\|$, and taking
the supremum over unit $v_1,\dots,v_l$,
\begin{align}
  \label{eq:F-bounds-explicit1}
  \|\nabla^{(1)}F(z)\|_{\op}
  &\le c_f^{(1)}\lambda_{\max}(\widetilde\Sigma)^{1/2}\,F(z),\\
  \|\nabla^{(2)}F(z)\|_{\op}
  &\le\big(c_f^{(2)}+(c_f^{(1)})^2\big)
     \lambda_{\max}(\widetilde\Sigma)\,F(z),\nonumber\\
  \|\nabla^{(3)}F(z)\|_{\op}
  &\label{eq:F-bounds-explicit3} \le\big(c_f^{(3)}+3c_f^{(2)}c_f^{(1)}+(c_f^{(1)})^3\big)
     \lambda_{\max}(\widetilde\Sigma)^{3/2}\,F(z).
\end{align}
The first bound in~\eqref{eq:F-bounds-explicit1}  identifies $C_R:=c_f^{(1)}\lambda_{\max}(\widetilde\Sigma)^{1/2}$ as the Lipschitz constant of $\log F$; since $\widetilde\Sigma$ is block-diagonal in $(T_{j_o},T_{j_o}^\perp)$, this equals $c_f^{(1)}\max\{\sigma_{j_o},\|\Gamma_{j_o}^{1/2}\|_{\op}\}$, the constant in Assumption~\ref{ass:linearizable} (and it governs both the increment and
the lower bounds used in Proposition~\ref{prop:alr-remainder}). Each
coefficient of $F(z)$ in~\eqref{eq:F-bounds-explicit3} is a finite
constant depending only on $c_f^{(1)},c_f^{(2)},c_f^{(3)}$ and
$\lambda_{\max}(\widetilde\Sigma)$, which we denote $\widetilde C_1^{(l)}$.
Using $F(z)=\exp\{f(w)\}\le\exp\{f(\widetilde\mu_n)+C_R\|z\|\}$ — the
mean-value theorem applied to $f$, with
$\|w-\widetilde\mu_n\|\le\lambda_{\max}(\widetilde\Sigma)^{1/2}\|z\|$ —
gives~\eqref{eq:deriv-F} for $j=1$.

\medskip
\textit{Part (i), $j=2$ (derivatives of $h\circ\cP\cdot F$).}
We first show $\sup_z\|\nabla^{(l)}\cP(z)\|_{\op}<\infty$. Write
$\cP(z)=N(z)/D(z)$ with
\begin{align*}
  N(z)=\int_{-\infty}^{u_0}
    \phi(u;\widetilde\mu_{n,1},\sigma_{j_o}^2)\,\Lambda^{E_o}(u,v_2)\,du,
  \qquad
  D(z)=\int_{-\infty}^{\infty}
    \phi(u;\widetilde\mu_{n,1},\sigma_{j_o}^2)\,\Lambda^{E_o}(u,v_2)\,du,
\end{align*}
where $u_0=\sigma_{j_o}z_1+\widetilde\mu_{n,1}$ and
$v_2=\Gamma_{j_o}^{1/2}z_2+\widetilde\mu_{n,2}$.

\medskip\noindent\emph{Derivatives in $z_2$.}
Since $\Lambda^{E_o}(u,\cdot)=\exp\{f(\cdot)\}$ and $f$ has bounded
derivatives by~\eqref{eq:f-bound}, the Fa\`a di Bruno argument used for
$F$ gives
$\|\nabla^{(l)}_{z_2}\Lambda^{E_o}(u,v_2)\|_{\op}\le C_l'\,
\Lambda^{E_o}(u,v_2)$ uniformly in $u$. Integrating against $\phi$ yields
$\|\nabla^{(l)}_{z_2}N\|_{\op}\le C_l' N$ and
$\|\nabla^{(l)}_{z_2}D\|_{\op}\le C_l' D$, and the quotient rule together
with $\cP=N/D\in[0,1]$ gives
$\|\nabla^{(l)}_{z_2}\cP(z)\|_{\op}\le C_l''<\infty$.

\medskip\noindent\emph{Derivatives in $z_1$.}
Since $D$ does not depend on $z_1$, $\partial_{z_1}^k\cP
=(\partial_{z_1}^k N)/D$. Viewed as a function of $z_1$, $\cP(z)$ is the
selective distribution function $\Pr_{\mathrm{sel}}[U\le u_0\mid v_2]$ of
the tilted law
\begin{align*}
  p_{\mathrm{sel}}(u\mid v_2)
  =\frac{\phi(u;\widetilde\mu_{n,1},\sigma_{j_o}^2)\,\Lambda^{E_o}(u,v_2)}
        {D(z)},
\end{align*}
so $\partial_{z_1}\cP=\sigma_{j_o}\,p_{\mathrm{sel}}(u_0\mid v_2)$ and the
higher $z_1$-derivatives are $\sigma_{j_o}^k$ times $u$-derivatives of
$p_{\mathrm{sel}}$ evaluated at $u_0$. We bound these uniformly in $z$
\emph{without any lower bound on $D$}, by controlling the selective
density directly. The density $p_{\mathrm{sel}}(\cdot\mid v_2)$ is the
fixed Gaussian $\phi(\cdot;\widetilde\mu_{n,1},\sigma_{j_o}^2)$ tilted by
the factor $\Lambda^{E_o}(\cdot,v_2)=\exp\{f(\cdot,v_2)\}$, whose first
three $u$-derivatives of $\log\Lambda^{E_o}$ are bounded by
$c_f^{(1)},c_f^{(2)},c_f^{(3)}$ via~\eqref{eq:f-bound} (the derivative
taken in the standardized coordinate $u$, so that the Jacobian of $r$ is
already absorbed into the $c_f^{(k)}$). Consequently
\begin{align*}
  \log p_{\mathrm{sel}}(u\mid v_2)
  =-\frac{(u-\widetilde\mu_{n,1})^2}{2\sigma_{j_o}^2}+f(u,v_2)-\log D(z)
\end{align*}
has $u$-derivatives bounded by $\sigma_{j_o}^{-1}+c_f^{(1)}$,
$\sigma_{j_o}^{-2}+c_f^{(2)}$, and $c_f^{(3)}$ at orders $1,2,3$; since
the leading $-\sigma_{j_o}^{-2}$ curvature dominates at large $|u|$,
$p_{\mathrm{sel}}(\cdot\mid v_2)$ has sup-norm and derivative sup-norms
bounded by constants depending only on $\sigma_{j_o}$ and
$c_f^{(1)},c_f^{(2)},c_f^{(3)}$, uniformly in $v_2$. The normalizer $D$
enters only to make $p_{\mathrm{sel}}$ a probability density and cancels
in every ratio, so no lower bound on $D$ is required. Hence
$|\partial_{z_1}^k\cP(z)|\le C_k''<\infty$ uniformly in $z$. Mixed
derivatives follow by combining the two arguments, giving
\begin{align}
  \label{eq:P-bound-proof}
  \sup_{z\in\real^p}\|\nabla^{(l)}\cP(z)\|_{\op}\le C^{(l)}<\infty,
  \qquad l=0,1,2,3.
\end{align}
Since $h\in\mathcal C^3_b(\real)$, the chain rule gives
$\|\nabla^{(l)}(h\circ\cP)(z)\|_{\op}\le\bar C^{(l)}<\infty$. Applying the
Leibniz rule to $g_2=(h\circ\cP)\,F$ and using~\eqref{eq:deriv-F} for
$j=1$,
\begin{align*}
  \|\nabla^{(l)}g_2(z)\|_{\op}
  &\le\sum_{k=0}^l\binom{l}{k}
     \|\nabla^{(k)}(h\circ\cP)(z)\|_{\op}\,
     \|\nabla^{(l-k)}F(z)\|_{\op}\\
  &\le\Big(\sum_{k=0}^l\binom{l}{k}\bar C^{(k)}\widetilde C_1^{(l-k)}\Big)F(z)
  =:\widetilde C_2^{(l)}F(z),
\end{align*}
which establishes~\eqref{eq:deriv-F} for $j=2$.

\medskip
\textit{Part (ii) (expected bounds at $\gamma_n$).}
By part~(i) it suffices to bound $\Ee[F(\gamma_n)]$. The mean-value
theorem and~\eqref{eq:f-bound} give
$f(\widetilde\Sigma^{1/2}\gamma_n+\widetilde\mu_n)
\le f(\widetilde\mu_n)+C_R\|\gamma_n\|$, so
\begin{align*}
  \Ee[F(\gamma_n)]
  \le\exp\{f(\widetilde\mu_n)\}\,\Ee\big[\exp\{C_R\|\gamma_n\|\}\big].
\end{align*}
By the uniform sub-Gaussianity of $\{Y_{i,n}\}$ in
Assumption~\ref{ass:linearizable} and independence, for every
$v\in\real^p$,
\begin{align*}
  \Ee\big[\exp\{v^\top\gamma_n\}\big]
  \le\exp\Big\{\frac{K_0^2}{2n}\sum_{i=1}^n
     \|\widetilde a_{i,n}\|_{\op}^2\|v\|^2\Big\}
  \le\exp\{K_0^2 p\,\|v\|^2/2\},
\end{align*}
using $\frac1n\sum_i\|\widetilde a_{i,n}\|_{\op}^2
\le\frac1n\sum_i\|\widetilde a_{i,n}\|_F^2=p$. Thus $\gamma_n$ is
sub-Gaussian uniformly in $n$, and a standard Chernoff bound gives
$\sup_n\Ee[\exp\{C_R\|\gamma_n\|\}]\le C<\infty$. Setting
$C_j^{(l)}=\widetilde C_j^{(l)}C$ yields~\eqref{eq:deriv-gamma}.
\end{proof}

\subsection{Results for Applications}

\subsubsection{Proof for A/B/n testing with binary outcomes}

For success probabilities, let $\mu_n=\sqrt {n} \pi$ where $\pi=c(\pi_1, \ldots, \pi_p)^\top$, $\Sigma=\operatorname{Diag}(\pi_j(1-\pi_j))_{j=1}^p$, and $T_{n}=\sqrt{n}\widehat\pi_n$ where $\widehat\pi_{n,j}=n^{-1}\sum_{i=1}^n O_{ij}$. 
For $i\in[n]$, denote $Y_{i,n}=\Sigma^{-1/2}(O_{i1}-\pi_1,\ldots,O_{ip}-\pi_p)^\top\in\mathbb R^d$, where $d=p$.
\begin{proposition}[ALR of A/B/n testing for success probability]\label{prop:ALR.ABn.prob}
Assume there exists $\varepsilon>0$ such that $\varepsilon \le \pi_j \le 1-\varepsilon$ for all $j\in[p]$.
Then
\begin{align*}
  \Sigma^{-1/2}(T_n - \mu_n)
  =
  \frac{1}{\sqrt{n}} \sum_{i=1}^n a_{i,n} Y_{i,n} + R_n,
\end{align*}
where (i) $a_{i,n}=I_p$, $R_n=0$, and 
\begin{align*}
  \frac{1}{n}\sum_{i=1}^n a_{i,n}a_{i,n}^{\top}=I_p,
  \quad
  \sup_{n}\max_{1\le i\le n}\|a_{i,n}\|_{\op} = 1 < \infty,
  \quad
  \sup_n\Ee\left[\exp\{4C_R\|R_n\|\}\right] = 1 <\infty,
\end{align*}
and (ii) $\{Y_{i,n}:1\le i\le n\}$ satisfies
\begin{align*}
  \Ee[Y_{i,n}]=0_d,
  \quad
  \Var(Y_{i,n})=I_d,
  \quad
  \sup_{n}\max_{1\le i\le n}
  \Ee[\exp\{s^\top Y_{i,n}\}]
  \le
  \exp\left\{\frac{\|s\|^2}{2\varepsilon(1-\varepsilon)}\right\}, ~
  \forall s\in\real^d.
\end{align*}
\end{proposition}

\begin{proof}\label{proof:prop:ALR.ABn.prob}    
Since $O_{ij}\sim \mathrm{Bernoulli}(\pi_j)$, it holds that $\Ee[O_{ij}-\pi_j] = 0$ and $\Var(O_{ij}-\pi_j)=\pi_j(1-\pi_j)$.
Therefore, $\Ee[Y_{i,n}]=0_d$ and $\Var(Y_{i,n})=I_d$.
In addition, as $\varepsilon \le \pi_j \le 1 - \varepsilon$, we have $|Y_{i,n,j}| \le 1/\sqrt{\varepsilon(1-\varepsilon)}$. 
Therefore, by Hoeffding's lemma, for any $s_j \in \real$,
\begin{align*}
    \Ee[\exp\{s_j Y_{i,n,j}\}]
    \le
    \exp\left\{
        \frac{s_j^2}{2\varepsilon(1-\varepsilon)}
    \right\}.
\end{align*}
Using the independence of $Y_{i,n,j}$ across $j$, for any $s\in\real^d$,
\begin{align*}
    \Ee[\exp\{s^\top Y_{i,n}\}]
    =
    \prod_{j=1}^p
    \Ee[\exp\{s_jY_{i,n,j}\}] 
    \le
    \exp\left\{
        \frac{\|s\|^2}{2\varepsilon(1-\varepsilon)}
    \right\},
\end{align*}
which completes the proof.
\end{proof}

For log-odds, for some $\varepsilon' > 0$, let
$T_{n,\varepsilon'}=\sqrt n\log\{\widehat\pi_{n,\varepsilon'}/(1-\widehat\pi_{n,\varepsilon'})\}$ where $\widehat\pi_{n,\varepsilon'} = \min \{1-\varepsilon', \max\{\varepsilon', \widehat\pi_{n}\}\}$,
$\mu_n=\sqrt n\log\{\pi/(1-\pi)\}$, and
$\Sigma=\operatorname{Diag}\left(1/\{\pi_j(1-\pi_j)\}\right)_{j=1}^p$. For $i\in[n]$, denote
$Y_{i,n}=\Sigma^{-1/2}\left(\frac{O_{i1}-\pi_1}{\pi_1(1-\pi_1)},\ldots,\frac{O_{ip}-\pi_p}{\pi_p(1-\pi_p)}\right)^\top\in\mathbb R^d$, where $d=p$.
\begin{proposition}[ALR of A/B/n testing for log-odds]\label{prop:ALR.ABn.odds}
Assume there exists $\varepsilon>0$ such that $\varepsilon \le \pi_j \le 1-\varepsilon$ for all $j\in[p]$. Then
\begin{align*}
  \Sigma^{-1/2}(T_{n, \varepsilon/2} - \mu_n)
  =
  \frac{1}{\sqrt{n}} \sum_{i=1}^n a_{i,n} Y_{i,n} + R_{n,\varepsilon/2},
\end{align*}
where (i) $a_{i,m}= I_p$ and $R_{n,\varepsilon/2} = o_p(1)$ and
\begin{align*}
  \frac{1}{n}\sum_{i=1}^n a_{i,n}a_{i,n}^{\top}=I_p,
  \quad
  \sup_{n}\max_{1\le i\le n}\|a_{i,n}\|_{\op} = 1 < \infty,
  \quad
  \sup_n\Ee\left[\exp\{4C_R\|R_{n,\varepsilon/2}\|\}\right] <\infty,
\end{align*}
and (ii) $\{Y_{i,n}:1\le i\le n\}$ satisfies
\begin{align*}
  \Ee[Y_{i,n}]=0_d,
  \quad
  \Var(Y_{i,n})=I_d,
  \quad
  \sup_{n}\max_{1\le i\le n}
  \Ee[\exp\{s^\top Y_{i,n}\}]
  \le
  \exp\left\{\frac{\|s\|^2}{2\varepsilon(1-\varepsilon)}\right\},
  \quad
  \forall s\in\real^d.
\end{align*}
\end{proposition}

\begin{proof}\label{proof:prop:ALR.ABn.odds}    
Similar to the proof of Proposition~\ref{proof:prop:ALR.ABn.prob}, we have $\Ee[Y_{i,n}]=0_d$, $\Var(Y_{i,n})=I_d$, and $|Y_{i,n,j}| \le 1/\sqrt{\varepsilon(1-\varepsilon)}$, $\Ee[\exp\{s^\top Y_{i,n}\}] \le \exp\left\{ \frac{\|s\|^2}{2\varepsilon(1-\varepsilon)} \right\}$ for any $s \in \real^d$.

Next, we analyze the reminder term $R_{n,\varepsilon/2}$. 
Define $\mathcal B_{n,\varepsilon/2} := \left\{\max_{j\in[p]}|\widehat\pi_{n,j}-\pi_j|\le \varepsilon/2
    \right\}$.
On $\mathcal B_{n,\varepsilon/2}$, we have
$\widehat\pi_{n,\varepsilon/2}=\widehat\pi_n$. Hence, by the coordinate-wise Taylor expansion of
$g(x)=\log\{x/(1-x)\}$ around $\pi_j$,
\begin{align*}
    \left\|
    \Sigma^{-1/2}(T_{n,\varepsilon/2}-\mu_n)
    -
    \frac{1}{\sqrt n}\sum_{i=1}^n a_{i,n}Y_{i,n}
    \right\|_2
    \le
    \|\Sigma^{-1/2}\|_{\op}\sqrt n  C
    \|\widehat\pi_{n,\varepsilon/2}-\pi\|_2^2  
    = o_p(1),
\end{align*}
where $C<\infty$ denotes an upper bound for $|g''(x)|$ on
$[\varepsilon/2,1-\varepsilon/2]$, and the last equality follows from
$\Ee[\sqrt n\|\widehat\pi_{n,\varepsilon/2}-\pi\|_2^2] < \Ee[\sqrt n\|\widehat\pi_{n}-\pi\|_2^2]=O(n^{-1/2})$.
On the other hand, by Hoeffding's inequality and a union bound,
\begin{align*}
\mathbb{P}(\mathcal B_{n,\varepsilon/2}^c) \le
    \sum_{j=1}^p
    \Pr\left(
        |\widehat\pi_{n,j}-\pi_j|>\varepsilon/2
    \right) \le
    2p\exp\{-n\varepsilon^2/2\}
    \to 0.
\end{align*}
Combining the two parts, we have $ R_{n,\varepsilon/2}=o_p(1)$.

In addition, as $\widehat\pi_{n,\varepsilon/2}/\pi$, $(1-\widehat\pi_{n,\varepsilon/2})/(1-\pi)>\varepsilon/2$, and $|\log(1+x) - x | \le \frac{x^2}{2 \min\{x+1,1\}}$ for $x > -1$, we have
\begin{align*}
    |R_{n, \varepsilon/2,j}| \le \sqrt{n} \left(\frac{(\widehat{\pi}_{n,\varepsilon/2,j} - \pi_j)^2}{2(\varepsilon/2)} + \frac{(1-\widehat{\pi}_{n,\varepsilon/2,j} -(1- \pi_j))^2}{2(\varepsilon/2)} \right) = \frac{2\sqrt{n} (\widehat{\pi}_{n,\varepsilon/2,j} - \pi_j)^2}{\varepsilon}.
\end{align*}
As $C_R = c_f^{(1)} \max\{\sigma_{j_o}, \|\Gamma_{j_o}^{1/2}\|_{\text{op}}\} \le \frac{c_f^{(1)}}{\varepsilon (1-\varepsilon)} < \infty$, by Hoeffding's lemma, 
\begin{align*}
    &\sup_n \Ee\!\left[\exp\{4C_R\|R_{n, \varepsilon/2}\|\}\right]
    \le  \sup_n \max_{j \in [p]} \Ee\!\left[\exp\{4pC_R|R_{n, \varepsilon/2,j}|\}\right]\\
    \le&  \sup_n \min\left\{\max_{j \in [p]} \int 2 \exp\{-2t^2\} \exp\{(8pC_R/\varepsilon) t^2 /\sqrt{n} \} dt,\exp\{8pC_R\sqrt{n}/\varepsilon\} \right\}
    <\infty.
\end{align*}

\end{proof}

\subsubsection{Proof for ranking via Bradley–Terry–Davidson models}

\begin{proposition}[ALR of the BT--Davidson MLE]
\label{prop:BTD_normal}
Assume $p_i^{(r)}\geq\varepsilon>0$, $\Sigma$ is positive definite, the log-likelihood~\eqref{eq:BTD_loglik} is uniformly strongly convex and has uniformly bounded third-order derivative, and $\theta$ is bounded.
Suppose $n_{jl}/n\to\lambda_{jl}>0$ as
$n\to\infty$. 
% Under $\sum_j\mu_j=0$ and standard regularity conditions,
There exist triangular array weights
$a_{i,n}\in\mathbb{R}^{(p-1)\times 2}$, $i\in[n]$, $n\geq 1$ such that
\begin{align}
  \Sigma^{-1/2}
  (T_n-\mu_n)
  \;=\;
  \frac{1}{\sqrt{n}}\sum_{i=1}^n a_{i,n}\,Y_{i,n}
  \;+\;o_p(1),
  \label{eq:BTD_ALR}
\end{align}
where (i) $R_n=o_p(1)$ satisfies 
\begin{align*}
  \frac{1}{n}\sum_{i=1}^n a_{i,n}a_{i,n}^{\top}=I_{p-1},
  \quad
  \sup_{n}\max_{1\le i\le n}\|a_{i,n}\|_{\op} = 1 < \infty,
  \quad
  \sup_n\Ee\left[\exp\{4C_R\|R_n\|\}\right] <\infty,
\end{align*}
and (ii) $\{Y_{i,n}:1\le i\le n\}$ satisfies
\begin{align*}
  \Ee[Y_{i,n}]=0_d,
  \quad
  \Var(Y_{i,n})=I_d,
  \quad
  \sup_{n}\max_{1\le i\le n}
  \Ee[\exp\{s^\top Y_{i,n}\}]
  \le
  \exp\left\{C\|s\|^2\right\},
  \quad
  \forall s\in\real^d,
\end{align*}
for some $C > 0$ and $d = 2$.

\end{proposition}

\begin{proof}

Let $S_{\theta,i,n}$ and $S_{\eta,i,n}$ denote the score regarding $\theta$, $\eta$ of match $i$, respectively, and define $\widetilde S_{i,n}:=S_{\theta,i,n}-I_{\theta\eta,n}I_{\eta\eta,n}^{-1}S_{\eta,i,n}$. Choose $L_{i,n}\in\mathbb R^{3\times2}$ such that $L_{i,n}L_{i,n}^\top=\operatorname{Diag}(p_i)-p_ip_i^\top$. 
For $i=(i_1,i_2)$ define 
$B_{i,n}:=(e_{i_1}, e_{i_2},
(e_{i_1}+e_{i_2})/2)-I_{\theta\eta}I_{\eta\eta}^{-1}(0,0,1)L_{i,n}$,
where $e_j$ is the $j$-th standard basis vector.
Then let $Y_{i,n} := B_{i,n}^{-1}\widetilde S_{i,n}$, where
$B_{i,n}^{-1}$ is the pseudo-inverse of $B_{i,n}$. 
Then $Y_{i,n}$ satisfies $\mathbb E[Y_{i,n}]=0$, and $\operatorname{Var}(Y_{i,n})=I_2$. 
Finally, we set $a_{i,n}:=\Sigma^{-1/2}B_{i,n}$, and then
$\frac{1}{n}\sum_{i=1}^n a_{i,n}a_{i,n}^\top=I_{p-1}$.

By the asymptotic linear expansion of the constrained maximum likelihood estimator,
\begin{align*}    
T_n - \mu_n
=\mathcal I_{\theta\mid\eta,n}^{-1}
\frac1{\sqrt n}\sum_{i=1}^n\widetilde S_{i,n}+o_p(1).
\end{align*}
Multiplying both sides by
$\Sigma^{-1/2}$
and by the convergence
\(\Sigma_n \to \Sigma \), we obtain
\[\Sigma^{-1/2}
(T_n - \mu_n)
=
\frac1{\sqrt n}\sum_{i=1}^n a_{i,n}Y_{i,n}
+
o_p(1).
\]

% Moreover, uniform positivity implies $\sup_{i,n}|L_{i,n}^{+}|{\op}<\infty$. 
Since the centered multinomial outcomes are bounded, then $|Y_{i,n}|\leq C$ uniformly.
Hence, by Hoeffding's lemma,
\begin{align*}
\sup_n\max_{1\leq i\leq n}
\Ee\left[\exp\{s^\top Y_{i,n}\}\right]
&\leq
\exp\left\{\frac{C^2}{2}\|s\|^2\right\},
\quad s\in\mathbb R^2.
\end{align*}

Define $Z_n:=\frac1{\sqrt n}\sum_{i=1}^n a_{i,n}Y_{i,n}$, $
U_n:=\frac1{\sqrt n}\sum_{i=1}^n
(S_{\theta,i,n}^\top,S_{\eta,i,n})^\top$. Let
$\delta_n:=\|\Sigma^{-1/2}\Sigma_n^{-1/2}-I_{p-1}\|_{\op}\to0$, where $\Sigma_n$ is the empirical estimate of $\Sigma$.
By the uniform strong
concavity and bounded third derivatives of the constrained likelihood, we have
\begin{align*}    
\|R_n\|\leq\delta_n\|Z_n\|+\frac{K}{\sqrt n}\|U_n\|^2,
\end{align*}
and $R_n=o_p(1)$.
Since $Z_n$ and $U_n$ are uniformly sub-Gaussian, by Cauchy's inequality, 
$\sup_n\Ee[\exp\{4C_R\|R_n\|\}]<\infty$.
\end{proof}

\subsubsection{Proof for non-parametric feature importance}

\begin{proposition}[ALR of the feature-importance estimator]
\label{prop:VIM_normal}
% Let $\widehat{\mu}_n=(\widehat{\mu}_{1,n},\ldots,\widehat{\mu}_{p,n})$ be the variable-importance estimator based on nuisance estimates obtained from an independent dataset. 
% \textcolor{red}{Since $a_{i,m}=I_p$, can we simply omit it from the representation below?}
% \zg{I was thinking of keeping the same notations as in  Assumption~\ref{ass:linearizable}. But happy to drop it if that feels redundant :)}
Assume $\Var(O_i) > \varepsilon$ for some $\varepsilon > 0$, $\Sigma$ is positive definite, $|O_i|$, $|m(x)|$, $|m_{-j}(x_{-j})|$, $|\widehat{m}(x)|$, $|\widehat{m}_{-j}(x_{-j})|$ are uniformly bounded,
% which implies $\psi_{j}(X,0, \neq 0$; 
% $\psi(X,O; \widehat{\eta})$ belongs to a $\mathbb{P}$-Donsker class with probability tending to one, 
and the nuisance estimators satisfy $\|\widehat m - m\|_2 = o_P(n^{-1/4})$ and $\|\widehat m_{-j} - m_{-j}\|_2 = o_P(n^{-1/4})$ for all $j \in [p]$. Then
\begin{align}
\Sigma^{-1/2}(T_{n,\varepsilon/2}-\mu_n)
&=
\frac{1}{\sqrt{n}}\sum_{i=1}^n a_{i,n} Y_{i,n}
+ R_n,
\label{eq:VIM_ALR}
\end{align}
where (i) $a_{i,m}=I_p$ and $R_n=o_p(1)$ satisfy
\begin{align*}
  \frac{1}{n}\sum_{i=1}^n a_{i,n}a_{i,n}^{\top}=I_p,
  \quad
  \sup_{n}\max_{1\le i\le n}\|a_{i,n}\|_{\op} = 1 < \infty,
  \quad
  \sup_n\Ee\left[\exp\{4C_R\|R_n\|\}\right] <\infty,
\end{align*}
and (ii) $\{Y_{i,n}:1\le i\le n\}$ satisfies
\begin{align*}
  \Ee[Y_{i,n}]=0_d,
  \quad
  \Var(Y_{i,n})=I_d,
  \quad
  \sup_{n}\max_{1\le i\le n}
  \Ee[\exp\{s^\top Y_{i,n}\}]
  \le
  \exp\left\{C\|s\|^2\right\},
  \quad
  \forall s\in\real^d,
\end{align*}
for some $C > 0$.
\end{proposition}

\begin{proof}\label{proof:prop:VIM_normal}

By the definition of $Y_{i,n}$, we have $ \Ee[Y_{i,n}]=0_d$, $\Var(Y_{i,n})=I_d$. 
As $\Var(O_i) > 0$, $\Sigma$ is positive definite, $O_i$, $m(x)$, $m_{-j}(x_{-j})$, $\widehat{m}(x)$, $\widehat{m}_{-j}(x_{-j})$ are uniformly bounded, there exists $C' > 0$ such that $|Y_{i,n,j}| \le C'$.
Analogous to the proof of Proposition~\ref{prop:ALR.ABn.prob}, we obtain $\Ee[\exp\{s^\top Y_{i,n}\}] \le \exp\left\{{C'}^2\|s\|^2/2\right\}$ for any $s\in\real^d$.

Next, we analyze $R_n$. According to Theorem 1 in \cite{williamson2023general}, under the assumptions in Proposition~\ref{prop:VIM_normal},it is shown that $R_n = o_p(1)$.
It remains to show $\sup_n \Ee\left[\exp\{4C_R\|R_n\|\}\right]$ is finite.
In fact, by the definition of $R_n$ and the triangle inequality, it is sufficient to prove
\begin{align*}
    \sup_n
    \Ee\left[
        \exp\left\{
           C\left\|\Sigma^{-1/2}(T_{n,\varepsilon/2}-\mu_n)\right\|
        \right\}
    \right], \quad 
    \sup_n
    \Ee\left[
        \exp\left\{
            C\left\|
            \frac{1}{\sqrt n}\sum_{i=1}^n a_{i,n}Y_{i,n}
            \right\|
        \right\}
    \right]
    <\infty,
\end{align*}
for any $C > 0$. 
We prove the first inequality; the second inequality follows similarly. 

For $T_{n,\varepsilon/2}-\mu_n$, as shown in the proof of Theorem~1 of \citet{williamson2023general}, $T_{n,\varepsilon/2}$ is consistent and the bias satisfies $\left\|\Ee[T_{n,\varepsilon/2}-\mu_n]\right\|=o(1)$.
In addition, since the variance estimator in the denominator of $T_{n,\varepsilon/2}$ is clipped at $\varepsilon/2$ and the outcome, nuisance functions are uniformly bounded,
$T_{n,\varepsilon/2}-\Ee[T_{n,\varepsilon/2}]$ is sub-Gaussian with a constant independent of $n$. 
Finally, as $\Sigma$ is positive definite, then $\Ee\left[ \exp\left\{C\left\|\Sigma^{-1/2}(T_{n,\varepsilon/2}-\mu_n) \right\| \right\} \right]$ is uniformly bounded over $n$.

\end{proof}

\begin{algorithm}[tbp]
\caption{Cross-fitted Feature Importance Estimation}
\label{alg:feature.importance}
\begin{algorithmic}[1]
\Require Data $\cD=\{(X_i,O_i)\}_{i=1}^n$, candidate set $\cS\subseteq[p]$,         folds $L\ge2$.
\State Split $\cD$ into $L$ equal folds; let $\cD_l$ denote fold $l$.
\For{$l=1,\ldots,L$}
  \State On $\cD\setminus\cD_l$, obtain $\widehat m_{-j}^{(-l)}$ for all 
         $j\in\cS$ and $\widehat m^{(-l)}$, denoted collectively by $\widehat{\eta}^{(-l)}$.
  \State For each $j\in\cS$, compute on fold $l$
  \[
  \widehat\psi_{n,j}^{(l)}
  \leftarrow
  \frac{1}{ |\cD_l|}\sum_{i\in\cD_l}
  \frac{
  \bigl(\widehat m^{(-l)}(X_i)-\widehat m_{-j}^{(-l)}(X_{i})\bigr)^2
  +
  2\bigl(O_i-\widehat m^{(-l)}(X_i)\bigr)
  \bigl(\widehat m^{(-l)}(X_i)-\widehat m_{-j}^{(-l)}(X_{i})\bigr)
  }{
  |\cD_l|^{-1}\sum_{i\in\cD_l}(O_i-\bar O_l)^2
  },
  \]
  where $\bar O_l$ denotes the average of $O_i$ on $\cD_l$. 
\EndFor
\State Set $\widehat\psi_{n,j}=L^{-1}\sum_{l=1}^L \widehat\psi_{n,j}^{(l)}$ for $j\in\cS$, and compute the estimated covariance matrix $\widehat\Sigma$.
% to add: variance formula
\State \Return $\widehat\psi_n$, $\widehat\Sigma$.
\end{algorithmic}
\end{algorithm}

\section{Additional simulation details}
\label{app:sim-details}
This appendix describes the four different data-generating processes, the common procedure used to define the signal regimes, and the implementation settings for the simulation designs in Section~\ref{sec:simulations}, together with the baseline Gaussian design and the conditional-coverage computation underlying the illustrative example in Figure~\ref{fig:example1}. Throughout, coordinates are indexed by $[p]=\{1,\dots,p\}$.
\paragraph{Data-generating processes.}
In every design the selection statistic is the standardized estimator
$T_n=\sqrt n\,\widehat\mu_n$ of a target $\mu_n=\sqrt n\,\mu$, with
sampling covariance $\Sigma=\operatorname{Cov}(T_n)$; signal is calibrated by
adding a block shift $b$ to the true top-$k$ coordinates of $\mu$ and
solving $\Delta_{\mathrm{std}}(b)=0.3$ (\emph{Weak Separation}) and
$\Delta_{\mathrm{std}}(b)=2.0$ (\emph{Strong Separation}).
\begin{enumerate}[leftmargin=*]
\item \textbf{Gaussian mean} ($p=20$, $k=3$, $n=50$; Figures~\ref{fig:example1} and~\ref{fig:example2}).
We draw $X_i\sim\mathcal N_p(\mu,\Sigma)$ and take $T_n=\sqrt n\,\bar X$, which is
exactly normal with known $\Sigma$. The baseline is
$z=(0.45,\,0.30,\,0,\,0.25,\,0.15\cdot\mathbf 1_4,\,-0.20\cdot\mathbf 1_{12})$ and
$\Sigma=\operatorname{Diag}(1,1,10,2,2,2,2,2,\mathbf 1_{12})$, so coordinate $3$ is
a high-variance decoy ($\sigma_3^2=10$) lying outside the true top-$k$ set $\{1,2,4\}$.
Figure~\ref{fig:example2} reports the two separation regimes $\mu=z+b\,\mathbf 1_{\{1,2,4\}}$,
for which calibration gives $b\approx-0.015$ (Weak) and $b\approx+0.466$ (Strong).
Figure~\ref{fig:example1} instead fixes the baseline configuration $\mu=z$ and adds a fourth panel reporting
\emph{conditional coverage for the high-variance decoy}. Its inflated variance lets
coordinate $3$ enter the selected set purely by chance despite $\mu_3=0$, making it the
coordinate most exposed to the winner's curse. Writing $\Ehat$ for the selected top-$k$
set and $C_3(T)$ for a method's interval at coordinate $3$, this panel estimates the
coverage of $\mu_3$ conditional on the decoy being selected,
\[
  \Pp\bigl[\mu_3\in C_3(T)\,\big|\,3\in\Ehat\bigr],
\]
computed as the fraction of replications in which $C_3(T)$ covers $\mu_3$, among those
replications in which $3\in\Ehat$. The event $\{3\in\Ehat\}$ aggregates over all selected
sets containing the decoy and is implied by the per-winner guarantee of
Section~\ref{sec:guarantees}; therefore, a conditionally valid method retains nominal coverage
here while a merely marginal method need not.
% \item \textbf{Gaussian mean} ($p=20$, $k=3$, $n=50$; Figure~\ref{fig:example2}).
% We draw $X_i\sim\mathcal N_p(\mu,\Sigma)$ and take $T_n=\sqrt n\,\bar X$, which is
% exactly normal with known $\Sigma$. The baseline is
% $z=(0.45,\,0.30,\,0,\,0.25,\,0.15\cdot\mathbf 1_4,\,-0.20\cdot\mathbf 1_{12})$ and
% $\Sigma=\operatorname{Diag}(1,1,10,2,2,2,2,2,\mathbf 1_{12})$, so coordinate $2$ is
% a high-variance decoy ($\sigma^2=10$) outside the true top-$k$ set $\{0,1,3\}$;
% $\mu=z+b\,\mathbf 1_{\{0,1,3\}}$. Calibration gives $b\approx-0.015$ (Weak) and
% $b\approx+0.466$ (Strong).

% \item \textbf{Binomial dosage} ($p=10$, $k=3$, $n=100$; Figure~\ref{fig:example3}).
% For each level $j$, $O_{ij}\sim\mathrm{Bernoulli}(\pi_j)$ with
% $\pi_j=\operatorname{logit}^{-1}(\eta_0+z_j+b\,\mathbf 1\{j\le k\})$, $$\eta_0=-1.2 \ \text{ and }
% \ z=(0.60,\,0.50,\,0.40,\,0.30,\,0.25\cdot\mathbf 1_6).$$ We take
% $T_n=\sqrt n\,\widehat\pi_n$ with plug-in covariance
% $\widehat\Sigma=\operatorname{Diag}\{\widehat\pi_j(1-\widehat\pi_j)\}_{j=1}^p$;
% the true top-$k$ set is $\{0,1,2\}$. Calibration gives $b\approx-0.007$ (Weak)
% and $b\approx+0.490$ (Strong).

%, whileselection uses the \emph{log-success} score$s_E(T)=\sum_{j\in E} T_j$; the two induce the same ordering

\item \textbf{Binomial dosage} ($p=10$, $k=3$, $n=100$; Figure~\ref{fig:example3}).
For each level $j$, $O_{ij}\sim\mathrm{Bernoulli}(\pi_j)$ with $z=(0.60,\,0.50,\,0.40,\,0.30,\,0.25\cdot\mathbf 1_6)$,
$\pi_j=\operatorname{logit}^{-1}(\eta_0+z_j+b\,\mathbf 1\{j\le k\})$, and  $\eta_0=-1.2$. We take
$T_n=\sqrt n\,\widehat\pi_n$ with plug-in covariance
$\widehat\Sigma=\operatorname{Diag}\{\widehat\pi_j(1-\widehat\pi_j)\}_{j=1}^p$. Signal is
calibrated on the selection scale, i.e.\ the standardized gap of $\pi$ with
delta-method standard error $\sqrt{(1-\pi_j)/(n\pi_j)}$, giving true top-$k$ set
$\{1,2,3\}$ and $b\approx-0.007$ (Weak), $b\approx+0.506$ (Strong).
\item \textbf{Bradley--Terry--Davidson} ($p=10$, $k=3$; Figure~\ref{fig:example4}).
Log-strengths have baseline
$z=(0.60,\,0.50,\,0.40,\,0.30,\,0.25,\,0.24,\,0.23,\,0.22,\,0.21,\,0.20)$; the
ability vector is $\theta=\mathrm{center}(z+b\,\mathbf 1\{j\le k\})$ with
$\sum_j\theta_j=0$ and true top-$k$ set $\{1,2,3\}$. Under a round-robin design
every pair plays $n_{\mathrm{matches}}=10$ games with tie parameter $\nu=0.2$; we
take $T_n=\sqrt n\,\widehat\theta_n$ (constrained MLE), with $\Sigma$ the inverse
Schur complement of the limiting Fisher information. Calibration gives
$b\approx-0.007$ (Weak) and $b\approx+0.538$ (Strong).
\item \textbf{Nonparametric feature importance} ($p=10$, $k=3$, $n=500$;
Figure~\ref{fig:example5}). Covariates are drawn uniformly on $[-1,1]^p$ and
$O_i=f(X_i)+\varepsilon_i$, $\varepsilon_i\sim\mathcal N(0,1)$, with the smooth
additive $f(x)=\sum_{j=1}^p\beta_j\,g_j(x_j)$ and orthonormal bases
\[
g_j(x)=
\begin{cases}
\sqrt3\,x, & j\equiv1\pmod5,\\[0.4ex]
\tfrac{\sqrt5}{2}(3x^2-1), & j\equiv2\pmod5,\\[0.4ex]
\tfrac{\sqrt7}{2}(5x^3-3x), & j\equiv3\pmod5,\\[0.4ex]
\sqrt2\sin(\pi x), & j\equiv4\pmod5,\\[0.4ex]
\sqrt2\sin(2\pi x), & j\equiv0\pmod5.
\end{cases}
\]
The target is the Williamson feature-importance $\psi_j$; $T_n=\sqrt n\,\widehat\psi_n$
is the one-step estimator with nuisances fit by a spline-GAM learner on an
independent holdout, and $\Sigma=\operatorname{Diag}(\operatorname{Var}(\psi_{ij}))$.
The baseline is $\beta=(0.65,\,0.60,\,0.55\cdot\mathbf 1_8)$ with true top-$k$ set
$\{1,2,3\}$ and $\beta\gets\beta+b\,\mathbf 1\{j\le k\}$. Calibration gives
$b\approx0.034$ (Weak) and $b\approx0.142$ (Strong).
\end{enumerate}

\paragraph{Calibration and implementation details. }
Across all designs the amount of signal is quantified by the standardized boundary gap between the $k$-th and $(k{+}1)$-th largest population scores,
\[
  \Delta_{\mathrm{std}}
  =\frac{\mu_{(k)}-\mu_{(k+1)}}{\operatorname{se}\bigl(T_{(k)}-T_{(k+1)}\bigr)},
  \qquad
  \operatorname{se}\bigl(T_{(k)}-T_{(k+1)}\bigr)=\sqrt{\Sigma_{jj}+\Sigma_{ll}-2\,\Sigma_{jl}},
\]
where $\mu_{(1)}\ge\cdots\ge\mu_{(p)}$ are the ordered population scores, $\{j,l\}$ index the coordinates attaining $\mu_{(k)},\mu_{(k+1)}$, respectively, and $\Sigma=\operatorname{Cov}(T)$ is the sampling covariance matrix of the selection statistic evaluated at the true parameter values.
For each of the four data-generating processes, we add a common shift $b$ to the true top-$k$ coordinates on the natural parameter scale---means, log-odds, log-strengths, and $\beta$-coefficients, respectively---and use bisection to calibrate $\Delta_{\mathrm{std}}(b)$ to design-specific target values representing the \emph{Weak separation} and \emph{Strong separation} regimes. Thus, within each design, the two regimes differ only in the separation of the top-$k$ block.
For every design, we use $k=3$, nominal level $1-\alpha=0.95$, the additive score $s_E(T)=\sum_{j\in E}T_j$, and the tuning-free temperature $\tau=\snoise(T)(\log|\Ek|)^{-1}$ corresponding to the regret budget $q=1$, as per Corollary~\ref{cor: temp par}. We evaluate conditional densities and selective pivots on a grid of size $500$, increasing the grid size to $2000$ for the Bradley--Terry--Davidson simulation design, as its log-strength targets are more sensitive to grid resolution. We use $B=500$ Monte Carlo replications for the Gaussian, Binomial, and feature-importance designs and $B=200$ for the Bradley--Terry--Davidson design.

\end{appendices}
\end{document}